\documentclass[12pt]{article}%
\usepackage{enumitem}
\usepackage{adjustbox}
\usepackage{bm}
\usepackage{amsfonts}
\usepackage[colorlinks,citecolor=blue,urlcolor=blue,hyperfootnotes=false]{hyperref}
\usepackage{amssymb}
\usepackage[bottom]{footmisc}
\usepackage[hyperref]{ntheorem}
\usepackage{graphicx}
\usepackage{float}
\usepackage[figuresright]{rotating}
\usepackage{natbib}  
\usepackage[doublespacing,nodisplayskipstretch]{setspace}
\usepackage[margin=1in, a4paper]{geometry}
\usepackage{booktabs}
\usepackage{threeparttable}
\usepackage{amsmath}
\usepackage{lscape}
\usepackage{scrextend}
\usepackage{relsize}
\usepackage{multicol}
\usepackage{chngcntr}
\usepackage{etoolbox}
\usepackage{amssymb}
\usepackage{tabularx}
\usepackage[english]{babel}
\usepackage{multirow}
\usepackage{booktabs}
\usepackage[labelfont=bf]{caption}
\usepackage{float}
\usepackage{titlesec}
\usepackage{array}
\usepackage{color}
\usepackage{afterpage}
\usepackage{framed}
\usepackage{float}
\usepackage{afterpage}%
\usepackage{appendix}
\usepackage{subcaption}
\providecommand{\U}[1]{\protect\rule{.1in}{.1in}}
\allowdisplaybreaks

\theoremstyle{plain}
\newcommand{\ii}{{\rm i}}
\newcommand{\ee}{{\rm e}}
\newtheorem{theorem}{Theorem}[section]
\newtheorem{Assumption}{Assumption}[section]

\newtheorem{lemma}{Lemma}[section]

\newenvironment{proof}[1][Proof]{\noindent\textbf{#1.} }{\ \rule{0.5em}{0.5em}}

\numberwithin{equation}{section}
\numberwithin{table}{section}
\numberwithin{figure}{section}

\theorembodyfont{\normalfont}

\deffootnote{1em}{1.6em}{\thefootnotemark \enskip}

\numberwithin{equation}{section}
\numberwithin{Assumption}{section}

\definecolor{lightgray}{gray}{0.9}
\begin{document}

%Testing heteroscedasticity in nonparametric regression

%\title{Tests for heteroskedasticity with measurement errors}
\title{Testing Heteroskedasticity Under Measurement Error}
\author{Xiaojun Song\thanks{Corresponding author. Guanghua School of Management, Peking University, Email: \texttt{sxj@gsm.pku.edu.cn}. This work was supported by the National Natural Science Foundation of China [Grant Numbers 72373007 and 72333001]. The author also gratefully acknowledges the research support from the Center for Statistical Science of Peking University, and the Key Laboratory of Mathematical Economics and Quantitative Finance (Peking University) of the Ministry of Education, China.}\\
	\and Jichao Yuan\thanks{Guanghua School of Management, Peking University, Email: \texttt{2201111030@stu.pku.edu.cn}.}
	}

\maketitle
\begin{abstract}
In this paper, we propose a novel approach to detect heteroskedasticity in regression models with regressors contaminated by measurement error. Specifically, inspired by the integrated conditional moment (ICM) approach, we construct test statistics based on a deconvolved residual-marked empirical process and establish their asymptotic properties in both ordinary smooth and supersmooth cases, assuming the measurement error distribution is known. The issue of an unknown measurement error distribution is addressed by employing estimators of the measurement error characteristic function based on repeated measurements. Furthermore, depending on whether the measurement error distribution is known or not, to obtain critical values from the case-dependent limiting null distributions, we propose two computationally attractive multiplier bootstrap methods where the %well-known 
``parameter estimation effect'' is successfully addressed. %eliminated by an orthogonal projection onto the tangent space of the nuisance parameter. 
Finally, simulation results and empirical studies about corn yields and household budget shares confirm the favorable properties of the proposed tests.
\end{abstract}

\newpage
\doublespacing

\section{Introduction}\label{sec.Intro}
As \cite{white1980heteroskedasticity} famously observed, \textquotedblleft the presence of heteroskedasticity in the disturbances of an otherwise properly specified linear model leads to consistent but inefficient parameter estimates and inconsistent covariance matrix estimates \textquotedblright{}, testing whether the error terms in regression models satisfy the homoskedasticity assumption is of fundamental importance, in order to avoid invalid hypothesis testing and a substantial deterioration in the predictive accuracy of the model. 

An extensive literature has emerged on testing for heteroskedasticity across a wide range of regression frameworks over the past half-century. Work includes \cite{bickel1978using}, \cite{breusch1979simple}, \cite{white1980heteroskedasticity}, \cite{cook1983diagnostics}, \cite{dette1998testing}, \cite{zhu2001heteroscedasticity}, \cite{dette2002consistent}, \cite{zheng2009testing}, \cite{su2013nonparametric}, \cite{guo2020pairwise}, \cite{tan2021testing} and \cite{xu2021distance}. Despite the maturity of this line of research, almost all existing procedures rest on a strong maintained assumption that all random variables are observed without error. This no–measurement-error assumption is rarely credible in empirical work, and the literature’s longstanding silence on it is largely driven by the severe technical difficulties. 

A growing literature has examined inference and testing problems in regression models with measurement error. Early work includes \cite{carroll2006measurement} and \cite{hu2012estimation}, which primarily focus on identification and estimation. Building on these developments, more recent work further extends these ideas: \cite{dong2021average} study the density-weighted average derivative estimator in measurement error models; \cite{otsu2021specification} develop smoothing-based specification tests, and \cite{dong2022estimation} develop semiparametric estimation for varying coefficient models with mismeasured regressors; \cite{dong2022nonparametric} propose nonparametric significance tests based on deconvolution estimators. More fundamentally, for estimating characteristic functions under unknown measurement error distributions, \cite{delaigle2008deconvolution} and \cite{kurisu2022uniform} provide nonparametric identification and uniform convergence results, respectively. However, in this setting, heteroskedasticity testing becomes more complex, as the researcher must consistently estimate both the underlying regression function and the disturbance variance structure, potentially without knowledge of the measurement error distribution.

Unfortunately, if researchers ignore the noise in the data and mechanically apply conventional procedures to test for heteroskedasticity, they can easily be led to misleading conclusions. From a theoretical perspective, the seminal analysis of heteroskedasticity tests in linear errors-in-variables models in \cite{wooldridge1996solutions} shows that, once the regressors are measured with error, the asymptotic behavior of the usual statistics based on auxiliary regressions of squared residuals no longer follows from the classical framework and has to be re-derived under additional regularity conditions. On the applied side, \cite{alica2025comparison}, by comparing the simulation performance of traditional tests (such as those discussed in \cite{goldfeld1965some}, \cite{park1966estimation}, \cite{glejser1969new}, \cite{harvey1976estimating}, \cite{breusch1979simple} and \cite{white1980heteroskedasticity}) when the data are contaminated by measurement error, demonstrates that ignoring measurement error can cause the empirical size of these tests to deviate substantially from the nominal level and their power to deteriorate dramatically.

We next review the limited literature addressing tests for heteroskedasticity in the presence of measurement error. In an early contribution, \cite{carroll1992diagnostics} proposed corrected residual plots that adjust for measurement error, together with associated formal tests. \cite{wooldridge1996solutions} provided an asymptotic analysis of familiar Lagrange multiplier–type tests for homoskedasticity when the regressors are subject to classical measurement error. \cite{wallentin2002test} showed, by simulation, that least-squares residuals may still be informative for detecting heteroskedasticity even when parameters are estimated by instrumental variables, whereas naive implementations of White’s test based on instrumental-variables residuals and observed regressors can exhibit severely inflated rejection rates. Most recently, \cite{romeo2024detecting} developed a measurement-error–adapted version of White’s test, implemented via a model-based bootstrap. However, these contributions are all derived in simple linear regression settings that are difficult to satisfy in many empirical applications involving multiple covariates and more flexible parametric specifications.

The primary focus of our analysis is to extend heteroskedasticity tests for measurement error to parametric regression models that are more representative of empirical practice. Specifically, we develop integrated conditional moment (ICM)–type tests based on a residual-marked empirical process in the spirit of \cite{bierens1982consistent} and \cite{bierens1990consistent}. The presence of measurement error in the regressors is handled by employing deconvolution kernel estimators, as motivated by \cite{carroll1988optimal} and \cite{stefanski1990deconvolving}, to construct residuals. Critical values are then obtained through a computationally effective multiplier bootstrap. In implementing this bootstrap, we project the empirical process orthogonally onto the tangent space of the nuisance parameter to remove the \textquotedblleft parametric estimation effect\textquotedblright{}, which is initially discussed in \cite{durbin1973distribution} and has long been regarded as an unavoidable and notoriously difficult obstacle to the practical implementation of ICM-type tests, thereby ensuring the validity of the bootstrap approximation.

It is important to highlight our contribution to addressing the long-standing difficulty of jointly recovering the regression relationship and the disturbance variance structure in the presence of measurement error. In particular, the resulting test statistics enjoy a parametric convergence rate, the associated bootstrap implementation is fast and conceptually transparent, and the procedure is robust with respect to tuning parameters, making it a promising tool for other testing problems beyond the specific setting considered here. In addition, we develop a companion testing procedure for the empirically realistic case in which the distribution of the measurement error is unknown, and we provide both asymptotic theory and empirical evidence for its performance. 

The rest of the paper is organized as follows. In Section \ref{sec.Test}, we outline the testing
framework and the construction of our statistics. Then the asymptotic properties of the test statistics with some reasonable assumption are discussed under the null, local alternatives, and global alternative in Section \ref{sec.Asy}. The cases of unknown measurement error distribution are addressed in Section \ref{sec.Unknown} and the analysis about asymptotic behaviors of the proposed test statistics in Section \ref{sec.Unknown} parallels that in Section \ref{sec.Asy} and differ only in the strength of assumptions. In Section \ref{sec.boot}, we detail the implementation of the projection-based multiplier bootstrap procedure, emphasizing that the customized, explicitly constructed projection is particularly appealing due to its ease of interpretation and computational convenience. Results of Monte Carlo simulations and empirical studies are presented in Section \ref{sec.Simulation} and \ref{sec.Examples}, respectively. Additional simulation results and proofs are provided in the online supplementary appendix.

\section{Testing Procedure}\label{sec.Test}
Let $Y$ be a %scalar 
response variable and $X$ be the scalar unobservable %error-free 
explanatory variable. Suppose that we consider a potentially heteroskedastic parametric regression model
\begin{align}\label{var.Y}
    Y = g(X;\theta_0)+U,
\end{align}
where $g(\cdot;\theta_0)$ is the conditional mean function of a known parametric form characterized by the unknown parameter $\theta_0\in\Theta\subset\mathbb{R}^p$ and $U$ is the error term, which is the part of $Y$ not explained by $X$. In addition, the condition $\mathbb{E}[U\vert X]=0$ almost surely (a.s.) is imposed to ensure the correct specification of the conditional mean function $g(X;\theta_0)$. Moreover, since the true regressor $X$ is not directly observed, we instead observe a scalar variable $W$ through an additive measurement error model
\begin{align}\label{var.W}
    W=X+\epsilon,\quad \mathbb{E}\left[\epsilon\right]=0,
\end{align}
where $\epsilon$ is assumed to follow the classical measurement error assumption in the sense that $\epsilon$ is independent of $X$. The density of $\epsilon$, $f_\epsilon$, is assumed to be known here and in Section \ref{sec.Asy}, while the case of an unknown density $f_\epsilon$ is addressed in Section \ref{sec.Unknown}. We are interested in testing whether the conditional variance function of $U$ given $X$ is equal to an unknown positive constant. That is, our null hypothesis of interest is
\begin{align}\label{hyp.H0}
    H_0:\mathbb{E}\left[U^2\vert X\right]= \sigma^2_0 \quad \text{a.s. for some }0<\sigma_0^2<\infty,
\end{align}
and the alternative hypothesis is 
\begin{align}\label{hyp.H1}
    H_1:\Pr\{\mathbb{E}\left[U^2\vert X\right]\neq \sigma^2\}>0 \quad \text{for all }0<\sigma^2<\infty,
\end{align}
which is the negation of $H_0$. Note that $\sigma^2_0$ is also the unconditional variance of $U$ when $H_0$ holds. %$\sigma^2_0=\mathbb{E}[U^2]$. 

Following \cite{bierens1982consistent}, \cite{bierens1990consistent}, and \cite{bierens1997asymptotic}, the conditional moment restriction in \eqref{hyp.H0} that characterizes $H_0$ can be equivalently expressed as a continuum number of unconditional moment restrictions as follows:
\begin{align}\label{stat.Population}
    S(\xi,\theta_0,\sigma^2_0) = \mathbb{E}\left\{\left[\left(Y-g(X;\theta_0)\right)^2-\sigma_0^2\right]\ee^{\ii X\xi}\right\}=0 \quad \text{for all } \xi\in\Pi,
\end{align}
%for some $\theta_0\in\Theta\subset\mathbb{R}^p\,,\sigma_0\in\mathbb{R}$, 
where $\ii=\sqrt{-1}$ denotes the imaginary unit, $\Pi$ is a properly chosen compact set with nonempty interior, and the weighting function $\ee^{\ii X\xi}$ is selected from a parametric family indexed by $\xi$. In the absence of measurement error, the parameters $\theta_0$ and $\sigma^2_0$ are estimable via parametric methods, enabling the construction of test statistics based on a sample analog of equation \eqref{stat.Population}. However, in the presence of measurement error, the explanatory variable $X$ is unobservable and only the contaminated version $W$ is available. Motivated by \cite{dong2022nonparametric}, equation \eqref{stat.Population} can be reformulated in terms of the joint density of $(Y,X)$, which is denoted by $f_{Y,X}(y,x)$, as detailed below
\begin{align}\label{stat.Integral}
    S(\xi,\theta_0,\sigma^2_0) = \iint \left[\left(y-g(x;\theta_0)\right)^2-\sigma_0^2\right]f_{Y,X}(y,x)\ee^{\ii x\xi}\,dy\,dx.
\end{align}
Given a random sample $\{(Y_i,W_i)'\}_{i=1}^n$ of size $n\geq 1$ and motivated by \cite{carroll1988optimal} and \cite{stefanski1990deconvolving}, the unknown density $f_{Y,X}(y,x)$ can then be estimated using deconvolution kernels associated with the measurement error $\epsilon$,
\begin{align}\label{stat.gfest}
    \hat{f}_{Y,X}(y,x) = \frac{1}{nb^2}\sum_{i=1}^nK\left(\frac{y-Y_i}{b}\right)\mathcal{K}_\epsilon\left(\frac{x-W_i}{b}\right)
\end{align}
where $K$ is a kernel, $b$ is bandwidth shrinking to zero at suitable rates and 
\begin{align}\label{stat.ft}
    \mathcal{K}_\epsilon(x) = \frac{1}{2\pi}\int \ee^{-\ii tx}\frac{K^{\text{ft}}(t)}{f^{\text{ft}}_\epsilon(t/b)}\,dt.
\end{align}
For notational simplicity, throughout the paper, let $\mathcal{K}_\epsilon(x) = b\mathcal{K}_b(x)$ and let $f_\eta^{\text{ft}}(t)=\int \ee^{\ii tx}f_\eta(x)\,dx$ denote the Fourier transform of the generic random variable $\eta$. Let $\theta_n$ denote the estimator for $\theta_0$ based on the sample $\{(Y_i,W_i)'\}_{i=1}^n$ when the measurement error density $f_\epsilon$ is assumed to be known. The estimator for the variance $\sigma_0^2$ can be obtained by expressing $\mathbb{E}[(Y-g(X;\theta_0))^2]$ as an integral by replacing the joint density $f_{Y,X}(y,x)$ with its estimator $\hat{f}_{Y,X}(y,x)$ in \eqref{stat.gfest}, which, after rearrangement, is given by %plugging in the estimator $\theta_n$ for $\theta_0$, and rearranging,
\begin{align}\label{stat.var}
    \sigma_n^2 = \frac{1}{n}\sum_{i=1}^n\int\left(Y_i-g(x;\theta_n)\right)^2\mathcal{K}_b\left(\frac{x-W_i}{b}\right)\,dx.
\end{align}
Based on the above estimators, we can express the sample version of $S(\xi,\theta_0,\sigma^2_0)$ in \eqref{stat.Integral} as the following empirical process:
\begin{align}\label{stat.known}
    S_n(\xi,\theta_n,\sigma^2_n) = \frac{1}{n}\sum_{i=1}^n\int\left[\left(Y_i-g(x;\theta_n)\right)^2-\sigma_n^2\right]\mathcal{K}_b\left(\frac{x-W_i}{b}\right)\ee^{\ii x\xi}\,dx,
\end{align}
which is expected to be close to zero under the null and deviate from zero under the alternative. 

The test statistics are constructed based on appropriate distances from $S_n(\cdot,\theta_n,\sigma^2_n)$ to zero. We adopt the commonly used Kolmogorov--Smirnov (KS)-type statistic, which is based on the sup norm, and the Cram\'{e}r--von Mises (CvM)-type statistic, which relies on the squared norm, respectively:
\begin{align*}
KS_{n}=\sup_{\xi\in\Pi}\left\vert\sqrt{n}S_{n}(\xi,\theta_n,\sigma^2_n)\right\vert\quad\text{and}\quad
CvM_{n}=\int_\Pi\left\vert\sqrt{n}S_{n}(\xi,\theta_n,\sigma^2_n)\right\vert^2\,d\xi,
\end{align*}
where the uniform integrating measure on $\Pi$ is employed as recommended in \cite{dong2022nonparametric}. The null hypothesis $H_0$ is rejected when the test statistics $KS_{n}$ and $CvM_{n}$ exceed their critical values, which are obtained via a computationally attractive multiplier bootstrap procedure detailed in Section \ref{sec.boot}.

\section{Theoretical Results}\label{sec.Asy}
In this section, we study the asymptotic properties of the $KS_n$ and $CvM_n$ statistics introduced in Section \ref{sec.Test}, under the assumption that the distribution of the measurement error is known. The case with unknown measurement error distribution will be discussed in detail in Section \ref{sec.Unknown}. Based on the smoothness of the measurement error $\epsilon$, we distinguish between the ordinary smooth and supersmooth cases. For each case, we provide the required assumptions and establish the limiting distribution of the empirical process $S_n(\cdot,\theta_n,\sigma^2_n)$ in \eqref{stat.known} under the null. A sequence of local alternatives converging to the null at the parametric rate $n^{-1/2}$ is subsequently investigated, and the corresponding asymptotic local power is derived. Finally, the global power of the proposed test is discussed by analyzing the asymptotic behavior of $S_n(\cdot,\theta_n,\sigma^2_n)$ under the alternative. 

In the following, let $f_X(x)$ denote the density function of $X$, $g^{(p)}(x;\theta)$ denote $p$-th derivative of function $g(x;\theta)$ with respect to $x$ and $\binom{a}{b}=a!/[(a-b)!b!]$ denote the binomial coefficient for nonnegative integers $a\geq b$. To facilitate the theoretical analysis, we begin by imposing a set of regularity conditions that hold for both the ordinary smooth and supersmooth cases.
\begin{Assumption}\label{ass.D}

    \quad

    \begin{enumerate}[label=(\roman*)]
	\item $\{(Y_i,W_i)'\}_{i=1}^n$ is an independent and identically distributed (i.i.d.) sample of $(Y,W)'$, which satisfies \eqref{var.Y}, \eqref{var.W}, and $\mathbb{E}|Y|^4<\infty$.
        
    \item The estimator $\theta_n$ satisfies a convergence rate $\theta_n-\theta_0=O_p(n^{\delta-1/2})$ for some $0\leq \delta<1/4$ in parametric model \eqref{var.Y} with measurement error \eqref{var.W}.

	\item The measurement error $\epsilon$ is independent of $X$.
    \end{enumerate}

\end{Assumption}

Assumption \ref{ass.D}({\romannumeral1}) imposes random sampling and the existence of the fourth moment of $Y$, which are necessary to ensure that the statistic constructed from $Y^2$ has a finite second-order moment. Assumption \ref{ass.D}({\romannumeral2}) relaxes the convergence rate condition on the parameter compared with the case without measurement error. As pointed out in the literature, estimation in parametric models involving measurement error often fails to attain the parametric rate unless additional assumptions are imposed, see \cite{taupin1998estimation}. To account for this, we relax the convergence rate assumption to allow rates between $n^{-1/4}$ and $n^{-1/2}$. Finally, Assumption \ref{ass.D}({\romannumeral3}) adopts the classical measurement error assumption, which ensures the asymptotic zero mean of the test statistic and the existence of its variance. Following \cite{delaigle2008deconvolution}, which highlights that ordinary smooth and supersmooth cases can be distinguished by the decay rate of the characteristic function of the measurement error, we first consider the setting where the characteristic function decays polynomially, and state the corresponding assumptions.
\begin{Assumption}\label{ass.O}

    \quad 

    \begin{enumerate}[label=(\roman*)]
	\item The functions $f_X(x)$ and $h(x;\theta)$ are $p$-times continuously differentiable with bounded and integrable derivatives, where $p$ is a positive integer satisfying $p>\alpha$ and $h(x;\theta)$ can be taken to be any of the following functions: $g(x;\theta)$, $g^2(x;\theta)$, $g^2(x;\theta)f_X(x)$, $\partial g(x;\theta)/\partial\theta$, $g(x;\theta)[\partial g(x;\theta)/\partial\theta]$, and $[\partial g(x;\theta)/\partial\theta][\partial g(x;\theta)/\partial\theta^\top]$, with $\theta$ takeing values in a neighborhood of $\theta_0$. Furthermore, we impose additional assumptions about the Lipschitz continuous properties of $f_X^{(p)}(x)$, $h^{(p)}(x;\theta)$, and $g(x;\theta_0)f_X(x)$ for almost every $x$ as follows:
    \begin{align*}
        &\left\vert f_X^{(p)}(x+y) - f_X^{(p)}(x) \right\vert \leq L_{f_X^{(p)}}(x)\vert y\vert,\\
        &\left\vert h^{(p)}(x+y;\theta) - h^{(p)}(x;\theta) \right\vert \leq L_{h^{(p)}}(x)\vert y\vert,\\
        & \left\vert \left[g(x+y;\theta_0)f_X(x)\right]^{(p)} - \left[g(x;\theta_0)f_X(x)\right]^{(p)} \right\vert \leq L_{\left[g f\right]^{(p)}}(x)\vert y\vert.
    \end{align*}
    The above functions $L_{f_X^{(p)}}(x)$, $L_{h^{(p)}}(x)$, and $L_{\left[g f\right]^{(p)}}(x)$ satisfy integrable conditions:
    \begin{align*}
        &\int \left\vert h(x;\theta)L_{f_X^{(p)}}(x)\right\vert\,dx<\infty,\quad \int \left\vert L_{h^{(p)}}(x)\right\vert\,dx<\infty,\\
        &\int \left\vert \frac{\partial g(x;\theta)}{\partial\theta}L_{\left[g f\right]^{(p)}}(x)\right\vert\,dx<\infty, \quad \text{and}\quad \int \left\vert \frac{\partial g(x;\theta)}{\partial\theta}\frac{\partial g(x;\theta)}{\partial\theta'}L_{\left[gf\right]^{(p)}}(x)\right\vert\,dx<\infty.
    \end{align*}

        \item The characteristic function of the measurement error $\epsilon$ is of the following form for all $t\in\mathbb{R}$, where $c_0^{os},c_1^{os},\ldots,c_{\alpha}^{os}$ are finite constants with $c_0^{os} = 1$ and $\alpha>0$,
        \begin{align*}
            f_\epsilon^{\mathrm{ft}}(t) = \frac{1}{c_0^{os}+c_{1}^{os}t+\cdots+c_{\alpha}^{os}t^{\alpha}}.
        \end{align*}
        
        \item The kernel function $K$ is differentiable to order $p+1$ and satisfies the following conditions:
        \begin{align*}
            &\int K(u)du = 1, \qquad \int u^{p}K(u)du\neq 0, \qquad \int u^lK(u)du = 0
        \end{align*}
        for $l=1,2,\cdots,p-1$. In addition, $K^{\mathrm{ft}}$ is compactly supported on $[-c_0,c_0]$, symmetric around zero, and bounded.
        
        \item $nb^{2p}\to 0$ as $n\to\infty$.
        
        \item For $c_l^{os}(\xi) = (-i)^l\sum\limits_{h=l}^{\alpha}c_h^{os}\binom{h}{l}\xi^{h-l}$, we have
        \begin{align*}
            \mathbb{E}\left[\sup_{\xi\in\Pi}\left\vert \sum_{l=0}^\alpha c_l^{os}(\xi)h^{(l)}(W;\theta)\right\vert^2\right]<\infty,
        \end{align*}
        where the function $h(x;\theta)$ is mentioned in ({\romannumeral1}).
    \end{enumerate}
\end{Assumption}

To facilitate a detailed analysis of the structural functions at the observed sample points $W_i$ using Taylor expansion, Assumption \ref{ass.O}({\romannumeral1}) places smoothness restrictions on the structural functions. This follows the assumptions of the ordinary smooth case in existing work on measurement error (e.g., \cite{dong2022nonparametric}), which impose Lipschitz continuity and integrability of the corresponding Lipschitz coefficients. Moreover, to ensure that the variance of the empirical process $S_n(\cdot,\theta_n,\sigma^2_n)$ in \eqref{stat.known} is not distorted by \textquotedblleft parameter estimation uncertainty\textquotedblright, initially discussed in \cite{durbin1973distribution}, we strengthen the smoothness requirements to hold uniformly over a neighborhood of the true value $\theta_0$. Assumption \ref{ass.O}({\romannumeral1}) strengthens the conventional ordinary smooth assumption by characterizing the exact limiting behavior of $f^{\mathrm{ft}}_\epsilon(\cdot)$, which is generalized to the form $f^{\mathrm{ft}}_\epsilon(t) = \exp(\ii t\zeta)/(c_0^{os}+c_{1}^{os}t+\cdots+c_{\alpha}^{os}t^{\alpha})$ for some real number $\zeta$ in \cite{fan1995average}. This refinement is essential for deriving the precise asymptotic form of the test statistic; otherwise, stronger assumptions would be required to impose on a more complex expression involving the Fourier transform to ensure the existence of the variance and the asymptotic negligibility of the expectation. In addition, commonly used distributions such as the Laplace and Gamma distributions are included under this assumption. Accordingly, we introduce higher-order kernel functions, following the construction method provided in \cite{alexander2009deconvolution} for kernels of arbitrary order in Assumption \ref{ass.O}({\romannumeral3}). Kernel functions are used to characterize the integration involving the deconvolution kernel, thereby deriving the precise asymptotic form of the test statistic. In addition, higher-order kernels are also employed to establish the asymptotic negligibility of the expectation. Assumption \ref{ass.O}({\romannumeral4}) requires that the bandwidth converges to zero sufficiently fast as the sample size increases. This so-called undersmoothing assumption, commonly used in the literature, is essential to guarantee the asymptotic negligibility of the expectation of the proposed statistics. Finally, Assumption \ref{ass.O}({\romannumeral5}) is a necessary condition to ensure the boundedness of the variance of $S_n(\cdot,\theta_n,\sigma^2_n)$, as noted in \cite{dong2022nonparametric}, and has been widely adopted in the literature on measurement error models.

Based on the assumptions stated above, the following theorem characterizes the asymptotic behavior of the empirical process $S_n(\cdot,\theta_n,\sigma^2_n)$ in \eqref{stat.known} under the null, in the case where the measurement error distribution is known and is of the ordinary smooth type. Let \textquotedblleft $\Longrightarrow$\textquotedblright{} denote weak convergence on $(l^{\infty}(\Pi),\mathcal{B}_\infty)$ in the sense of Hoffmann--J\textup{\o{}}rgensen, where $\mathcal{B}_\infty$ denotes the corresponding Borel $\sigma$-algebra, see, e.g., Definition $1.3.3$ in \cite{van1996weak}.
\begin{theorem}\label{theorem.known ordinary smooth under H0}
    Suppose Assumptions \ref{ass.D} and \ref{ass.O} hold. Under $H_0$ in \eqref{hyp.H0},
    \begin{align}\label{theorem.known ordinary smooth under H0 eq1}
        \sqrt n S_{n}(\cdot,\theta_n,\sigma^2_n)\Longrightarrow S_{\infty}^{os}(\cdot,\theta_0,\sigma^2_0),
    \end{align}
    where $S_{\infty}^{os}(\cdot,\theta_0,\sigma^2_0)$ is a Gaussian process with mean zero and covariance structure
    \begin{align*}
        &Cov\left[S_{\infty}^{os}(\xi_1,\theta_0,\sigma^2_0),S_{\infty}^{os}(\xi_2,\theta_0,\sigma^2_0)\right] = \mathbb{E}\left[r^{os}_{\infty}(Y,W;\xi_1,\theta_0,\sigma^2_0)r^{os}_{\infty}(Y,W;\xi_2,\theta_0,\sigma^2_0)\right],
    \end{align*}
    with
    \begin{align}\label{theorem.ordinary smooth under H_0 ros}
        r^{os}_{\infty}(Y,W;\xi,\theta_0,\sigma^2_0) =& \ee^{\ii W\xi}\sum\limits_{l=0}^{\alpha} c_l^{os}(\xi) \left[\left(Y-g(W;\theta_0)\right)^2-\sigma_0^2\right]^{(l)} \notag\\
        & - f_X^{\mathrm{ft}}(\xi)\sum\limits_{l=0}^{\alpha} c_l^{os}(0) \left[\left(Y-g(W;\theta_0)\right)^2-\sigma_0^2\right]^{(l)}.
    \end{align} 
\end{theorem}

Theorem \ref{theorem.known ordinary smooth under H0} shows that the proposed empirical process $S_n(\cdot,\theta_n,\sigma^2_n)$ converges at the parametric rate (i.e., $\sqrt{n}$) to a centered Gaussian process. As a consequence, the tests $KS_n$ and $CvM_n$ based on $S_n(\cdot,\theta_n,\sigma^2_n)$ also achieve the parametric rate.\footnote{Given the conditions in Theorem \ref{theorem.known ordinary smooth under H0}, the continuous mapping theorem (see \cite{van1996weak}) implies that the proposed $KS_n$ and $CvM_n$ statistics converge to the sup norm and the squared $L_2$-norm of the limiting Gaussian process $S_\infty^{os}(\cdot,\theta_0,\sigma^2_0)$, respectively. As similar techniques are used to derive the limiting distributions of these statistics under the null and the alternative hypotheses, for both ordinary smooth and supersmooth measurement errors, and in the multiplier bootstrap procedure, we do not repeat the details in the subsequent analysis.} Notably, this parametric rate is not affected by the use of deconvolution kernels, which are nonparametric estimators, highlighting a key advantage of the proposed procedure. In fact, such an advantage has also been observed in various methods that address measurement error problems using nonparametric estimators, including semiparametric estimators for regression functions, nonlinear regression estimation, and significance testing; see, for example, the discussions in \cite{hall2007semiparametric} and \cite{dong2022nonparametric}. It is also worth noting that although we use a deconvolution kernel to estimate $f_{Y,X}(y,x)$, the proposed test statistics evaluate the discrepancy over the entire parameter space for $\xi$. As a result, both the size accuracy and power of the tests exhibit robustness with respect to the bandwidth choice, which will be further discussed in Section \ref{sec.Simulation}.

Next, we consider the scenario where the characteristic function of the measurement error decays at an exponential rate, a setting commonly referred to as the supersmooth case, as described in Assumption \ref{ass.S}({\romannumeral2}). The required assumptions are presented below.
\begin{Assumption}\label{ass.S}

    \quad

    \begin{enumerate}[label=(\roman*)]
	    \item The functions $f_X(x)$ and $h(x;\theta)$ are infinitely differentiable with respect to $x$, where $h(x;\theta)$ is as defined in Assumption \ref{ass.O}({\romannumeral1}).

        \item The measurement error $\epsilon$ follows a Gaussian distribution with the characteristic function of the following form for all $t\in\mathbb{R}$ and some positive constant $\mu$,
        \begin{align*}
            f^{\mathrm{ft}}_\epsilon(t) = \ee^{-\mu t^2}.
        \end{align*}
        
        \item The kernel function $K$ is infinitely differentiable and satisfies the following conditions for all $l\in\mathbb{N}$,
        \begin{align*}
            \int K(u)du = 1, \qquad \int u^lK(u)du = 0.
        \end{align*}
        In addition, $K^{\mathrm{ft}}$ has a compact support set, is symmetric around zero, and is bounded.
        
        \item $b\to0$ as $n\to\infty$.
        \item For $c_l^{ss}(\xi) = (-i)^l\sum\limits_{h\geq\frac{l}{2}}^{\infty}\frac{\mu^h}{h!}\binom{2h}{l}\xi^{2h-l}$,
        \begin{align*}
            \mathbb{E}\left[\sup_{\xi\in\Pi}\left\vert \sum_{l=0}^\alpha c_l^{ss}(\xi)h^{(l)}(W;\theta)\right\vert^2\right]<\infty.
        \end{align*}
    \end{enumerate}
\end{Assumption}
It is worth noting that for the supersmooth case, Assumptions \ref{ass.S}({\romannumeral1}), ({\romannumeral2}), and ({\romannumeral3}) impose stronger conditions than their counterparts for the ordinary smooth case in Assumption \ref{ass.O}. Specifically, we require that the characteristic function of the measurement error be of exponential type, the structural functions be infinitely differentiable, and the kernel functions be of infinite order. While these assumptions are more restrictive, they are satisfied by many practical examples—for instance, using polynomials, exponentials, or circular functions as structural functions (with additional examples provided in \cite{alexander2009deconvolution}); normally distributed measurement error; and the infinite-order kernel constructed in \cite{mcmurry2004nonparametric}, which is also employed in our simulation study in Section \ref{sec.Simulation}. Moreover, such strengthened assumptions are necessary in the supersmooth setting. Without these strengthened assumptions, two main difficulties arise. First, it becomes challenging to ensure the boundedness of the variance of the test statistic, as required in Assumption \ref{ass.S}({\romannumeral5}). Second, one would need to impose additional conditions to guarantee the asymptotic negligibility of the bias, such as the undersmoothing bandwidths and the Lipschitz continuity of the structural functions stated in Assumption \ref{ass.O}, which would significantly complicate the set of required assumptions. In addition, the conditions can be relaxed to permit the measurement error to be distributed as the convolution of a Gaussian density and any ordinary smooth density that satisfies Assumption \ref{ass.O}({\romannumeral2}), thereby broadening the applicability of the supersmooth case.

\begin{theorem}\label{theorem.known supersmooth under H0}
    Suppose Assumptions \ref{ass.D} and \ref{ass.S} hold. Under $H_0$ in \eqref{hyp.H0},
    \begin{align}\label{theorem.known supersmooth under H0 eq1}
        \sqrt n S_{n}(\cdot,\theta_n,\sigma^2_n)\Longrightarrow S_{\infty}^{ss}(\cdot,\theta_0,\sigma^2_0),
    \end{align}
    where $S_{\infty}^{ss}(\cdot,\theta_0,\sigma^2_0)$ is a Gaussian process with mean zero and covariance structure
    \begin{align*}
        &Cov\left[S_{\infty}^{ss}(\xi_1,\theta_0,\sigma^2_0),S_{\infty}^{ss}(\xi_2,\theta_0,\sigma^2_0)\right] = \mathbb{E}\left[r^{ss}_{\infty}(Y,W;\xi_1,\theta_0,\sigma^2_0)r^{ss}_{\infty}(Y,W;\xi_2,\theta_0,\sigma^2_0)\right],
    \end{align*}
    with
    \begin{align}\label{theorem.supersmooth under H_0 ros}
        r^{ss}_{\infty}(Y,W;\xi,\theta_0,\sigma^2_0) =& \ee^{\ii W\xi}\sum\limits_{l=0}^{\infty} c_l^{ss}(\xi) \left[\left(Y-g(W;\theta_0)\right)^2-\sigma_0^2\right]^{(l)} \notag\\
        & - f_X^{\mathrm{ft}}(\xi)\sum\limits_{l=0}^{\infty} c_l^{ss}(0) \left[\left(Y-g(W;\theta_0)\right)^2-\sigma_0^2\right]^{(l)}.
    \end{align} 
\end{theorem}

Theorem \ref{theorem.known supersmooth under H0} shows that our empirical process $S_n(\cdot,\theta_n,\sigma^2_n)$ also achieves the parametric rate in the supersmooth case, and that the limiting Gaussian process exhibits a variance structure similar to that of the ordinary smooth case. However, existing literature suggests that, in other measurement error-related problems, the convergence rate and the variance structure of the limiting null process may differ substantially between the ordinary smooth and the supersmooth cases. In particular, \cite{otsu2021specification} demonstrates that, in specification testing based on nonparametric estimators, the supersmooth case often yields a slower convergence rate and a more complex asymptotic structure. A detailed comparison between the ordinary smooth and the supersmooth cases is beyond the scope of this paper. Instead, in Section \ref{sec.Unknown}, we focus on constructing alternative test statistics when the distribution of the measurement error is unknown.

We next introduce a sequence of local alternatives and evaluate the asymptotic behavior of the proposed tests under these alternatives, demonstrating that the tests possess nontrivial local power to detect nonconstant variance. Specifically, we assume 
\begin{align}\label{hyp.H1n}
    H_{1n}:\mathbb{E}\left[U^2\mid X\right] = \sigma_0^2+\frac{\Delta(X)}{\sqrt{n}}\quad \text{a.s. for some }0<\sigma_0^2<\infty,
\end{align}
where $\Delta:\mathbb{R}\to\mathbb{R}$ is a nonzero function satisfying $\mathbb E[\Delta(X)]=0$ and $\mathbb{E}|\Delta(X)|<\infty$, so that the sequence of local alternatives converges to the null hypothesis at the parametric rate. We then investigate the asymptotic convergence results of the test statistic separately under the ordinary smooth and the supersmooth cases.
\begin{theorem}\label{theorem.known under H1n}
   Suppose Assumption \ref{ass.D} holds. Under $H_{1n}$ in \eqref{hyp.H1n}, if Assumption \ref{ass.O} holds for the ordinary smooth case,
    \begin{equation*}
        \sqrt n S_{n}(\cdot,\theta_n,\sigma^2_n)\Longrightarrow S_{\infty}^{os}(\cdot,\theta_0,\sigma^2_0)+\mu_\Delta(\cdot),
    \end{equation*}
    and if Assumption \ref{ass.S} holds for the supersmooth case,
    \begin{equation*}
        \sqrt n S_{n}(\cdot,\theta_n,\sigma_n)\Longrightarrow S_{\infty}^{ss}(\cdot,\theta_0,\sigma_0)+\mu_\Delta(\cdot),
    \end{equation*}
    where $S_{\infty}^{os}(\cdot,\theta_0,\sigma^2_0)$ and $S_{\infty}^{ss}(\cdot,\theta_0,\sigma^2_0)$ are the centered Gaussian processes as defined in Theorems \ref{theorem.known ordinary smooth under H0} and \ref{theorem.known supersmooth under H0}, respectively, and $\mu_\Delta(\cdot)$ is a deterministic shift function given by
    %\begin{align*}
    %    \mu_\Delta(\xi) = \mathbb{E}\left[\Delta(X)\ee^{\ii X\xi}\right]-\mathbb{E}\left[\Delta(X)\right]f_{X}^{\mathrm{ft}}(\xi).
    %\end{align*}
    \begin{align*}
        \mu_\Delta(\xi) = \mathbb{E}\left[\Delta(X)\ee^{\ii X\xi}\right].
    \end{align*}
\end{theorem}

Theorem \ref{theorem.known under H1n} implies that under the sequence of local alternatives $H_{1n}$ in \eqref{hyp.H1n}, the empirical process $S_{n}(\cdot,\theta_n,\sigma^2_n)$ still converges at the parametric rate to its limiting distribution. Notably, for both the ordinary smooth and the supersmooth cases, the limiting distribution under $H_{1n}$ differs from that under $H_0$ by a deterministic shift function $\mu_\Delta(\xi)$. This shift corresponds to the covariance between $\Delta(X)$ and $\ee^{\ii X\xi}$, ensuring that the limiting distributions under $H_{1n}$ and $H_0$ are different and thus yielding nontrivial asymptotic local power for the proposed $S_{n}(\cdot,\theta_n,\sigma^2_n)$.

At the end of this section, we establish the global consistency of $S_{n}(\cdot,\theta_n,\sigma^2_n)$ by showing that it possesses nontrivial power under the alternative hypothesis $H_1$ in \eqref{hyp.H1}, as formalized in the following theorem. Let $\sigma_\ast^2$ denote the probability limit of $\sigma_n^2$ under $H_1$, i.e., the pseudo-true value and the unconditional variance of $U$. Note that %$\sigma^2_\ast=\mathbb{E}[U^2]$ is the unconditional variance of $U$ under $H_1$ and 
$\sigma_\ast^2=\sigma_0^2$ under $H_0$.
\begin{theorem}\label{theorem.known under H1}
    Suppose Assumption \ref{ass.D} holds. Under $H_{1}$ in \eqref{hyp.H1}, if either Assumption \ref{ass.O} holds for the ordinary smooth case or Assumption \ref{ass.S} holds for the supersmooth case,     
    %\begin{align*}
    %    &\sup_{\xi\in\Pi}\vert S_{n}(\xi,\theta_n,\sigma^2_n) - C(\xi,\theta_0,\sigma^2_\ast)\vert =  o_p(1),
    %\end{align*}
    \begin{align*}
        \sup_{\xi\in\Pi}\vert S_{n}(\xi,\theta_n,\sigma^2_n) - S(\xi,\theta_0,\sigma^2_\ast)\vert =  o_p(1),
    \end{align*}
    where $S(\xi,\theta_0,\sigma^2_\ast) = \mathbb{E}\{[(Y-g(X;\theta_0)^2-\sigma_\ast^2]\ee^{\ii X\xi}\}\neq 0$ for some $\xi\in\Pi$ with a positive measure.
    %\begin{align}\label{theorem.known under H1 eq1}
     %   &C(\xi,\theta_0,\sigma^2_\ast) = \mathbb{E}\left[\left(U^2-\sigma_\ast^2\right)\ee^{\ii X\xi}\right]-\mathbb{E}\left(U^2-\sigma_\ast^2\right)f_{X}^{\mathrm{ft}}(\xi).
    %\end{align}
    %\begin{align}\label{theorem.known under H1 eq1}
    %    C(\xi,\theta_0,\sigma^2_\ast) = \mathbb{E}\left[\left(U^2-\sigma_\ast^2\right)\ee^{\ii X\xi}\right].
    %\end{align}
\end{theorem}

%Note that the deterministic function $C(\xi,\theta_0,\sigma^2_0)$ in \eqref{theorem.known under H1 eq1} is the covariance between $U^2=(Y-g(X;\theta_0))^2$ and $\ee^{\ii X\xi}$. 
Under $H_1$ %in which the error in the regression model is not independent of the regressors, such as the setting discussed in \cite{wang2007assessing}, 
where the conditional variance $\mathbb E[U^2|X]$ is a nonconstant function of the unobserved regressor $X$, Theorem \ref{theorem.known under H1} implies that $S_{n}(\xi,\theta_n,\sigma^2_n)$ converges in probability uniformly to the nonzero constant function $S(\xi,\theta_0,\sigma^2_\ast)$. Consequently, the test statistics $KS_n$ and $CvM_n$ based on $\sqrt nS_{n}(\xi,\theta_n,\sigma^2_n)$ diverge to positive infinity in probability, thereby guaranteeing the consistency of $KS_n$ and $CvM_n$ against $H_1$.

\section{Case of Unknown Measurement Error}\label{sec.Unknown}
In this section, we focus on constructing tests for heteroskedasticity when information about the measurement error distribution is unavailable. The necessity arises because a misspecified measurement error model can lead to biased estimators, as noted in \cite{hall2007semiparametric}. The assumption of replicated measurements has been widely adopted in numerous theoretical studies, see, e.g., \cite{carroll1992diagnostics} and \cite{delaigle2008deconvolution}. In practice, repeated measurements can be obtained in certain settings, for example, through high-frequency observations over a short time interval or by rapidly generating data from simulated systems. As noted in \cite{allen1999checking}, climate scientists, when applying optimal fingerprinting methods to detect and attribute anthropogenic influences on climate, can readily obtain repeated measurements from a climate model's control run. As commonly stated in the literature \citep{carroll1992diagnostics}, repeated measurements on unobservable regressors are usually required in multiple replicates to obtain accurate information. One of the notable advantages of our method is that a single set of repeated measurements suffices to guarantee favorable asymptotic behavior of the test statistics, as detailed in the setting below:
\begin{align}\label{var.repeated ME}
    &W^r = X + \epsilon^r,
\end{align}
where $W$ and $W^r$ denote a pair of repeated measurements on the latent variable $X$ and $(\epsilon,\epsilon^r)$ are assumed to be i.i.d. measurement errors. Additionally, we make the following assumptions,
\begin{Assumption}\label{ass.D'}

    \quad

    \begin{enumerate}[label=(\roman*)]
	\item $\{W_i^r\}_{i=1}^n$ is an i.i.d sample of $W^r$ satisfying \eqref{var.repeated ME}.

	\item $\mathbb{E}\vert\epsilon\vert^{(p+1)(2+\zeta)}<\infty$ for some positive constant $\zeta$ and $f_\epsilon$ is symmetric around zero.
    \end{enumerate}

\end{Assumption}

Based on the repeated measurements described in Assumption \ref{ass.D'}({\romannumeral1}), Assumption \ref{ass.D'}({\romannumeral2}) first imposes a symmetry condition on the distribution of the measurement error, which is equivalent to requiring that its characteristic function be real-valued. It then places relatively strong moment conditions on the measurement error $\epsilon$. Commonly used measurement error distributions, such as the Laplace and normal distributions, satisfy these conditions. While some studies have considered relaxations of this assumption, for instance, \cite{li1998nonparametric} relaxes the symmetry condition, and \cite{delaigle2016methodology} weakens the repeated measurement requirement by imposing stronger restrictions on the true regressor, we do not pursue these extensions here for simplicity. Assumption \ref{ass.D'}({\romannumeral2}), motivated by the later analysis of the asymptotic behavior of the derivative of the unknown characteristic function of the measurement error, is discussed in detail in \cite{kurisu2022uniform}. It characterizes the integration of the estimated deconvolution kernel in our work.

Given the availability of replicate measurements for each sample $X_i$, we can obtain an estimator for $\theta_0$, which we denote as $\hat\theta_n$. Note that $\hat\theta_n$ is different from $\theta_n$, which is constructed under the known measurement error distribution assumption. The repeated measurements are also used to estimate the characteristic function of the measurement error, which is then plugged into the previously constructed deconvolution kernel to obtain an estimator, as given by
\begin{align}\label{stat.ME density}
    &\hat{f}^{\mathrm{ft}}_\epsilon(t) = \left\vert\frac{1}{n}\sum\limits_{i=1}^n\cos\left[t(W_i-W_i^r)\right]\right\vert^{1/2},\quad \hat{\mathcal{K}}_b(x) = \frac{1}{2\pi b}\int \ee^{-\ii tx}\frac{K^{\mathrm{ft}}(t)}{\hat{f}^{\mathrm{ft}}_\epsilon(t/b)}\,dt.
\end{align}
The estimated kernel $\hat{\mathcal{K}}_b(x)$ is subsequently used in place of the original kernel in the proposed estimator of the variance \eqref{stat.var} and the empirical process \eqref{stat.known} to obtain 
\begin{align}\label{stat.var unknown}
    \hat{\sigma}_n^2 = \frac{1}{n}\sum_{i=1}^n\int\left(Y_i-g(x;\hat{\theta}_n)\right)^2\hat{\mathcal{K}}_b\left(\frac{x-W_i}{b}\right)\,dx
\end{align}
and
\begin{align}\label{stat.unknown}
    \hat{S}_n(\xi,\hat{\theta}_n,\hat{\sigma}^2_n) = \frac{1}{n}\sum_{i=1}^n\int\left[\left(Y_i-g(x;\hat{\theta}_n)\right)^2-\hat{\sigma}_n^2\right]\hat{\mathcal{K}}_b\left(\frac{x-W_i}{b}\right)\ee^{\ii x\xi}\,dx.
\end{align}
The corresponding test statistics with repeated measurements, denoted by $\widehat{KS}_n$ and $\widehat{CvM}_n$, are then constructed in the same manner as the original $KS_n$ and $CvM_n$ statistics, respectively. The bootstrap procedure for obtaining critical values will be discussed in detail in Section \ref{sec.boot}. In addition, we introduce the notation $\Pi_\epsilon$ to evaluate the accuracy of $\hat{f}^{\mathrm{ft}}_\epsilon(\cdot)$ as an estimator of $f^{\mathrm{ft}}_\epsilon(\cdot)$,
\begin{align*}
    \Pi_{\epsilon}(t) = \frac{1}{2}-\frac{\cos(t(W-W^r))}{2\left\vert f^{\mathrm{ft}}_\epsilon(t)\right\vert^2}.
\end{align*}
In fact, due to the repeated measurements and the symmetry of the error distribution as assumed in Assumption \ref{ass.D'}, $\Pi_\epsilon$ is unbiased. 

We still distinguish between the ordinary smooth and the supersmooth cases based on the decay rate of the characteristic function of the measurement error in the case of an unknown measurement error distribution due to the differing conditions required to derive the asymptotic properties of the reconstructed test statistics. In the first case, the ordinary smooth scenario, we add the following assumption to Assumption \ref{ass.O}.
\begin{Assumption}\label{ass.O'}

    \quad 

    \begin{enumerate}[label=(\roman*)]
        \item $nb^{10\alpha+6}\log(\frac{1}{b})^{-4}\to \infty$ as $n\to\infty$.

        \item Strengthened convergence condition for the term brought by the unknown measurement error distribution estimation:
    \begin{align*}
        &\mathbb{E}\left[\sup_{\xi\in\Pi} \left\vert r^{\epsilon,os}_{\infty}(Y,W,W^r;\xi,\theta_0,\sigma^2_0)\right\vert^2\right]<\infty,
    \end{align*}
    where
    \begin{align*}
        r^{\epsilon,os}_{\infty}(Y,W,W^r;\xi,\theta_0,\sigma^2_0) =&\ee^{\ii W\xi}\sum_{l=0}^\alpha c_l^{os}(\xi)\left[\left(Y-g(W;\theta_0)\right)^2-\sigma_0^2\right]^{(l)} \\
        &+\left[g^2 f_X\right]^{\mathrm{ft}}(\xi)\Pi_{\epsilon}(\xi)+\frac{1}{2\pi}\int f_X^{\mathrm{ft}}(t)(g^2)^{\mathrm{ft}}(\xi-t)\Pi_{\epsilon}(t)\,dt\\
        &-\frac{1}{\pi}\int (g f_X)^{\mathrm{ft}}(t)g^{\mathrm{ft}}(\xi-t)\Pi_{\epsilon}(t)\,dt.
    \end{align*}
    
    \end{enumerate}

\end{Assumption}
Assumption \ref{ass.O'} is essential, as the estimation of $\hat{f}^{\mathrm{ft}}_\epsilon(\cdot)$ introduces additional higher-order terms and, unfortunately, distorts the variance structure of the main term. By imposing Assumption \ref{ass.O'}({\romannumeral1}), we establish a lower bound on the bandwidth, which ensures the negligibility of the unavoidable higher-order terms and provides a theoretically justified bandwidth range combined with the upper bound specified in Assumption \ref{ass.O}({\romannumeral4}). In addition, estimating the error characteristic functions modifies the form of the main term of the original empirical process. Specifically, it adds an additional component to the leading term $r^{os}_{\infty}(Y,W,W^r;\cdot,\theta_0,\sigma^2_0)$ that involves the stochastic term $r^{\epsilon,os}_{\infty}(Y,W,W^r;\cdot,\theta_0,\sigma^2_0)$ mentioned in Assumption \ref{ass.O'}({\romannumeral2}), thereby necessitating a strengthening of Assumption \ref{ass.O}({\romannumeral5}) through the requirement of bounded second moment of $r^{\epsilon,os}_{\infty}(Y,W,W^r;\cdot,\theta_0,\sigma^2_0)$ to ensure the finiteness of the variance. We then establish the asymptotic theory for the reconstructed test statistics for the ordinary smooth case with unknown measurement error distribution, as stated in the following theorem.
\begin{theorem}\label{theorem.unknown ordinary smooth under H0}
    Suppose Assumptions \ref{ass.D}, \ref{ass.O}, \ref{ass.D'}, and \ref{ass.O'} hold. Under $H_0$ in \eqref{hyp.H0},
    \begin{align}\label{theorem.unknown ordinary smooth under H0 eq1}
        \sqrt n \hat{S}_{n}(\cdot,\hat\theta_n,\hat{\sigma}^2_n)\Longrightarrow \hat{S}_{\infty}^{os}(\cdot,\theta_0,\sigma^2_0),
     \end{align}
    where $\hat{S}_{\infty}^{os}(\cdot,\theta_0,\sigma^2_0)$ is a Gaussian process with mean zero and covariance structure
    \begin{align*}
        &Cov\left[\hat{S}_{\infty}^{os}(\xi_1,\theta_0,\sigma^2_0),\hat{S}_{\infty}^{os}(\xi_2,\theta_0,\sigma^2_0)\right] = \mathbb{E}\left[\hat{r}^{os}_{\infty}(Y,W,W^r;\xi_1,\theta_0,\sigma^2_0)\hat{r}^{os}_{\infty}(Y,W,W^r;\xi_2,\theta_0,\sigma^2_0)\right],
    \end{align*}
    with
    \begin{align*}
        &\hat{r}^{os}_{\infty}(Y,W,W^r;\xi,\theta_0,\sigma^2_0) = r^{\epsilon,os}_{\infty}(Y,W,W^r;\xi,\theta_0,\sigma^2_0)-f_X^{\mathrm{ft}}(\xi)r^{\epsilon,os}_{\infty}(Y,W,W^r;0,\theta_0,\sigma^2_0),
    \end{align*}
    where $r^{\epsilon,os}_{\infty}(Y,W,W^r;\xi,\theta_0,\sigma^2_0)$ is defined in Assumption \ref{ass.O'}({\romannumeral2}).
\end{theorem}

Theorem \ref{theorem.unknown ordinary smooth under H0} establishes that the reconstructed empirical process still converges at the parametric rate\footnote{The asymptotic null distributions of statistics $\widehat{KS}_n$ and $\widehat{CvM}_n$ are established in the same manner as in Theorem \ref{theorem.known ordinary smooth under H0}. By applying the continuous mapping theorem, it follows directly that $\widehat{KS}_n$ and $\widehat{CvM}_n$ converge to the sup norm and the squared $L_2$-norm of the limiting Gaussian process $\hat{S}^{os}_\infty(\cdot,\theta_0,\sigma^2_0)$, respectively.}, unaffected by the error brought by the estimation of the error characteristic functions. In contrast to the case with a known measurement error distribution, the covariance structure of the limiting process differs and requires more restrictive assumptions on the variance. 

For the supersmooth case, we assume
\begin{Assumption}\label{ass.S'}

    \quad 

    \begin{enumerate}[label=(\roman*)]
        \item $n\ee^{-6\mu (1+b^{-1})^{2}}\log(\frac{1}{b})^{-2}\to \infty$ as $n\to\infty$.
        
        \item We assume
        \begin{align*}
            &\mathbb{E}\left[\sup_{\xi\in\Pi} \left\vert r^{\epsilon,ss}_{\infty}(Y,W,W^r;\xi,\theta_0,\sigma^2_0)\right\vert^2\right]<\infty.
        \end{align*}
        where
        \begin{align*}
            r^{\epsilon,ss}_{\infty}(Y,W,W^r;\xi,\theta_0,\sigma^2_0) =&\ee^{\ii W\xi}\sum_{l=0}^\infty c_l^{ss}(\xi)\left[\left(Y-g(W;\theta_0)\right)^2-\sigma_0^2\right]^{(l)} \\
            &+\left[g^2 f_X\right]^{\mathrm{ft}}(\xi)\Pi_{\epsilon}(\xi)+\frac{1}{2\pi}\int f_X^{\mathrm{ft}}(t)(g^2)^{\mathrm{ft}}(\xi-t)\Pi_{\epsilon}(t)\,dt\\
            &-\frac{1}{\pi}\int (g f_X)^{\mathrm{ft}}(t)g^{\mathrm{ft}}(\xi-t)\Pi_{\epsilon}(t)\,dt
        \end{align*}
        and $\Pi_{\epsilon}(t)$ is defined in Assumption \ref{ass.O'}({\romannumeral2}).
    \end{enumerate}

\end{Assumption}

Analogously, for the supersmooth case with unknown measurement error, the estimation brings a higher-order term, rendered asymptotically negligible through a stronger bandwidth condition in Assumption \ref{ass.S'}({\romannumeral1}), as well as additional uncertainty in the main term, which necessitates a more stringent condition to ensure variance boundedness. Compared with the ordinary smooth case under an unknown measurement error distribution, Assumption \ref{ass.S'}({\romannumeral2}) is imposed to ensure the finiteness of the variance. A similar comparison can be observed in Assumptions \ref{ass.O} and \ref{ass.S} when the error distribution is known.
\begin{theorem}\label{theorem.unknown supersmooth under H0}
    Suppose Assumptions \ref{ass.D}, \ref{ass.S}, \ref{ass.D'}, and \ref{ass.S'} hold. Under $H_0$ in \eqref{hyp.H0},
    \begin{align}\label{theorem.unknown supersmooth under H0 eq1}
        \sqrt n \hat{S}_{n}(\cdot,\hat\theta_n,\hat{\sigma}^2_n)\Longrightarrow \hat{S}_{\infty}^{ss}(\cdot,\theta_0,\sigma^2_0),
     \end{align}
    where $\hat{S}_{\infty}^{ss}(\cdot,\theta_0,\sigma^2_0)$ is a Gaussian process with mean zero and covariance structure
    \begin{align*}
        &Cov\left[\hat{S}_{\infty}^{ss}(\xi_1,\theta_0,\sigma^2_0),\hat{S}_{\infty}^{ss}(\xi_2,\theta_0,\sigma^2_0)\right] = \mathbb{E}\left[\hat{r}^{ss}_{\infty}(Y,W,W^r;\xi_1,\theta_0,\sigma^2_0)\hat{r}^{ss}_{\infty}(Y,W,W^r;\xi_2,\theta_0,\sigma^2_0)\right],
    \end{align*}
    with
    \begin{align*}
        &\hat{r}^{ss}_{\infty}(Y,W,W^r;\xi,\theta_0,\sigma^2_0) = r^{\epsilon,ss}_{\infty}(Y,W,W^r;\xi,\theta_0,\sigma^2_0)-f_X^{ft}(\xi)r^{\epsilon,ss}_{\infty}(Y,W,W^r;0,\theta_0,\sigma^2_0),
    \end{align*}
    where $r^{\epsilon,ss}_{\infty}(Y,W,W^r;\xi,\theta_0,\sigma^2_0)$ is defined in Assumption \ref{ass.S'}({\romannumeral2}).
\end{theorem}

In general, the lack of measurement error information and the supersmooth case are typically considered to lead to slower convergence rates. Nevertheless, the reconstructed empirical process proposed in this paper retains the parametric rate of convergence, a desirable property not well established in the testing literature when measurement error is present. Overall, the findings from the two cases suggest that estimating the error characteristic functions does not alter the requirements on the smoothness of the structural functions or on the order of the kernel functions. Moreover, it does not affect the convergence rate of the test statistic toward its asymptotic null distribution. However, it leads to different admissible range for the bandwidth and changes the asymptotically equivalent form of $\hat{S}_{n}(\cdot,\hat\theta_n,\hat{\sigma}^2_n)$, denoted respectively by $\hat{r}^{os}_{\infty}(Y,W,W^r;\cdot,\theta_0,\sigma^2_0)$ and $\hat{r}^{ss}_{\infty}(Y,W,W^r;\cdot,\theta_0,\sigma^2_0)$ for ordinary smooth and supersmooth, resulting in different variance structures of the limiting process. 

Similar to the case where the measurement error distribution is known, we claim that the proposed test exhibits nontrivial local power even in the absence of the information about the measurement error distribution, as demonstrated by studying the asymptotic properties of the reconstructed statistics under the sequence of local alternative hypotheses described in \eqref{hyp.H1n}. We present the following theorem to establish this result.
\begin{theorem}\label{theorem.unknown under H1n}
    Suppose Assumptions \ref{ass.D} and \ref{ass.D'} hold. Under $H_{1n}$ in \eqref{hyp.H1n}, if Assumptions \ref{ass.O} and \ref{ass.O'} hold for the ordinary smooth case,
    \begin{align*}
        \sqrt n \hat{S}_{n}(\cdot,\hat\theta_n,\hat{\sigma}^2_n)\Longrightarrow \hat{S}_{\infty}^{os}(\cdot,\theta_0,\sigma^2_0)+\mu_\Delta(\cdot),
    \end{align*}
    and if Assumptions \ref{ass.S} and \ref{ass.S'} hold for the supersmooth case, 
    \begin{align*}
        \sqrt n \hat{S}_{n}(\cdot,\hat\theta_n,\hat{\sigma}^2_n)\Longrightarrow \hat{S}_{\infty}^{ss}(\cdot,\theta_0,\sigma^2_0)+\mu_\Delta(\cdot),
    \end{align*}
    where $\hat{S}_{\infty}^{os}(\cdot,\theta_0,\sigma^2_0)$ and $\hat{S}_{\infty}^{ss}(\cdot,\theta_0,\sigma^2_0)$ are the centered Gaussian processes as defined in Theorems \ref{theorem.unknown ordinary smooth under H0} and \ref{theorem.unknown supersmooth under H0}, respectively, and $\mu_\Delta(\cdot)$ is the deterministic shift function as defined in Theorem \ref{theorem.known under H1n}.
\end{theorem}

We observe that when the alternative hypothesis deviates from the null at $\sqrt{n}$-rate, the limiting distribution of the test statistic, whether for the ordinary smooth or supersmooth case, adds a deterministic shift term to the corresponding null distribution, aligning with the case where the measurement error distribution is known. Furthermore, we proceed to analyze the limiting distribution of the test statistics under the alternative hypothesis \eqref{hyp.H1} to investigate the global power of our test.
\begin{theorem}\label{theorem.unknown under H1}
    Suppose Assumptions \ref{ass.D} and \ref{ass.D'} hold. Under $H_{1}$ in \eqref{hyp.H1}, if either Assumptions \ref{ass.O} and \ref{ass.O'} hold for the ordinary smooth case or Assumptions \ref{ass.S} and \ref{ass.S'} hold for the supersmooth case, 
    \begin{align*}
        &\sup_{\xi\in\Pi}\vert \hat{S}_{n}(\xi,\hat{\theta}_n,\hat{\sigma}^2_n) - S(\xi,\theta_0,\sigma^2_\ast) \vert = o_p\left(1\right),
    \end{align*}
    where $S(\xi,\theta_0,\sigma^2_\ast) = \mathbb{E}\{[(Y-g(X;\theta_0)^2-\sigma_\ast^2]\ee^{\ii X\xi}\}\neq 0$ for some $\xi\in\Pi$ with a positive measure. %$S(\xi,\theta_0,\sigma^2_\ast)$ is defined in Theorem \ref{theorem.known under H1}. %\eqref{theorem.known under H1 eq1}.
\end{theorem}

Theorem \ref{theorem.unknown under H1} shows that $\sqrt{n}\hat{S}_{n}(\xi,\hat{\theta}_n,\hat{\sigma}^2_n)$ diverges to infinity as $n\to\infty$ under $H_1$, %when the regressor $X$ and error $U$ are correlated, 
thereby confirming that the consistency of the tests is not affected by the absence of information about the measurement error distribution. To summarize, when the measurement error distribution is unknown, the convergence rates of the test statistics under the null, as well as their local and global powers, remain unchanged, thereby facilitating the bootstrap procedure developed in Section \ref{sec.boot}.

\section{Multiplier Bootstrap}\label{sec.boot}
Theorems in Sections \ref{sec.Asy} and \ref{sec.Unknown} demonstrate that the limiting processes in both cases are case-dependent; specifically, the variance functions depend on the underlying data distribution, which motivates the use of a bootstrap method. In this section, we derive the corresponding test statistics for implementing a computationally attractive multiplier bootstrap and obtain the critical values. 

In classical bootstrap methods (see, for example, \cite{koul1994bootstrapping}), researchers typically compute residuals from observed data and estimate their joint population distribution. Bootstrap samples are then generated by resampling from the estimated distribution, and the bootstrap test statistics are constructed in the same manner as the original statistics. This so-called residual-based bootstrap approach, however, faces substantial challenges in the presence of measurement error. As noted in \cite{dong2022nonparametric}, the unavailability of the true regressors prevents direct estimation of the residual distribution. Consequently, researchers have resorted to deconvolution kernel techniques, as suggested in \cite{hall2007testing}. Nonetheless, such methods suffer from slow convergence rates, difficulties in selecting tuning parameters, and high computational costs. These limitations motivate the development of alternative bootstrap procedures. Inspired by \cite{van1996weak}, we compute the corresponding term from each observation and multiply each of these terms by an independent mean-zero, unit-variance random variable before summing them to obtain the bootstrap version of the statistics. By performing $B$ rounds of multiplier resampling, we obtain $B$ bootstrap versions of the test statistics. 

Unfortunately, the application of the multiplier bootstrap is complicated by the \textquotedblleft parameter estimation effect\textquotedblright, as noted in \cite{durbin1973distribution}. Specifically, the previously proposed empirical process for the presence of measurement error can be rewritten as
\begin{align*}
    &S_n(\xi,\theta_n,\sigma^2_n) \\
    = & S_n(\xi,\theta_0,\sigma^2_0)+\frac{1}{n}\sum_{i=1}^n \int\left[\left(Y_i-g(x;\theta_n)\right)^2-\left(Y_i-g(x;\theta_0)\right)^2\right]\mathcal{K}_b\left(\frac{x-W_i}{b}\right)\ee^{\ii x\xi}\,dx\\
    &-\left(\sigma_n^2-\sigma_0^2\right)\left[\frac{1}{n}\sum_{i=1}^n \int\mathcal{K}_b\left(\frac{x-W_i}{b}\right)\ee^{\ii x\xi}\,dx\right]
\end{align*}
and a similar formulation for the case where the measurement error distribution is unknown. The first term, which corresponds to the empirical process constructed using the true parameter values, is asymptotically equivalent to the leading term in each case. As for the second term, we claim that it converges to zero at $\sqrt{n}$-rate in probability, under reasonable assumptions on the structural functions, kernels and bandwidth (as discussed in Sections \ref{sec.Asy} and \ref{sec.Unknown}), and by employing the parameter estimation method satisfying the conditions in Assumption \ref{ass.D}({\romannumeral2}). That is, although estimation uncertainty for $\theta_0$ exists, it does not affect the asymptotic behavior of the test statistics. The third term, commonly referred to as \textquotedblleft parameter estimation effect\textquotedblright, converges to the limiting process at rate $O_p(\sigma_n^2-\sigma_0^2)$ in the original statistic but at a slower rate during the bootstrap procedure. Henceforth, the notation $O_p(\sigma_n^2-\sigma_0^2)$ denotes the convergence rate in probability of the stochastic process to its limiting distribution, which is the same as the rate at which $\sigma^2_n$ converges in probability to $\sigma_0^2$. We address this issue by modifying the weight function $\ee^{\ii X\xi}$ in \eqref{stat.Population} to its mean-centered version in the construction of the bootstrap version of statistics, that is, subtracting the sample analog of its expectation,
\begin{align}\label{stat.known boot}
    S_n^{\ast}(\xi,\theta_n,\sigma^2_n) = \frac{1}{n}\sum_{i=1}^n V_i\int\left[\left(Y_i-g(x;\theta_n)\right)^2-\sigma_n^2\right]\mathcal{K}_b\left(\frac{x-W_i}{b}\right)\mathcal{P}_n(x;\xi)\,dx,
\end{align}
where %\textcolor{red}{$e-G_n$}
\begin{align}\label{stat.known projection}
    \mathcal{P}_n(x;\xi) = \ee^{\ii x\xi} -  \frac{1}{n}\sum_{i=1}^n \int\mathcal{K}_b\left(\frac{x-W_i}{b}\right)\ee^{\ii x\xi}\,dx,
\end{align}
and $\{V_i\}_{i=1}^n$ is a sequence of i.i.d. random variables with mean zero, variance one, bounded support, and independent of the sample $\{(Y_i,W_i)'\}_{i=1}^n$, e.g., \cite{mammen1993bootstrap}'s two-point distribution: $$
V_i=\left\{\begin{array}{ll}
	(1-\sqrt{5}) / 2, & \text { with probability }(\sqrt{5}+1) / 2 \sqrt{5}, \\
	(1+\sqrt{5}) / 2, & \text { with probability }(\sqrt{5}-1) / 2 \sqrt{5}.
\end{array}\right.
$$ 

The above construction ensures that $(\sigma_n^2-\sigma_0^2)\mathbb{E}(\ee^{\ii X\xi}-\mathbb{E}\ee^{\ii X\xi}) = 0$, which in turn asymptotically eliminates the influence of the \textquotedblleft parameter estimation effect\textquotedblright arising from $\sigma_n^2$. As a result, the bootstrap versions of empirical processes $S_n^{\ast}(\cdot,\theta_n,\sigma^2_n)$ converge to the same limiting processes as the original counterparts $S_n(\cdot,\theta_n,\sigma^2_n)$ for both cases under the null hypothesis, as shown in Theorem \ref{theorem.known boot}. In the bootstrap procedure, the test statistics $KS^{\ast}_n$ and $CvM_n^{\ast}$ can be constructed in a completely analogous way to those in Section \ref{sec.Asy}, based on the proposed $S_{n}^{\ast}(\cdot,\theta_n,\sigma^2_n)$. Specifically, we replace the original process $S_{n}(\cdot,\theta_n,\sigma^2_n)$ in $KS_n$ and $CvM_n$ with its bootstrap counterpart $S_{n}^{\ast}(\cdot,\theta_n,\sigma^2_n)$ and compute the sup norm and the squared $L_2$-norm accordingly. 

Defining \textquotedblleft $\overset{\ast}{\Longrightarrow}$\textquotedblright{} as weak convergence and $\mathbb{P}_n^\ast$ as the bootstrap probability under the bootstrap law, i.e., conditional on the original sample, see, e.g., Section 2.9 of \cite{van1996weak}, the validity of the proposed multiplier bootstrap is formally established by the following theorem.
\begin{theorem}\label{theorem.known boot}
    Under $H_0$ in \eqref{hyp.H0} or $H_{1n}$ in \eqref{hyp.H1n},
    \begin{align}\label{theorem.known boot ordi}
        \sqrt n S_{n}^{\ast}(\cdot,\theta_n,\sigma^2_n)\overset{\ast}{\Longrightarrow}S_{\infty}^{os}(\cdot,\theta_0,\sigma^2_0)
    \end{align}
    holds for the ordinary smooth case under Assumptions \ref{ass.D} and \ref{ass.O}, and
    \begin{align}\label{theorem.known boot super}
        \sqrt n S_{n}^{\ast}(\cdot,\theta_n,\sigma^2_n)\overset{\ast}{\Longrightarrow}S_{\infty}^{ss}(\cdot,\theta_0,\sigma^2_0)
    \end{align}
    holds for the supersmooth case under Assumptions \ref{ass.D} and \ref{ass.S}, where $S^{os}_\infty(\cdot,\theta_0,\sigma^2_0)$ and $S^{ss}_\infty(\cdot,\theta_0,\sigma^2_0)$ are the centered Gaussian processes as defined in Theorems \ref{theorem.known ordinary smooth under H0} and \ref{theorem.known supersmooth under H0}, respectively.
\end{theorem}
The centering adjustment of the weighting function does not affect the parametric-rate convergence of the bootstrap version of the empirical process to the limiting process. Moreover, the bootstrap version $S_{n}^{\ast}(\cdot,\theta_n,\sigma^2_n)$ converges to the same limiting process as that of the original test process $S_{n}(\cdot,\theta_n,\sigma^2_n)$ under the null hypothesis, for both the ordinary smooth and supersmooth cases, as established in Section \ref{sec.Asy}. Notably, the bootstrap limiting distributions for $KS^{\ast}_n$ and $CvM_n^{\ast}$ remain unchanged under both the null and local alternative hypothesis\footnote{The convergence results stated here follow directly from the continuous mapping theorem.}. By contrast, the limiting distribution of the test statistic under local alternatives includes an additional deterministic shift, which shifts its value away from the null distribution, making it more likely to exceed the bootstrap critical value and thus leading to the rejection of the null hypothesis. This ensures the validity of the bootstrap approach with respect to both size control and nontrivial local power. As a consequence of the above analysis, the asymptotic critical value at the significance level $\alpha$ is $c^\ast_{\alpha}=\inf\{c_\alpha\in[0,\infty):\lim_{n\to\infty}\mathbb{P}_n^\ast\{KS_n^\ast>c_\alpha\}=\alpha\}$, where we take $KS_n^\ast$ as an example and we also note that the bootstrap procedure for $CvM_n$ can be implemented in an analogous way. In practice, $c^\ast_{\alpha}$ can be approximated as $c^\ast_{n,\alpha} = \{KS_n^\ast\}_{B(1-\alpha)}$, the $B(1-\alpha)$-th order statistic for $B$ replicates $\{KS_n^\ast\}_{b=1}^B$ and we reject $H_0$ if $KS_n>c^\ast_{n,\alpha}$.

For the case where the measurement error distribution is unknown, empirical processes constructed in the manner above are not valid, as they fail to account for the additional terms arising from the estimation of the unknown characteristic function. Specifically, in the ordinary smooth case, directly replacing the deconvolution kernel in $S_{n}^{\ast}(\cdot,\theta_n,\sigma^2_n)$ by $\hat{\mathcal{K}}_b(\cdot)$ causes the zero-mean multipliers to eliminate the estimation effect brought by $\hat{f}^{\mathrm{ft}}_\epsilon(\cdot)$. Consequently, the bootstrapped empirical process converges to $S_\infty^{os}(\cdot,\theta_0,\sigma_0^2)$ rather than to the expected null limit $\hat{S}_\infty^{os}(\cdot,\theta_0,\sigma_0^2)$. This discrepancy invalidates the bootstrap approximation, and a similar issue persists in the supersmooth case. Fortunately, motivated by \cite{dong2022nonparametric}, we can perturb the estimator of the characteristic function using a set of multipliers with unit mean and unit variance. Specifically, we introduce a new sequence of multipliers $\{V^\ast_i\}_{i=1}^n$ satisfying $\mathbb{E}(V_i^\ast)=1$ and $\text{Var}(V_i^\ast)=1$ (e.g., the standard exponential distribution) and use them to construct perturbed analogs of the characteristic function estimator and the deconvolution kernel as in equation \eqref{stat.ME density},
\begin{align*}
    &\hat{f}^{\mathrm{ft}\ast}_\epsilon(t) = \left\vert\frac{1}{n}\sum\limits_{i=1}^nV_j^\ast\cos\left[t(W_i-W_i^r)\right]\right\vert^{1/2},\quad \hat{\mathcal{K}}^\ast_b(x) = \frac{1}{2\pi b}\int \ee^{-\ii tx}\frac{K^{\mathrm{ft}}(t)}{\hat{f}^{\mathrm{ft}\ast}_\epsilon(t/b)}\,dt.
\end{align*}
The resulting kernel estimator $\hat{\mathcal{K}}^\ast_b(\cdot)$ is then substituted into equations \eqref{stat.var unknown} and \eqref{stat.known projection} to obtain the bootstrap version of the variance estimator ${\hat{\sigma}_n^2}{}^\ast$ and the bootstrap empirical process, as defined by
\begin{align*}
    {\hat{\sigma}_n^2}{}^\ast = \frac{1}{n}\sum_{i=1}^n\int\left(Y_i-g(x;\hat{\theta}_n)\right)^2\hat{\mathcal{K}}^\ast_b\left(\frac{x-W_i}{b}\right)\,dx
\end{align*}
and
\begin{align}\label{stat.unknown boot}
    \hat{S}_n^\ast(\xi,\hat{\theta}_n,{\hat{\sigma}_n^2}{}^\ast) = \frac{1}{n}\sum_{i=1}^n\int\left[\left(Y_i-g(x;\hat{\theta}_n)\right)^2-{\hat{\sigma}_n^2}{}^\ast\right]\hat{\mathcal{K}}_b^\ast\left(\frac{x-W_i}{b}\right)\ee^{\ii x\xi}\,dx.
\end{align}
The following theorem confirms the validity of the bootstrap procedure even in the absence of information about the measurement error distribution, as in the case with a known measurement error distribution.
\begin{theorem}\label{theorem.unknown boot}
    Suppose Assumptions \ref{ass.D} and \ref{ass.D'} hold. %Under both the null and local alternative hypotheses, 
    Under $H_0$ in \eqref{hyp.H0} or $H_{1n}$ in \eqref{hyp.H1n},
    \begin{align}\label{theorem.unknown boot ordi}
        \sqrt n \left(\hat{S}^{\ast}_{n}(\cdot,\hat\theta_n,{\hat{\sigma}_n^2}{}^\ast)-\mathbb{E}\left[\hat{S}^{\ast}_{n}(\cdot,\hat\theta_n,{\hat{\sigma}_n^2}{}^\ast)\right]\right)\overset{\ast}{\Longrightarrow}\hat{S}_{\infty}^{os}(\cdot,\theta_0,\sigma_0^2)
    \end{align}
    holds for the ordinary smooth case under Assumptions \ref{ass.O} and \ref{ass.O'}, and
    \begin{align}\label{theorem.unknown boot super}
        \sqrt n \left(\hat{S}^{\ast}_{n}(\cdot,\hat\theta_n,{\hat{\sigma}_n^2}{}^\ast)-\mathbb{E}\left[\hat{S}^{\ast}_{n}(\cdot,\hat\theta_n,{\hat{\sigma}_n^2}{}^\ast)\right]\right)\overset{\ast}{\Longrightarrow}\hat{S}_{\infty}^{ss}(\cdot,\theta_0,\sigma_0^2)
    \end{align}
    holds for the supersmooth case under Assumptions \ref{ass.S} and \ref{ass.S'}, where $\hat{S}^{os}_\infty(\cdot,\theta_0,\sigma_0^2)$ and $\hat{S}^{ss}_\infty(\cdot,\theta_0,\sigma_0^2)$ are the centered Gaussian processes as defined in Theorems \ref{theorem.unknown ordinary smooth under H0} and \ref{theorem.unknown supersmooth under H0}, respectively.
\end{theorem}

Under the null or the local alternatives, the bootstrap versions of the empirical processes with a centering adjustment for the ordinary smooth and supersmooth cases converge to the same limiting processes $\hat{S}^{os}_\infty(\cdot,\theta_0,\sigma_0^2)$ and $\hat{S}^{ss}_\infty(\cdot,\theta_0,\sigma_0^2)$ as those of the original statistics constructed in Section \ref{sec.Unknown}, respectively. A key implication of the above theorem is that the bootstrap versions of statistics can be constructed based on centralized correction of $\hat{S}^{\ast}_{n}(\cdot,\hat\theta_n,{\hat{\sigma}_n^2}{}^\ast)$, 
\begin{align*}
&\widehat{KS}_n^\ast=\sqrt{n}\sup_{\xi\in\Pi}\left\vert\hat{S}^{\ast}_{n}(\xi,\hat\theta_n,{\hat{\sigma}_n^2}{}^\ast)-\mathbb{E}\left[\hat{S}^{\ast}_{n}(\xi,\hat\theta_n,{\hat{\sigma}_n^2}{}^\ast)\right]\right\vert,\\
&\widehat{CvM}_n^\ast=n
\int_\Pi \left\vert\hat{S}^{\ast}_{n}(\xi,\hat\theta_n,{\hat{\sigma}_n^2}{}^\ast)-\mathbb{E}\left[\hat{S}^{\ast}_{n}(\xi,\hat\theta_n,{\hat{\sigma}_n^2}{}^\ast)\right]\right\vert^2\,d\xi.
\end{align*}
For the same reasons discussed earlier, these bootstrap statistics converge, under both the null hypothesis and local alternatives, to the supremum norm and the squared $L_2$-norm of the corresponding limiting Gaussian processes $\hat{S}_{\infty}^{os}(\cdot,\theta_0,\sigma_0^2)$ for the ordinary smooth case and $\hat{S}_{\infty}^{ss}(\cdot,\theta_0,\sigma_0^2)$ for the supersmooth case, respectively. Combined with the fact established in Section \ref{sec.Unknown} that, under local alternatives, the limiting distributions of the statistics are perturbed by a deterministic shift function, this result ensures the validity of the proposed bootstrap procedure.

It is worth highlighting the projection-based approach proposed in \cite{sant2019specification}, \cite{yang2024model}, and \cite{song2025unified}, which intuitively eliminates the \textquotedblleft parameter estimation effect\textquotedblright{} term by the orthogonal projection of the weighting function onto the tangent space of nuisance parameters in probability space. From this perspective, the centering method employed in this paper corresponds to projecting the weighting function onto the constant function $1$. It is also worth noting that, compared with the wild bootstrap methods discussed in other studies, our approach offers the significant advantage of greater computational efficiency. Moreover, as will be discussed in detail in Section \ref{sec.Simulation}, the selection of tuning parameters such as the bandwidth does not pose practical difficulties when using the multiplier bootstrap.

\section{Simulation Study}\label{sec.Simulation}
In the simulations that follow, we aim to examine several key aspects. First, we evaluate whether the proposed tests, under both known and unknown measurement error distributions, maintain size close to the nominal level and exhibit reasonable power. Second, we investigate whether test performance differs substantially between the ordinary smooth and supersmooth cases. Third, we explore how different model specifications influence the validity of heteroskedasticity diagnostics. Finally, we assess whether the choice of bandwidth has a noticeable impact on the size and power characteristics of the tests. 

We start with a linear model specification, hereafter referred to as model $1$. In this setting, data are generated according to $Y=\alpha_0+\alpha_1X+\sigma U$, where we set $\alpha_0=\alpha_1=1$, and assume that both $X$ and $U$ follow independent normal distributions with mean zero and unit variance. Heteroskedasticity is introduced through $\sigma$. To investigate the performance of the proposed tests, we consider three distinct data-generating processes (DGPs) that represent both the null and alternative hypotheses. 
\begin{align*}
    &\text{DGP$(0)$: } \sigma^2(X_i) = 1,\\
    &\text{DGP$(1)$: } \sigma^2(X_i) = 1+\vert \cos \pi X_i\vert^2,\\
    &\text{DGP$(2)$: } \sigma^2(X_i) = 1 + \exp{\left\vert X_i\right\vert}.
\end{align*}
Specifically, DGP(0) corresponds to data generated under the null hypothesis, while DGP(1) and DGP(2) represent alternative scenarios with relatively high-frequency and low-frequency deviations from the null, respectively. In practical measurement error problems, the covariate $X$ is unobservable. We introduce additive noise $\epsilon\sim N(0,1/3)$, independent of $X$, such that the signal-to-noise ratio, defined as $Var(U)/Var(\epsilon)$, equals $3$. The observed contaminated variable is then given by $W=X + \epsilon$. We employ an infinite-order kernel defined by its Fourier transformation that is widely used in the measurement error literature and has been shown to be effective in a variety of applications (see, for example, \cite{mcmurry2004nonparametric}, \cite{dong2021average}, and \cite{dong2022nonparametric}):
\begin{equation*}
K^{\mathrm{ft}}(t) = \begin{cases} 
1 & \text{if } |t| \leq 0.05, \\
\exp \left\{ \frac{-\exp(-(|t|-0.05)^{-2})}{(|t|-1)^2} \right\} & \text{if } 0.05 < |t| < 1, \\
0 & \text{if } |t| \geq 1.
\end{cases}
\end{equation*}
We estimate the parameters $\alpha_0$ and $\alpha_1$ using the method proposed in \cite{cheng1998polynomial}. It is worth noting that under the null hypothesis, in which the model exhibits homoskedasticity, the estimators can be applied directly. Under the alternative hypothesis, the omission of heteroskedasticity in the estimation procedure does not compromise the performance of the test, since, as discussed in Section \ref{sec.boot}, the estimation of $\alpha_0$ and $\alpha_1$ does not introduce any additional \textquotedblleft parameter estimation effect\textquotedblright{} on the asymptotic behavior of the test statistics. Subsequently, the bandwidth is selected based on the rule-of-thumb proposed in \cite{delaigle2008deconvolution} and evaluated over a grid of candidate values, with the aim of minimizing pointwise mean squared error of the estimated characteristic function of $f^{\mathrm{ft}}_\epsilon(\cdot)$ in the unknown measurement error setting. Specifically, we set $b=c(5\sigma^4/n)^{1/27}$ for the ordinary smooth case and $b=c(4\sigma^2/\log(n))^{1/2}$ for the supersmooth case, where $c$ takes integer values from $1$ to $10$. To evaluate the robustness of the proposed test, we consider sample sizes $n\in\{250,500,1000\}$ and nominal significance levels $\alpha\in\{0.01,0.05,0.1\}$. Critical values are obtained using $199$ bootstrap replications, and $1000$ times Monte Carlo experiments. Due to space limitations, we present only the results for $\alpha=0.05$ and $c\in\{0.1,0.5,1\}$ with the remaining results reported in the Appendix \ref{sec.AppendixA} on the online supplementary appendix.
\begin{table}[htbp]
\centering
\caption{Results for $S_n(\xi,\theta_n,\sigma_n^2)$, model $1$, $\alpha=0.05$}
\begin{tabular}{cccccccc}
\hline
\multicolumn{2}{c}{\textbf{Ordinary Smooth}} & \multicolumn{6}{c}{} \\
\hline
\multirow{2}{*}{$n$} & \multirow{2}{*}{$c$} & \multicolumn{2}{c}{DGP(0)} & \multicolumn{2}{c}{DGP(1)} & \multicolumn{2}{c}{DGP(2)}\\ 
\cline{3-8}
& & KS & CvM & KS & CvM & KS & CvM \\
\hline
$n=500$ & 0.1 & 0.058 & 0.058 & 0.897 & 0.887 & 0.997 & 0.999 \\
& 0.5 & 0.042 & 0.038 & 0.896 & 0.884 & 0.999 & 1.000 \\
& 1 & 0.050 & 0.046 & 0.909 & 0.894 & 0.998 & 0.999 \\
\hline
$n=1000$ & 0.1 & 0.043 & 0.040 & 0.997 & 0.996 & 1.000 & 1.000 \\
& 0.5 & 0.048 & 0.052 & 0.996 & 0.998 & 1.000 & 1.000 \\
& 1 & 0.052 & 0.048 & 0.993 & 0.989 & 1.000 & 1.000 \\
\hline
\multicolumn{2}{c}{\textbf{Supersmooth}} & \multicolumn{6}{c}{} \\
\hline
\multirow{2}{*}{$n$} & \multirow{2}{*}{$c$} & \multicolumn{2}{c}{DGP(0)} & \multicolumn{2}{c}{DGP(1)} & \multicolumn{2}{c}{DGP(2)} \\
\cline{3-8}
& & KS & CvM & KS & CvM & KS & CvM \\
\hline
$n=500$ & 0.1 & 0.042 & 0.040 & 0.976 & 0.973 & 1.000 & 1.000 \\
& 0.5 & 0.060 & 0.062 & 0.977 & 0.976 & 1.000 & 1.000 \\
& 1 & 0.043 & 0.047 & 0.980 & 0.979 & 1.000 & 1.000 \\
\hline
$n=1000$ & 0.1 & 0.067 & 0.065 & 1.000 & 1.000 & 1.000 & 1.000 \\
& 0.5 & 0.054 & 0.054 & 1.000 & 1.000 & 1.000 & 1.000 \\
& 1 & 0.049 & 0.049 & 1.000 & 1.000 & 1.000 & 1.000 \\
\hline
\end{tabular}
\label{tab.known mod1}
\end{table}

\begin{table}[htbp]
\centering
\caption{Results for $\hat{S}_n(\xi,\hat{\theta}_n,\hat{\sigma}_n^2)$, model $1$, $\alpha=0.05$}
\begin{tabular}{cccccccc}
\hline
\multicolumn{2}{c}{\textbf{Ordinary Smooth}} & \multicolumn{6}{c}{} \\
\hline
\multirow{2}{*}{$n$} & \multirow{2}{*}{$c$} & \multicolumn{2}{c}{DGP(0)} & \multicolumn{2}{c}{DGP(1)} & \multicolumn{2}{c}{DGP(2)}\\ 
\cline{3-8}
& & KS & CvM & KS & CvM & KS & CvM \\
\hline
$n=500$ & 0.1 & 0.066 & 0.060 & 0.894 & 0.886 & 1.000 & 1.000 \\
& 0.5 & 0.062 & 0.059 & 0.901 & 0.892 & 0.996 & 0.998 \\
& 1 & 0.055 & 0.057 & 0.917 & 0.906 & 0.997 & 0.999 \\
\hline
$n=1000$ & 0.1 & 0.045 & 0.045 & 0.997 & 0.995 & 1.000 & 1.000 \\
& 0.5 & 0.054 & 0.059 & 0.996 & 0.996 & 1.000 & 1.000 \\
& 1 & 0.046 & 0.050 & 0.995 & 0.992 & 1.000 & 1.000 \\
\hline
\multicolumn{2}{c}{\textbf{Supersmooth}} & \multicolumn{6}{c}{} \\
\hline
\multirow{2}{*}{$n$} & \multirow{2}{*}{$c$} & \multicolumn{2}{c}{DGP(0)} & \multicolumn{2}{c}{DGP(1)} & \multicolumn{2}{c}{DGP(2)} \\
\cline{3-8}
& & KS & CvM & KS & CvM & KS & CvM \\
\hline
$n=500$ & 0.1 & 0.057 & 0.059 & 0.938 & 0.942 & 0.999 & 1.000 \\
& 0.5 & 0.043 & 0.044 & 0.923 & 0.926 & 0.999 & 0.999 \\
& 1 & 0.055 & 0.055 & 0.929 & 0.927 & 0.997 & 1.000 \\
\hline
$n=1000$ & 0.1 & 0.051 & 0.052 & 0.995 & 0.996 & 1.000 & 1.000 \\
& 0.5 & 0.068 & 0.068 & 0.996 & 0.996 & 1.000 & 1.000 \\
& 1 & 0.068 & 0.069 & 0.998 & 0.997 & 1.000 & 1.000 \\
\hline
\end{tabular}
\label{tab.unknown mod1}
\end{table}
Tables \ref{tab.known mod1} and \ref{tab.unknown mod1} report the simulation results under known and unknown measurement error distributions, respectively. The proposed tests exhibit desirable size accuracy and power properties across all considered cases. Specifically, as the sample size increases, the empirical size approaches the nominal level, and the power improves accordingly. A comparison between the results under DGP(1) and DGP(2) reveals that the test tends to be more powerful against low-frequency alternatives, which is consistent with results reported in the literature employing global test procedures. In our test, the convergence rate of the test statistic under the supersmooth case is $\sqrt{n}$, the same as that under the ordinary smooth case, which allows the supersmooth case to avoid the typical drawback of reduced power often reported in the literature. Interestingly, we observe that the test exhibits slightly higher power in the supersmooth case. We conjecture that this improvement stems from the test statistic and the associated plug-in estimators being exactly unbiased under the supersmooth setting, rather than merely asymptotically unbiased in the ordinary smooth case. To examine this conjecture, we further consider a constant model, where the data are generated from $Y=\alpha_0+\sigma U$, and conduct simulations under the same experimental settings as in the linear model.
\begin{table}[htbp]
\centering
\caption{Results for $S_n(\xi,\theta_n,\sigma_n^2)$, model $2$, $\alpha=0.05$}
\begin{tabular}{cccccccc}
\hline
\multicolumn{2}{c}{\textbf{Ordinary Smooth}} & \multicolumn{6}{c}{} \\
\hline
\multirow{2}{*}{$n$} & \multirow{2}{*}{$c$} & \multicolumn{2}{c}{DGP(0)} & \multicolumn{2}{c}{DGP(1)} & \multicolumn{2}{c}{DGP(2)}\\ 
\cline{3-8}
& & KS & CvM & KS & CvM & KS & CvM \\
\hline
$n=500$ & 0.1 & 0.052 & 0.052 & 0.977 & 0.979 & 1.000 & 1.000 \\
& 0.5 & 0.054 & 0.056 & 0.981 & 0.981 & 1.000 & 1.000 \\
& 1 & 0.047 & 0.051 & 0.984 & 0.985 & 1.000 & 1.000 \\
\hline
$n=1000$ & 0.1 & 0.053 & 0.051 & 1.000 & 1.000 & 1.000 & 1.000 \\
& 0.5 & 0.043 & 0.049 & 1.000 & 1.000 & 1.000 & 1.000 \\
& 1 & 0.057 & 0.055 & 1.000 & 1.000 & 1.000 & 1.000 \\
\hline
\multicolumn{2}{c}{\textbf{Supersmooth}} & \multicolumn{6}{c}{} \\
\hline
\multirow{2}{*}{$n$} & \multirow{2}{*}{$c$} & \multicolumn{2}{c}{DGP(0)} & \multicolumn{2}{c}{DGP(1)} & \multicolumn{2}{c}{DGP(2)} \\
\cline{3-8}
& & KS & CvM & KS & CvM & KS & CvM \\
\hline
$n=500$ & 0.1 & 0.042 & 0.040 & 0.976 & 0.973 & 1.000 & 1.000 \\
& 0.5 & 0.060 & 0.062 & 0.977 & 0.976 & 1.000 & 1.000 \\
& 1 & 0.043 & 0.047 & 0.980 & 0.979 & 1.000 & 1.000 \\
\hline
$n=1000$ & 0.1 & 0.067 & 0.065 & 1.000 & 1.000 & 1.000 & 1.000 \\
& 0.5 & 0.054 & 0.054 & 1.000 & 1.000 & 1.000 & 1.000 \\
& 1 & 0.049 & 0.049 & 1.000 & 1.000 & 1.000 & 1.000 \\
\hline
\end{tabular}
\label{tab.known mod2}
\end{table}

\begin{table}[htbp]
\centering
\caption{Results for $\hat{S}_n(\xi,\hat{\theta}_n,\hat{\sigma}_n^2)$, model $2$, $\alpha=0.05$}
\begin{tabular}{cccccccc}
\hline
\multicolumn{2}{c}{\textbf{Ordinary Smooth}} & \multicolumn{6}{c}{} \\
\hline
\multirow{2}{*}{$n$} & \multirow{2}{*}{$c$} & \multicolumn{2}{c}{DGP(0)} & \multicolumn{2}{c}{DGP(1)} & \multicolumn{2}{c}{DGP(2)}\\ 
\cline{3-8}
& & KS & CvM & KS & CvM & KS & CvM \\
\hline
$n=500$ & 0.1 & 0.048 & 0.049 & 0.972 & 0.973 & 1.000 & 1.000 \\
& 0.5 & 0.056 & 0.058 & 0.984 & 0.984 & 1.000 & 1.000 \\
& 1 & 0.049 & 0.045 & 0.982 & 0.982 & 1.000 & 1.000 \\
\hline
$n=1000$ & 0.1 & 0.055 & 0.054 & 1.000 & 1.000 & 1.000 & 1.000 \\
& 0.5 & 0.048 & 0.048 & 1.000 & 1.000 & 1.000 & 1.000 \\
& 1 & 0.056 & 0.059 & 1.000 & 1.000 & 1.000 & 1.000 \\
\hline
\multicolumn{2}{c}{\textbf{Supersmooth}} & \multicolumn{6}{c}{} \\
\hline
\multirow{2}{*}{$n$} & \multirow{2}{*}{$c$} & \multicolumn{2}{c}{DGP(0)} & \multicolumn{2}{c}{DGP(1)} & \multicolumn{2}{c}{DGP(2)} \\
\cline{3-8}
& & KS & CvM & KS & CvM & KS & CvM \\
\hline
$n=500$ & 0.1 & 0.042 & 0.039 & 0.981 & 0.979 & 1.000 & 1.000 \\
& 0.5 & 0.049 & 0.051 & 0.974 & 0.976 & 0.999 & 1.000 \\
& 1 & 0.051 & 0.051 & 0.979 & 0.976 & 1.000 & 1.000 \\
\hline
$n=1000$ & 0.1 & 0.062 & 0.062 & 1.000 & 1.000 & 1.000 & 1.000 \\
& 0.5 & 0.065 & 0.063 & 1.000 & 1.000 & 1.000 & 1.000 \\
& 1 & 0.069 & 0.067 & 1.000 & 1.000 & 1.000 & 1.000 \\
\hline
\end{tabular}
\label{tab.unknown mod2}
\end{table}
As shown in Tables \ref{tab.known mod2} and \ref{tab.unknown mod2}, the test exhibits comparable power across the ordinary smooth and supersmooth cases, which supports our conjecture that, in small samples and under the parametric model setting, the exact unbiasedness of the plug-in estimators in the supersmooth case leads to improved performance. Finally, different bandwidth choices across the grid points do not substantially affect the size and power of the test, confirming the desirable robustness to tuning-parameter selection as discussed in Sections \ref{sec.Asy} and \ref{sec.Unknown}.

Overall, the simulation results demonstrate that our proposed test is particularly sensitive to low-frequency alternatives, exhibits improved power for the supersmooth case, and is robust to bandwidth selection.

\section{Empirical Application}\label{sec.Examples}
In this section, we illustrate the practical utility of our proposed tests by applying them to classic datasets from the measurement error literature. We first focus on detecting heteroskedasticity in the relationship between yields of corn and determinations of available soil nitrogen collected on Marshall soil in Iowa, a context where measurement error typically arising because only a small soil sample is taken from each plot and because of noise in the chemical analysis used to determine the nitrogen level in the soil, is often taken into account (see \cite{fuller2009measurement}), whereas heteroskedasticity in the regression equation caused by differences in sampling times and conditions is often ignored. We use the dataset originally presented in \cite{fuller2009measurement} (Tables $1.2.1$ and $3.1.1$). This dataset comprises observations on corn yield $Y$ and available soil nitrogen $X$ for a series of experimental plots. The primary econometric objective is to estimate the production relationship:
\begin{align*}
    Y_i = \alpha_0+\alpha_1X_i+U_i, \quad i=1,2,\cdots,n,
\end{align*}
where $X_i$ is the true soil nitrogen content but is contaminated by measurement error, so that the researcher only observes the variable $W_i$, which is generated as $W_i = X_i + \epsilon_i$, with $\epsilon_i$ denoting the measurement error.

Standard analyses of this dataset (e.g., Ordinary Least Squares) typically ignore the measurement error, leading to attenuation bias in the slope coefficient. More sophisticated errors-in-variables (EIV) methods account for the bias but often maintain the strong assumption of homoskedasticity, $(Var(\epsilon_i\vert X_i)=\sigma_0^2)$. If the variance of the error term depends on the level of nitrogen (e.g., higher variability in yield at higher nitrogen levels), the standard errors of these EIV estimators may be invalid, leading to incorrect conclusions mentioned in Section \ref{sec.Intro}. Therefore, testing for heteroskedasticity in the presence of measurement error is a crucial diagnostic step. Specifically, we consider three scenarios corresponding to the data structures available in \cite{fuller2009measurement}: known measurement error variance where we make use of the information $\sigma_{\epsilon_i}=57$ provided in Example $1.2.1$ of \cite{fuller2009measurement} and specify the measurement error distribution to be Laplace and normal, representing ordinary smooth and supersmooth errors, respectively, and unknown measurement error distribution where we utilize the replicate measurements provided in Table $3.1.1$. We apply our proposed KS and CVM-type tests to this dataset. To evaluate the sensitivity of our procedure to tuning parameters and distributional assumptions, we report $p-$values across a range of bandwidth parameters $c$, and the results are summarized in Table \ref{tab.case1}.
\begin{table}[htbp]
\centering
\caption{Test results on heteroskedasticity in the regression of corn yields on soil nitrogen content for Marshall soil in Iowa}
\begin{tabular}{ccccccc}
\hline
\multirow{2}{*}{$c$} & \multicolumn{2}{c}{Laplace} & \multicolumn{2}{c}{Normal} & \multicolumn{2}{c}{Unknown}\\ 
\cline{2-7}
 & KS & CvM & KS & CvM & KS & CvM \\
\hline
0.1 & 0.080 & 0.080 & 0.090 & 0.090 & 0.045 & 0.030 \\
0.2 & 0.106 & 0.116 & 0.090 & 0.090 & 0.065 & 0.065 \\
0.3 & 0.095 & 0.095 & 0.065 & 0.065 & 0.121 & 0.116 \\
0.4 & 0.065 & 0.070 & 0.055 & 0.055 & 0.075 & 0.075 \\
0.5 & 0.070 & 0.085 & 0.070 & 0.070 & 0.070 & 0.070 \\
0.6 & 0.085 & 0.080 & 0.075 & 0.075 & 0.065 & 0.065 \\
0.7 & 0.106 & 0.101 & 0.095 & 0.095 & 0.060 & 0.060 \\
0.8 & 0.090 & 0.090 & 0.050 & 0.055 & 0.075 & 0.075 \\
0.9 & 0.080 & 0.085 & 0.080 & 0.080 & 0.060 & 0.060 \\
1 & 0.080 & 0.080 & 0.060 & 0.060 & 0.090 & 0.090 \\
\hline
\end{tabular}
\label{tab.case1}
\end{table}

First, evidence of heteroskedasticity can be drawn from the empirical results presented in Table \ref{tab.case1}. Specifically, if we adopt a standard significance level of $\alpha = 0.10$, the null hypothesis of homoskedasticity is rejected in most experimental settings. This suggests that the conditional variance of corn yield varies with soil nitrogen levels. Ignoring this heteroskedasticity, as is common in standard linear EIV regressions applied to this data, could lead to inefficient estimation and misleading confidence intervals. Our test provides a robust means of detecting latent structures otherwise obscured by measurement error. Second, we find robustness to distributional assumptions, bandwidth choice, and statistics selection. In practical applications, the precise distribution of the measurement error is rarely known. Fortunately, a comparison between the \textquotedblleft Laplace\textquotedblright{} and \textquotedblleft Normal\textquotedblright{} columns reveals that the resulting $p$-values are largely insensitive to this misspecification, although the theoretical derivation of the deconvolution kernel depends on the error distribution. Furthermore, consistent with the Monte Carlo simulations in Section \ref{sec.Simulation}, the proposed tests exhibit robustness with respect to a wide range of bandwidth constants $c$. Both statistics yield comparable inferences, reinforcing the reliability of the testing procedure. Third, the results in the \textquotedblleft Unknown\textquotedblright{} columns demonstrate that our test remains powerful even when the error variance is not known and must be estimated via repeated measurements (as detailed in Section \ref{sec.Unknown}). However, it is worth noting that the admissible range for the bandwidth choice appears slightly narrower than in the known-variance case. This is an expected consequence of the additional estimation noise introduced by estimating the characteristic function of the error term, suggesting that practitioners should exercise caution in bandwidth selection when relying on replicates.

To further demonstrate the applicability of our test for detecting latent heteroskedasticity, we consider a second application that focuses on estimating Engel curves. This analysis utilizes data from the 2023 Consumer Expenditure Survey (CES), a dataset widely used in demand analysis (e.g., \cite{hausman1995nonlinear}). We adopt the standard Leser--Working functional form (\cite{leser1963forms}), which relates the budget share of a commodity to the logarithm of total expenditure. For household $i$ and commodity $j$ in quarter $t$, the structural equation is specified as:
\begin{align*}
    Y_{ijt} = \alpha_{0j}+\alpha_{1j}X_{it}+U_{ijt}
\end{align*}
where $Y_{ijt}$ represents the budget share of commodity $j$ (e.g., Food, Clothing), $X_{it}$ denotes the true log total expenditure and the error term $U_{ijt}$ captures heterogeneity in preferences. However, total expenditure is notoriously difficult to measure accurately in survey data. Following the literature (e.g., \cite{hausman1995nonlinear}), we assume that the observed log total expenditure, $W_{it}$, measures the true latent log total expenditure $X_{it}$ with an additive error:
\begin{align*}
    W_{it} = X_{it} + \epsilon_{it},
\end{align*}
where $\epsilon_{it}$ represents the measurement error. Since the variance of the measurement error is unknown, we adopt the repeated measurement outlined in Section \ref{sec.Unknown}. Specifically, we treat the log total expenditure reported in the subsequent quarter $x_{i, t+1}$ as a replicate measurement for the current quarter $t$. This allows us to estimate the characteristic function of the measurement error and construct the test statistic without making parametric assumptions about the error distribution. We analyze five major expenditure categories: Food, Clothing, Recreation, Health care, and Transportation. The tests are conducted separately for the second, third, and fourth quarters (Q2, Q3, Q4) of 2023 to check for temporal stability. The results of the $p$-values for the KS and CvM test statistics are presented in Table \ref{tab.case2}.
\begin{table}[htbp]
\centering
\caption{Heteroskedasticity test results for regressions of budget shares on log expenditure based on Leser-working Engel curve specification using $2023$ CES data}
\begin{tabular}{lcccccc}
\hline
\multirow{2}{*}{Budget names} & \multicolumn{2}{c}{Q2} & \multicolumn{2}{c}{Q3} & \multicolumn{2}{c}{Q4}\\ 
\cline{2-7}
 & KS & CvM & KS & CvM & KS & CvM \\
\hline
Food & 0.603 & 0.628 & 0.763 & 0.854 & 0.568 & 0.658 \\
Clothing & 0.367 & 0.497 & 0.653 & 0.633 & 0.271 & 0.286 \\
Recreation & 0.166 & 0.201 & 0.050 & 0.025 & 0.136 & 0.151 \\
Health care & 0.201 & 0.307 & 0.518 & 0.372 & 0.196 & 0.251 \\
Transportation & 0.010 & 0.015 & 0.000 & 0.000 & 0.010 & 0.015 \\
\hline
\end{tabular}
\label{tab.case2}
\end{table}

The empirical results in Table \ref{tab.case2} reveal distinct patterns across commodity groups. Most notably, \textquotedblleft Transportation\textquotedblright{} exhibits strong evidence of heteroskedasticity. The $p$-values are consistently close to zero across all quarters for both KS and CvM statistics. This rejection of the null hypothesis aligns with the findings in \cite{hausman1995nonlinear}, who argued that the simple linear-in-log Leser-Working specification is insufficient for certain goods, particularly Transportation. They found statistically significant coefficients for higher-order polynomial terms (quadratic in log expenditure). If the true relationship contains a quadratic term (e.g., $\alpha_2 X_{it}^2$) but is modeled linearly, the omitted nonlinear component becomes part of the error term, causing the error variance to vary with the level of expenditure. Therefore, our robust rejection of the homoskedasticity null hypothesis corroborates the need for higher-order Engel curve specifications as suggested by \cite{hausman1995nonlinear}.

%\newpage

\bibliographystyle{apalike_revised}
\bibliography{Ref}

\newpage

\begin{center}
\LARGE{Testing Heteroskedasticity Under Measurement Error}
\end{center}

\begin{center}
\LARGE{Online Supplementary Appendix}
\end{center}

\begin{center}
\Large{Xiaojun Song and Jichao Yuan}
\end{center}

\begin{center}
\Large{Peking University}
\end{center}

In this online supplementary appendix, we report additional simulation results in Section \ref{sec.AppendixA}, introduce necessary notations and definitions in Section \ref{sec.AppendixB}, and provide proofs of the main theoretical results in Section \ref{sec.AppendixC}. Finally, auxiliary lemmas and their proofs are collected in Section \ref{sec.AppendixD}.

\appendix
\section{Supplementary Simulation Results}\label{sec.AppendixA}
\begin{table}[htbp]
\centering
\caption{Results for $S_n(\xi,\theta_n,\sigma^2_n)$, model $1$, ordinary smooth case, $\alpha=0.01$}
\begin{tabular}{cccccccc}
\hline
\multirow{2}{*}{$n$} & \multirow{2}{*}{$10c$} & \multicolumn{2}{c}{DGP(0)} & \multicolumn{2}{c}{DGP(1)} & \multicolumn{2}{c}{DGP(2)} \\
\cline{3-8}
& & KS & CvM & KS & CvM & KS & CvM \\
\hline
\multirow{10}{*}{$n=250$} & 1 & 0.008 & 0.010 & 0.364 & 0.354 & 0.733 & 0.768 \\
& 2 & 0.017 & 0.017 & 0.376 & 0.369 & 0.762 & 0.800 \\
& 3 & 0.007 & 0.004 & 0.353 & 0.346 & 0.748 & 0.786 \\
& 4 & 0.012 & 0.012 & 0.407 & 0.402 & 0.727 & 0.764 \\
& 5 & 0.012 & 0.011 & 0.366 & 0.360 & 0.759 & 0.783 \\
& 6 & 0.017 & 0.016 & 0.389 & 0.389 & 0.751 & 0.786 \\
& 7 & 0.007 & 0.007 & 0.360 & 0.347 & 0.738 & 0.768 \\
& 8 & 0.004 & 0.003 & 0.361 & 0.350 & 0.791 & 0.814 \\
& 9 & 0.016 & 0.015 & 0.359 & 0.351 & 0.740 & 0.766 \\
& 10 & 0.008 & 0.007 & 0.356 & 0.352 & 0.747 & 0.775 \\
\hline
\multirow{10}{*}{$n=500$} & 1 & 0.014 & 0.013 & 0.752 & 0.740 & 0.990 & 0.989 \\
& 2 & 0.015 & 0.011 & 0.728 & 0.713 & 0.991 & 0.993 \\
& 3 & 0.012 & 0.013 & 0.741 & 0.723 & 0.987 & 0.991 \\
& 4 & 0.016 & 0.014 & 0.725 & 0.717 & 0.989 & 0.994 \\
& 5 & 0.010 & 0.008 & 0.749 & 0.731 & 0.990 & 0.994 \\
& 6 & 0.017 & 0.015 & 0.779 & 0.759 & 0.993 & 0.997 \\
& 7 & 0.009 & 0.009 & 0.745 & 0.725 & 0.989 & 0.990 \\
& 8 & 0.016 & 0.015 & 0.738 & 0.725 & 0.987 & 0.993 \\
& 9 & 0.011 & 0.010 & 0.746 & 0.726 & 0.992 & 0.996 \\
& 10 & 0.014 & 0.016 & 0.751 & 0.730 & 0.990 & 0.997 \\
\hline
\multirow{10}{*}{$n=1000$} & 1 & 0.017 & 0.018 & 0.974 & 0.968 & 1.000 & 1.000 \\
& 2 & 0.010 & 0.014 & 0.975 & 0.972 & 1.000 & 1.000 \\
& 3 & 0.009 & 0.008 & 0.975 & 0.971 & 1.000 & 1.000 \\
& 4 & 0.011 & 0.012 & 0.986 & 0.982 & 1.000 & 1.000 \\
& 5 & 0.017 & 0.017 & 0.969 & 0.962 & 1.000 & 1.000 \\
& 6 & 0.016 & 0.015 & 0.978 & 0.975 & 1.000 & 1.000 \\
& 7 & 0.017 & 0.018 & 0.973 & 0.967 & 1.000 & 1.000 \\
& 8 & 0.015 & 0.013 & 0.974 & 0.967 & 1.000 & 1.000 \\
& 9 & 0.009 & 0.011 & 0.974 & 0.969 & 1.000 & 1.000 \\
& 10 & 0.020 & 0.019 & 0.977 & 0.974 & 1.000 & 1.000 \\
\hline
\end{tabular}
\end{table}

\begin{table}[htbp]
\centering
\caption{Results for $S_n(\xi,\theta_n,\sigma^2_n)$, model $1$, ordinary smooth case, $\alpha=0.05$}
\begin{tabular}{cccccccc}
\hline
\multirow{2}{*}{$n$} & \multirow{2}{*}{$10c$} & \multicolumn{2}{c}{DGP(0)} & \multicolumn{2}{c}{DGP(1)} & \multicolumn{2}{c}{DGP(2)} \\
\cline{3-8}
& & KS & CvM & KS & CvM & KS & CvM \\
\hline
\multirow{10}{*}{$n=250$} & 1 & 0.066 & 0.065 & 0.630 & 0.620 & 0.926 & 0.939 \\
& 2 & 0.063 & 0.059 & 0.600 & 0.604 & 0.894 & 0.927 \\
& 3 & 0.051 & 0.050 & 0.593 & 0.580 & 0.919 & 0.936 \\
& 4 & 0.043 & 0.041 & 0.614 & 0.603 & 0.920 & 0.930 \\
& 5 & 0.052 & 0.048 & 0.613 & 0.610 & 0.913 & 0.927 \\
& 6 & 0.057 & 0.057 & 0.622 & 0.619 & 0.912 & 0.925 \\
& 7 & 0.055 & 0.047 & 0.634 & 0.635 & 0.908 & 0.924 \\
& 8 & 0.055 & 0.053 & 0.600 & 0.590 & 0.920 & 0.941 \\
& 9 & 0.044 & 0.046 & 0.609 & 0.604 & 0.928 & 0.941 \\
& 10 & 0.043 & 0.045 & 0.616 & 0.610 & 0.911 & 0.918 \\
\hline
\multirow{10}{*}{$n=500$} & 1 & 0.058 & 0.058 & 0.897 & 0.887 & 0.997 & 0.999 \\
& 2 & 0.060 & 0.064 & 0.889 & 0.876 & 0.997 & 1.000 \\
& 3 & 0.049 & 0.050 & 0.884 & 0.878 & 1.000 & 1.000 \\
& 4 & 0.066 & 0.054 & 0.902 & 0.895 & 0.997 & 0.998 \\
& 5 & 0.042 & 0.038 & 0.896 & 0.884 & 0.999 & 1.000 \\
& 6 & 0.056 & 0.060 & 0.898 & 0.890 & 0.998 & 0.999 \\
& 7 & 0.048 & 0.045 & 0.909 & 0.899 & 0.998 & 0.998 \\
& 8 & 0.055 & 0.053 & 0.898 & 0.888 & 0.998 & 0.999 \\
& 9 & 0.044 & 0.042 & 0.897 & 0.886 & 0.999 & 0.999 \\
& 10 & 0.050 & 0.046 & 0.909 & 0.894 & 0.998 & 0.999 \\
\hline
\multirow{10}{*}{$n=1000$} & 1 & 0.043 & 0.040 & 0.997 & 0.996 & 1.000 & 1.000 \\
& 2 & 0.050 & 0.047 & 0.994 & 0.993 & 1.000 & 1.000 \\
& 3 & 0.047 & 0.043 & 0.993 & 0.991 & 1.000 & 1.000 \\
& 4 & 0.043 & 0.045 & 0.997 & 0.996 & 1.000 & 1.000 \\
& 5 & 0.048 & 0.052 & 0.996 & 0.998 & 1.000 & 1.000 \\
& 6 & 0.046 & 0.046 & 0.998 & 0.998 & 1.000 & 1.000 \\
& 7 & 0.041 & 0.044 & 0.993 & 0.992 & 1.000 & 1.000 \\
& 8 & 0.044 & 0.044 & 0.997 & 0.995 & 1.000 & 1.000 \\
& 9 & 0.052 & 0.058 & 0.996 & 0.994 & 1.000 & 1.000 \\
& 10 & 0.052 & 0.048 & 0.993 & 0.989 & 1.000 & 1.000 \\
\hline
\end{tabular}
\end{table}

\begin{table}[htbp]
\centering
\caption{Results for $S_n(\xi,\theta_n,\sigma^2_n)$, model $1$, ordinary smooth case, $\alpha=0.1$}
\begin{tabular}{cccccccc}
\hline
\multirow{2}{*}{$n$} & \multirow{2}{*}{$10c$} & \multicolumn{2}{c}{DGP(0)} & \multicolumn{2}{c}{DGP(1)} & \multicolumn{2}{c}{DGP(2)} \\
\cline{3-8}
& & KS & CvM & KS & CvM & KS & CvM \\
\hline
\multirow{10}{*}{$n=250$} & 1 & 0.105 & 0.102 & 0.742 & 0.741 & 0.961 & 0.974 \\
& 2 & 0.102 & 0.105 & 0.746 & 0.738 & 0.969 & 0.974 \\
& 3 & 0.110 & 0.105 & 0.729 & 0.725 & 0.954 & 0.964 \\
& 4 & 0.115 & 0.113 & 0.735 & 0.733 & 0.965 & 0.971 \\
& 5 & 0.105 & 0.105 & 0.730 & 0.719 & 0.961 & 0.966 \\
& 6 & 0.113 & 0.106 & 0.746 & 0.738 & 0.961 & 0.965 \\
& 7 & 0.093 & 0.094 & 0.739 & 0.736 & 0.966 & 0.971 \\
& 8 & 0.096 & 0.095 & 0.760 & 0.747 & 0.963 & 0.966 \\
& 9 & 0.098 & 0.102 & 0.731 & 0.715 & 0.967 & 0.974 \\
& 10 & 0.110 & 0.103 & 0.716 & 0.712 & 0.961 & 0.980 \\
\hline
\multirow{10}{*}{$n=500$} & 1 & 0.120 & 0.112 & 0.936 & 0.929 & 1.000 & 1.000 \\
& 2 & 0.107 & 0.106 & 0.944 & 0.937 & 0.999 & 1.000 \\
& 3 & 0.129 & 0.127 & 0.941 & 0.930 & 0.999 & 1.000 \\
& 4 & 0.098 & 0.102 & 0.933 & 0.930 & 0.999 & 1.000 \\
& 5 & 0.088 & 0.087 & 0.944 & 0.938 & 1.000 & 1.000 \\
& 6 & 0.098 & 0.099 & 0.935 & 0.926 & 0.999 & 1.000 \\
& 7 & 0.111 & 0.107 & 0.950 & 0.943 & 1.000 & 1.000 \\
& 8 & 0.104 & 0.099 & 0.941 & 0.935 & 1.000 & 1.000 \\
& 9 & 0.128 & 0.124 & 0.938 & 0.926 & 1.000 & 1.000 \\
& 10 & 0.103 & 0.105 & 0.945 & 0.939 & 1.000 & 1.000 \\
\hline
\multirow{10}{*}{$n=1000$} & 1 & 0.101 & 0.098 & 0.999 & 0.997 & 1.000 & 1.000 \\
& 2 & 0.108 & 0.111 & 0.997 & 0.997 & 1.000 & 1.000 \\
& 3 & 0.105 & 0.105 & 0.997 & 0.995 & 1.000 & 1.000 \\
& 4 & 0.091 & 0.099 & 0.999 & 0.998 & 1.000 & 1.000 \\
& 5 & 0.100 & 0.097 & 0.997 & 0.995 & 1.000 & 1.000 \\
& 6 & 0.109 & 0.109 & 0.999 & 0.997 & 1.000 & 1.000 \\
& 7 & 0.119 & 0.111 & 0.996 & 0.996 & 1.000 & 1.000 \\
& 8 & 0.102 & 0.096 & 0.995 & 0.994 & 1.000 & 1.000 \\
& 9 & 0.106 & 0.101 & 0.998 & 0.998 & 1.000 & 1.000 \\
& 10 & 0.115 & 0.118 & 0.999 & 1.000 & 1.000 & 1.000 \\
\hline
\end{tabular}
\end{table}

\begin{table}[htbp]
\centering
\caption{Results for $S_n(\xi,\theta_n,\sigma^2_n)$, model $1$, supersmooth case, $\alpha=0.01$}
\begin{tabular}{cccccccc}
\hline
\multirow{2}{*}{$n$} & \multirow{2}{*}{$10c$} & \multicolumn{2}{c}{DGP(0)} & \multicolumn{2}{c}{DGP(1)} & \multicolumn{2}{c}{DGP(2)} \\
\cline{3-8}
& & KS & CvM & KS & CvM & KS & CvM \\
\hline
\multirow{10}{*}{$n=250$} & 1 & 0.010 & 0.010 & 0.389 & 0.398 & 0.778 & 0.832 \\
& 2 & 0.010 & 0.009 & 0.390 & 0.388 & 0.754 & 0.788 \\
& 3 & 0.007 & 0.008 & 0.398 & 0.393 & 0.736 & 0.778 \\
& 4 & 0.011 & 0.010 & 0.397 & 0.397 & 0.737 & 0.778 \\
& 5 & 0.006 & 0.005 & 0.372 & 0.373 & 0.750 & 0.779 \\
& 6 & 0.015 & 0.016 & 0.400 & 0.401 & 0.752 & 0.790 \\
& 7 & 0.010 & 0.011 & 0.431 & 0.431 & 0.751 & 0.803 \\
& 8 & 0.010 & 0.008 & 0.404 & 0.403 & 0.738 & 0.797 \\
& 9 & 0.010 & 0.009 & 0.379 & 0.383 & 0.732 & 0.777 \\
& 10 & 0.014 & 0.013 & 0.399 & 0.392 & 0.733 & 0.779 \\
\hline
\multirow{10}{*}{$n=500$} & 1 & 0.014 & 0.014 & 0.784 & 0.777 & 0.990 & 0.991 \\
& 2 & 0.018 & 0.018 & 0.790 & 0.794 & 0.995 & 0.994 \\
& 3 & 0.015 & 0.016 & 0.792 & 0.794 & 0.991 & 0.995 \\
& 4 & 0.011 & 0.010 & 0.773 & 0.779 & 0.990 & 0.994 \\
& 5 & 0.014 & 0.015 & 0.788 & 0.794 & 0.985 & 0.991 \\
& 6 & 0.011 & 0.011 & 0.808 & 0.810 & 0.990 & 0.992 \\
& 7 & 0.012 & 0.012 & 0.791 & 0.795 & 0.989 & 0.990 \\
& 8 & 0.013 & 0.012 & 0.798 & 0.800 & 0.987 & 0.993 \\
& 9 & 0.019 & 0.017 & 0.785 & 0.783 & 0.994 & 0.995 \\
& 10 & 0.012 & 0.012 & 0.807 & 0.805 & 0.997 & 0.996 \\
\hline
\multirow{10}{*}{$n=1000$} & 1 & 0.015 & 0.016 & 0.985 & 0.984 & 1.000 & 1.000 \\
& 2 & 0.017 & 0.019 & 0.984 & 0.985 & 1.000 & 1.000 \\
& 3 & 0.013 & 0.013 & 0.991 & 0.993 & 1.000 & 1.000 \\
& 4 & 0.027 & 0.026 & 0.987 & 0.989 & 1.000 & 1.000 \\
& 5 & 0.011 & 0.012 & 0.983 & 0.982 & 1.000 & 1.000 \\
& 6 & 0.020 & 0.019 & 0.991 & 0.989 & 1.000 & 1.000 \\
& 7 & 0.014 & 0.014 & 0.990 & 0.991 & 1.000 & 1.000 \\
& 8 & 0.009 & 0.010 & 0.986 & 0.986 & 1.000 & 1.000 \\
& 9 & 0.013 & 0.014 & 0.984 & 0.984 & 1.000 & 1.000 \\
& 10 & 0.014 & 0.015 & 0.992 & 0.993 & 1.000 & 1.000 \\
\hline
\end{tabular}
\end{table}

\begin{table}[htbp]
\centering
\caption{Results for $S_n(\xi,\theta_n,\sigma^2_n)$, model $1$, supersmooth case, $\alpha=0.05$}
\begin{tabular}{cccccccc}
\hline
\multirow{2}{*}{$n$} & \multirow{2}{*}{$10c$} & \multicolumn{2}{c}{DGP(0)} & \multicolumn{2}{c}{DGP(1)} & \multicolumn{2}{c}{DGP(2)} \\
\cline{3-8}
& & KS & CvM & KS & CvM & KS & CvM \\
\hline
\multirow{10}{*}{$n=250$} & 1 & 0.059 & 0.057 & 0.660 & 0.659 & 0.922 & 0.937 \\
& 2 & 0.053 & 0.052 & 0.645 & 0.649 & 0.920 & 0.941 \\
& 3 & 0.038 & 0.038 & 0.659 & 0.662 & 0.926 & 0.948 \\
& 4 & 0.051 & 0.054 & 0.642 & 0.647 & 0.913 & 0.931 \\
& 5 & 0.046 & 0.046 & 0.621 & 0.632 & 0.904 & 0.929 \\
& 6 & 0.058 & 0.055 & 0.653 & 0.658 & 0.928 & 0.950 \\
& 7 & 0.054 & 0.052 & 0.660 & 0.656 & 0.900 & 0.924 \\
& 8 & 0.047 & 0.049 & 0.655 & 0.664 & 0.914 & 0.939 \\
& 9 & 0.048 & 0.046 & 0.659 & 0.654 & 0.914 & 0.938 \\
& 10 & 0.057 & 0.056 & 0.656 & 0.662 & 0.915 & 0.939 \\
\hline
\multirow{10}{*}{$n=500$} & 1 & 0.050 & 0.048 & 0.932 & 0.929 & 0.999 & 1.000 \\
& 2 & 0.058 & 0.057 & 0.912 & 0.914 & 0.997 & 0.997 \\
& 3 & 0.054 & 0.052 & 0.931 & 0.927 & 1.000 & 1.000 \\
& 4 & 0.054 & 0.051 & 0.935 & 0.932 & 0.998 & 0.999 \\
& 5 & 0.051 & 0.052 & 0.930 & 0.928 & 1.000 & 1.000 \\
& 6 & 0.054 & 0.055 & 0.926 & 0.924 & 0.998 & 0.999 \\
& 7 & 0.052 & 0.051 & 0.898 & 0.900 & 0.997 & 0.998 \\
& 8 & 0.046 & 0.053 & 0.922 & 0.924 & 0.998 & 1.000 \\
& 9 & 0.073 & 0.074 & 0.907 & 0.910 & 0.997 & 1.000 \\
& 10 & 0.054 & 0.052 & 0.917 & 0.917 & 0.999 & 1.000 \\
\hline
\multirow{10}{*}{$n=1000$} & 1 & 0.049 & 0.046 & 0.995 & 0.997 & 1.000 & 1.000 \\
& 2 & 0.042 & 0.044 & 0.996 & 0.996 & 1.000 & 1.000 \\
& 3 & 0.056 & 0.057 & 0.997 & 0.997 & 1.000 & 1.000 \\
& 4 & 0.057 & 0.059 & 0.996 & 0.996 & 1.000 & 1.000 \\
& 5 & 0.049 & 0.050 & 1.000 & 1.000 & 1.000 & 1.000 \\
& 6 & 0.054 & 0.053 & 1.000 & 0.999 & 1.000 & 1.000 \\
& 7 & 0.050 & 0.050 & 0.995 & 0.995 & 1.000 & 1.000 \\
& 8 & 0.059 & 0.058 & 0.998 & 0.998 & 1.000 & 1.000 \\
& 9 & 0.049 & 0.050 & 0.998 & 0.998 & 1.000 & 1.000 \\
& 10 & 0.060 & 0.058 & 0.998 & 0.999 & 1.000 & 1.000 \\
\hline
\end{tabular}
\end{table}

\begin{table}[htbp]
\centering
\caption{Results for $S_n(\xi,\theta_n,\sigma^2_n)$, model $1$, supersmooth case, $\alpha=0.1$}
\begin{tabular}{cccccccc}
\hline
\multirow{2}{*}{$n$} & \multirow{2}{*}{$10c$} & \multicolumn{2}{c}{DGP(0)} & \multicolumn{2}{c}{DGP(1)} & \multicolumn{2}{c}{DGP(2)} \\
\cline{3-8}
& & KS & CvM & KS & CvM & KS & CvM \\
\hline
\multirow{10}{*}{$n=250$} & 1 & 0.097 & 0.098 & 0.780 & 0.780 & 0.965 & 0.969 \\
& 2 & 0.116 & 0.117 & 0.746 & 0.752 & 0.963 & 0.974 \\
& 3 & 0.104 & 0.106 & 0.771 & 0.770 & 0.958 & 0.968 \\
& 4 & 0.090 & 0.094 & 0.774 & 0.777 & 0.953 & 0.963 \\
& 5 & 0.110 & 0.108 & 0.774 & 0.773 & 0.951 & 0.963 \\
& 6 & 0.096 & 0.096 & 0.791 & 0.788 & 0.957 & 0.971 \\
& 7 & 0.102 & 0.103 & 0.770 & 0.766 & 0.957 & 0.973 \\
& 8 & 0.101 & 0.100 & 0.773 & 0.772 & 0.953 & 0.971 \\
& 9 & 0.105 & 0.101 & 0.787 & 0.781 & 0.952 & 0.964 \\
& 10 & 0.089 & 0.092 & 0.776 & 0.772 & 0.954 & 0.966 \\
\hline
\multirow{10}{*}{$n=500$} & 1 & 0.084 & 0.085 & 0.969 & 0.968 & 1.000 & 1.000 \\
& 2 & 0.085 & 0.087 & 0.958 & 0.956 & 1.000 & 1.000 \\
& 3 & 0.100 & 0.101 & 0.967 & 0.966 & 0.999 & 1.000 \\
& 4 & 0.083 & 0.087 & 0.965 & 0.970 & 0.999 & 1.000 \\
& 5 & 0.113 & 0.113 & 0.956 & 0.957 & 0.999 & 0.999 \\
& 6 & 0.106 & 0.103 & 0.963 & 0.959 & 1.000 & 1.000 \\
& 7 & 0.111 & 0.115 & 0.962 & 0.961 & 1.000 & 1.000 \\
& 8 & 0.104 & 0.105 & 0.956 & 0.958 & 1.000 & 1.000 \\
& 9 & 0.091 & 0.088 & 0.978 & 0.979 & 1.000 & 1.000 \\
& 10 & 0.116 & 0.115 & 0.959 & 0.958 & 0.999 & 0.999 \\
\hline
\multirow{10}{*}{$n=1000$} & 1 & 0.119 & 0.122 & 1.000 & 1.000 & 1.000 & 1.000 \\
& 2 & 0.091 & 0.094 & 0.999 & 0.999 & 1.000 & 1.000 \\
& 3 & 0.102 & 0.103 & 1.000 & 1.000 & 1.000 & 1.000 \\
& 4 & 0.099 & 0.098 & 1.000 & 1.000 & 1.000 & 1.000 \\
& 5 & 0.111 & 0.108 & 1.000 & 1.000 & 1.000 & 1.000 \\
& 6 & 0.102 & 0.100 & 1.000 & 1.000 & 1.000 & 1.000 \\
& 7 & 0.103 & 0.104 & 0.999 & 0.999 & 1.000 & 1.000 \\
& 8 & 0.094 & 0.092 & 0.999 & 0.999 & 1.000 & 1.000 \\
& 9 & 0.110 & 0.104 & 1.000 & 1.000 & 1.000 & 1.000 \\
& 10 & 0.108 & 0.109 & 1.000 & 1.000 & 1.000 & 1.000 \\
\hline
\end{tabular}
\end{table}

\begin{table}[htbp]
\centering
\caption{Results for $\hat{S}_n(\xi,\hat{\theta}_n,\hat{\sigma}^2_n)$, model $1$, ordinary smooth case, $\alpha=0.01$}
\begin{tabular}{cccccccc}
\hline
\multirow{2}{*}{$n$} & \multirow{2}{*}{$10c$} & \multicolumn{2}{c}{DGP(0)} & \multicolumn{2}{c}{DGP(1)} & \multicolumn{2}{c}{DGP(2)} \\
\cline{3-8}
& & KS & CvM & KS & CvM & KS & CvM \\
\hline
\multirow{10}{*}{$n=250$} & 1 & 0.010 & 0.008 & 0.357 & 0.349 & 0.752 & 0.787 \\
& 2 & 0.016 & 0.014 & 0.358 & 0.352 & 0.787 & 0.818 \\
& 3 & 0.007 & 0.009 & 0.401 & 0.393 & 0.748 & 0.771 \\
& 4 & 0.011 & 0.011 & 0.350 & 0.344 & 0.747 & 0.780 \\
& 5 & 0.022 & 0.022 & 0.390 & 0.382 & 0.734 & 0.761 \\
& 6 & 0.018 & 0.014 & 0.363 & 0.356 & 0.776 & 0.791 \\
& 7 & 0.015 & 0.013 & 0.365 & 0.366 & 0.720 & 0.770 \\
& 8 & 0.020 & 0.018 & 0.379 & 0.382 & 0.752 & 0.787 \\
& 9 & 0.011 & 0.011 & 0.366 & 0.366 & 0.801 & 0.822 \\
& 10 & 0.012 & 0.011 & 0.376 & 0.376 & 0.757 & 0.774 \\
\hline
\multirow{10}{*}{$n=500$} & 1 & 0.009 & 0.008 & 0.758 & 0.746 & 0.989 & 0.992 \\
& 2 & 0.014 & 0.013 & 0.763 & 0.748 & 0.983 & 0.991 \\
& 3 & 0.010 & 0.009 & 0.750 & 0.734 & 0.983 & 0.992 \\
& 4 & 0.013 & 0.009 & 0.743 & 0.729 & 0.988 & 0.993 \\
& 5 & 0.014 & 0.014 & 0.744 & 0.731 & 0.987 & 0.991 \\
& 6 & 0.016 & 0.015 & 0.770 & 0.751 & 0.988 & 0.987 \\
& 7 & 0.015 & 0.013 & 0.749 & 0.737 & 0.991 & 0.994 \\
& 8 & 0.008 & 0.008 & 0.758 & 0.740 & 0.988 & 0.992 \\
& 9 & 0.012 & 0.010 & 0.751 & 0.741 & 0.991 & 0.995 \\
& 10 & 0.019 & 0.018 & 0.728 & 0.716 & 0.988 & 0.987 \\
\hline
\multirow{10}{*}{$n=1000$} & 1 & 0.012 & 0.015 & 0.986 & 0.975 & 1.000 & 1.000 \\
& 2 & 0.019 & 0.020 & 0.975 & 0.968 & 1.000 & 1.000 \\
& 3 & 0.016 & 0.014 & 0.979 & 0.972 & 1.000 & 1.000 \\
& 4 & 0.017 & 0.015 & 0.971 & 0.965 & 1.000 & 1.000 \\
& 5 & 0.021 & 0.017 & 0.978 & 0.973 & 1.000 & 1.000 \\
& 6 & 0.012 & 0.011 & 0.967 & 0.963 & 1.000 & 1.000 \\
& 7 & 0.010 & 0.008 & 0.974 & 0.968 & 1.000 & 1.000 \\
& 8 & 0.014 & 0.016 & 0.981 & 0.976 & 1.000 & 1.000 \\
& 9 & 0.010 & 0.007 & 0.975 & 0.965 & 1.000 & 1.000 \\
& 10 & 0.008 & 0.009 & 0.980 & 0.973 & 1.000 & 1.000 \\
\hline
\end{tabular}
\end{table}

\begin{table}[htbp]
\centering
\caption{Results for $\hat{S}_n(\xi,\hat{\theta}_n,\hat{\sigma}^2_n)$, model $1$, ordinary smooth case, $\alpha=0.05$}
\begin{tabular}{cccccccc}
\hline
\multirow{2}{*}{$n$} & \multirow{2}{*}{$10c$} & \multicolumn{2}{c}{DGP(0)} & \multicolumn{2}{c}{DGP(1)} & \multicolumn{2}{c}{DGP(2)} \\
\cline{3-8}
& & KS & CvM & KS & CvM & KS & CvM \\
\hline
\multirow{10}{*}{$n=250$} & 1 & 0.060 & 0.051 & 0.594 & 0.580 & 0.914 & 0.931 \\
& 2 & 0.063 & 0.062 & 0.608 & 0.592 & 0.925 & 0.937 \\
& 3 & 0.059 & 0.050 & 0.645 & 0.648 & 0.910 & 0.929 \\
& 4 & 0.067 & 0.062 & 0.614 & 0.606 & 0.913 & 0.930 \\
& 5 & 0.059 & 0.054 & 0.605 & 0.603 & 0.909 & 0.931 \\
& 6 & 0.060 & 0.057 & 0.615 & 0.608 & 0.926 & 0.936 \\
& 7 & 0.061 & 0.059 & 0.640 & 0.631 & 0.902 & 0.931 \\
& 8 & 0.057 & 0.049 & 0.606 & 0.598 & 0.916 & 0.931 \\
& 9 & 0.048 & 0.046 & 0.623 & 0.613 & 0.913 & 0.936 \\
& 10 & 0.045 & 0.039 & 0.625 & 0.617 & 0.907 & 0.934 \\
\hline
\multirow{10}{*}{$n=500$} & 1 & 0.066 & 0.060 & 0.894 & 0.886 & 1.000 & 1.000 \\
& 2 & 0.060 & 0.053 & 0.903 & 0.901 & 0.997 & 0.997 \\
& 3 & 0.067 & 0.064 & 0.876 & 0.863 & 0.999 & 1.000 \\
& 4 & 0.070 & 0.067 & 0.911 & 0.905 & 0.999 & 1.000 \\
& 5 & 0.062 & 0.059 & 0.901 & 0.892 & 0.996 & 0.998 \\
& 6 & 0.053 & 0.048 & 0.901 & 0.898 & 0.997 & 0.999 \\
& 7 & 0.050 & 0.045 & 0.892 & 0.884 & 0.997 & 0.999 \\
& 8 & 0.050 & 0.049 & 0.895 & 0.880 & 0.999 & 0.999 \\
& 9 & 0.041 & 0.042 & 0.892 & 0.883 & 0.999 & 0.999 \\
& 10 & 0.055 & 0.057 & 0.917 & 0.906 & 0.997 & 0.999 \\
\hline
\multirow{10}{*}{$n=1000$} & 1 & 0.045 & 0.045 & 0.997 & 0.995 & 1.000 & 1.000 \\
& 2 & 0.053 & 0.048 & 0.995 & 0.994 & 1.000 & 1.000 \\
& 3 & 0.059 & 0.058 & 0.994 & 0.992 & 1.000 & 1.000 \\
& 4 & 0.051 & 0.050 & 0.995 & 0.993 & 1.000 & 1.000 \\
& 5 & 0.054 & 0.059 & 0.996 & 0.996 & 1.000 & 1.000 \\
& 6 & 0.056 & 0.054 & 0.993 & 0.991 & 1.000 & 1.000 \\
& 7 & 0.059 & 0.053 & 0.996 & 0.994 & 1.000 & 1.000 \\
& 8 & 0.063 & 0.062 & 0.994 & 0.994 & 1.000 & 1.000 \\
& 9 & 0.051 & 0.049 & 0.993 & 0.989 & 1.000 & 1.000 \\
& 10 & 0.046 & 0.050 & 0.995 & 0.992 & 1.000 & 1.000 \\
\hline
\end{tabular}
\end{table}

\begin{table}[htbp]
\centering
\caption{Results for $\hat{S}_n(\xi,\hat{\theta}_n,\hat{\sigma}^2_n)$, model $1$, ordinary smooth case, $\alpha=0.1$}
\begin{tabular}{cccccccc}
\hline
\multirow{2}{*}{$n$} & \multirow{2}{*}{$10c$} & \multicolumn{2}{c}{DGP(0)} & \multicolumn{2}{c}{DGP(1)} & \multicolumn{2}{c}{DGP(2)} \\
\cline{3-8}
& & KS & CvM & KS & CvM & KS & CvM \\
\hline
\multirow{10}{*}{$n=250$} & 1 & 0.111 & 0.107 & 0.724 & 0.718 & 0.964 & 0.975 \\
& 2 & 0.104 & 0.094 & 0.729 & 0.733 & 0.965 & 0.970 \\
& 3 & 0.112 & 0.110 & 0.716 & 0.709 & 0.960 & 0.966 \\
& 4 & 0.110 & 0.107 & 0.740 & 0.724 & 0.969 & 0.974 \\
& 5 & 0.098 & 0.097 & 0.735 & 0.730 & 0.955 & 0.973 \\
& 6 & 0.120 & 0.119 & 0.743 & 0.745 & 0.969 & 0.978 \\
& 7 & 0.102 & 0.101 & 0.740 & 0.733 & 0.957 & 0.961 \\
& 8 & 0.124 & 0.125 & 0.728 & 0.717 & 0.961 & 0.956 \\
& 9 & 0.124 & 0.120 & 0.739 & 0.736 & 0.965 & 0.975 \\
& 10 & 0.126 & 0.123 & 0.752 & 0.744 & 0.962 & 0.977 \\
\hline
\multirow{10}{*}{$n=500$} & 1 & 0.102 & 0.108 & 0.949 & 0.944 & 0.999 & 1.000 \\
& 2 & 0.127 & 0.125 & 0.942 & 0.930 & 0.999 & 1.000 \\
& 3 & 0.108 & 0.098 & 0.937 & 0.935 & 0.999 & 1.000 \\
& 4 & 0.107 & 0.100 & 0.933 & 0.929 & 0.999 & 0.999 \\
& 5 & 0.102 & 0.099 & 0.953 & 0.947 & 1.000 & 1.000 \\
& 6 & 0.111 & 0.116 & 0.936 & 0.928 & 1.000 & 1.000 \\
& 7 & 0.125 & 0.125 & 0.938 & 0.927 & 1.000 & 1.000 \\
& 8 & 0.094 & 0.094 & 0.949 & 0.944 & 0.998 & 0.999 \\
& 9 & 0.106 & 0.102 & 0.948 & 0.938 & 0.999 & 0.999 \\
& 10 & 0.105 & 0.105 & 0.938 & 0.927 & 0.999 & 1.000 \\
\hline
\multirow{10}{*}{$n=1000$} & 1 & 0.093 & 0.092 & 0.997 & 0.997 & 1.000 & 1.000 \\
& 2 & 0.131 & 0.130 & 0.997 & 0.995 & 1.000 & 1.000 \\
& 3 & 0.102 & 0.103 & 0.998 & 0.996 & 1.000 & 1.000 \\
& 4 & 0.116 & 0.116 & 1.000 & 0.998 & 1.000 & 1.000 \\
& 5 & 0.103 & 0.100 & 0.998 & 0.998 & 1.000 & 1.000 \\
& 6 & 0.114 & 0.109 & 0.997 & 0.995 & 1.000 & 1.000 \\
& 7 & 0.102 & 0.097 & 0.999 & 0.998 & 1.000 & 1.000 \\
& 8 & 0.120 & 0.112 & 1.000 & 1.000 & 1.000 & 1.000 \\
& 9 & 0.108 & 0.102 & 0.998 & 0.998 & 1.000 & 1.000 \\
& 10 & 0.114 & 0.099 & 1.000 & 1.000 & 1.000 & 1.000 \\
\hline
\end{tabular}
\end{table}

\begin{table}[htbp]
\centering
\caption{Results for $\hat{S}_n(\xi,\hat{\theta}_n,\hat{\sigma}^2_n)$, model $1$, supersmooth case, $\alpha=0.01$}
\begin{tabular}{cccccccc}
\hline
\multirow{2}{*}{$n$} & \multirow{2}{*}{$10c$} & \multicolumn{2}{c}{DGP(0)} & \multicolumn{2}{c}{DGP(1)} & \multicolumn{2}{c}{DGP(2)} \\
\cline{3-8}
& & KS & CvM & KS & CvM & KS & CvM \\
\hline
\multirow{10}{*}{$n=250$} & 1 & 0.006 & 0.010 & 0.399 & 0.404 & 0.749 & 0.800 \\
& 2 & 0.011 & 0.011 & 0.412 & 0.412 & 0.731 & 0.795 \\
& 3 & 0.011 & 0.009 & 0.406 & 0.408 & 0.731 & 0.775 \\
& 4 & 0.013 & 0.012 & 0.376 & 0.377 & 0.763 & 0.799 \\
& 5 & 0.008 & 0.008 & 0.413 & 0.417 & 0.720 & 0.760 \\
& 6 & 0.015 & 0.014 & 0.398 & 0.396 & 0.709 & 0.750 \\
& 7 & 0.017 & 0.019 & 0.404 & 0.402 & 0.731 & 0.802 \\
& 8 & 0.012 & 0.013 & 0.418 & 0.417 & 0.740 & 0.782 \\
& 9 & 0.011 & 0.012 & 0.398 & 0.396 & 0.759 & 0.807 \\
& 10 & 0.010 & 0.009 & 0.407 & 0.404 & 0.742 & 0.773 \\
\hline
\multirow{10}{*}{$n=500$} & 1 & 0.013 & 0.013 & 0.797 & 0.798 & 0.989 & 0.994 \\
& 2 & 0.011 & 0.009 & 0.801 & 0.800 & 0.988 & 0.989 \\
& 3 & 0.015 & 0.013 & 0.774 & 0.776 & 0.982 & 0.992 \\
& 4 & 0.012 & 0.012 & 0.791 & 0.789 & 0.985 & 0.989 \\
& 5 & 0.013 & 0.015 & 0.781 & 0.782 & 0.989 & 0.994 \\
& 6 & 0.010 & 0.010 & 0.773 & 0.775 & 0.990 & 0.993 \\
& 7 & 0.014 & 0.014 & 0.798 & 0.802 & 0.988 & 0.993 \\
& 8 & 0.006 & 0.005 & 0.795 & 0.794 & 0.989 & 0.992 \\
& 9 & 0.014 & 0.015 & 0.785 & 0.787 & 0.987 & 0.996 \\
& 10 & 0.015 & 0.015 & 0.782 & 0.782 & 0.991 & 0.997 \\
\hline
\multirow{10}{*}{$n=1000$} & 1 & 0.017 & 0.015 & 0.990 & 0.990 & 1.000 & 1.000 \\
& 2 & 0.014 & 0.014 & 0.994 & 0.991 & 1.000 & 1.000 \\
& 3 & 0.014 & 0.016 & 0.984 & 0.982 & 1.000 & 1.000 \\
& 4 & 0.011 & 0.011 & 0.988 & 0.987 & 1.000 & 1.000 \\
& 5 & 0.019 & 0.019 & 0.993 & 0.991 & 1.000 & 1.000 \\
& 6 & 0.012 & 0.012 & 0.990 & 0.987 & 1.000 & 1.000 \\
& 7 & 0.015 & 0.014 & 0.979 & 0.979 & 1.000 & 1.000 \\
& 8 & 0.010 & 0.010 & 0.986 & 0.987 & 1.000 & 1.000 \\
& 9 & 0.015 & 0.014 & 0.984 & 0.986 & 1.000 & 1.000 \\
& 10 & 0.015 & 0.016 & 0.989 & 0.988 & 1.000 & 1.000 \\
\hline
\end{tabular}
\end{table}

\begin{table}[htbp]
\centering
\caption{Results for $\hat{S}_n(\xi,\hat{\theta}_n,\hat{\sigma}^2_n)$, model $1$, supersmooth case, $\alpha=0.05$}
\begin{tabular}{cccccccc}
\hline
\multirow{2}{*}{$n$} & \multirow{2}{*}{$10c$} & \multicolumn{2}{c}{DGP(0)} & \multicolumn{2}{c}{DGP(1)} & \multicolumn{2}{c}{DGP(2)} \\
\cline{3-8}
& & KS & CvM & KS & CvM & KS & CvM \\
\hline
\multirow{10}{*}{$n=250$} & 1 & 0.041 & 0.038 & 0.621 & 0.620 & 0.923 & 0.945 \\
& 2 & 0.063 & 0.065 & 0.654 & 0.657 & 0.929 & 0.954 \\
& 3 & 0.055 & 0.052 & 0.665 & 0.661 & 0.912 & 0.932 \\
& 4 & 0.053 & 0.055 & 0.669 & 0.667 & 0.919 & 0.949 \\
& 5 & 0.070 & 0.068 & 0.662 & 0.667 & 0.912 & 0.939 \\
& 6 & 0.049 & 0.053 & 0.640 & 0.645 & 0.918 & 0.944 \\
& 7 & 0.052 & 0.053 & 0.654 & 0.654 & 0.920 & 0.940 \\
& 8 & 0.055 & 0.058 & 0.637 & 0.636 & 0.924 & 0.942 \\
& 9 & 0.038 & 0.038 & 0.643 & 0.635 & 0.909 & 0.939 \\
& 10 & 0.060 & 0.058 & 0.672 & 0.676 & 0.909 & 0.936 \\
\hline
\multirow{10}{*}{$n=500$} & 1 & 0.057 & 0.059 & 0.938 & 0.942 & 0.999 & 1.000 \\
& 2 & 0.061 & 0.061 & 0.912 & 0.913 & 1.000 & 1.000 \\
& 3 & 0.065 & 0.063 & 0.946 & 0.944 & 1.000 & 0.999 \\
& 4 & 0.053 & 0.055 & 0.923 & 0.919 & 0.999 & 1.000 \\
& 5 & 0.043 & 0.044 & 0.923 & 0.926 & 0.999 & 0.999 \\
& 6 & 0.056 & 0.057 & 0.928 & 0.927 & 0.997 & 0.998 \\
& 7 & 0.054 & 0.053 & 0.934 & 0.933 & 1.000 & 1.000 \\
& 8 & 0.065 & 0.064 & 0.920 & 0.917 & 0.999 & 0.999 \\
& 9 & 0.064 & 0.064 & 0.933 & 0.934 & 1.000 & 1.000 \\
& 10 & 0.055 & 0.055 & 0.929 & 0.927 & 0.997 & 1.000 \\
\hline
\multirow{10}{*}{$n=1000$} & 1 & 0.051 & 0.052 & 0.995 & 0.996 & 1.000 & 1.000 \\
& 2 & 0.061 & 0.059 & 1.000 & 1.000 & 1.000 & 1.000 \\
& 3 & 0.052 & 0.050 & 1.000 & 1.000 & 1.000 & 1.000 \\
& 4 & 0.048 & 0.051 & 0.999 & 0.999 & 1.000 & 1.000 \\
& 5 & 0.068 & 0.068 & 0.996 & 0.996 & 1.000 & 1.000 \\
& 6 & 0.062 & 0.062 & 0.998 & 0.998 & 1.000 & 1.000 \\
& 7 & 0.054 & 0.052 & 0.999 & 0.999 & 1.000 & 1.000 \\
& 8 & 0.055 & 0.055 & 0.999 & 0.999 & 1.000 & 1.000 \\
& 9 & 0.056 & 0.057 & 0.997 & 0.997 & 1.000 & 1.000 \\
& 10 & 0.068 & 0.069 & 0.998 & 0.997 & 1.000 & 1.000 \\
\hline
\end{tabular}
\end{table}

\begin{table}[htbp]
\centering
\caption{Results for $\hat{S}_n(\xi,\hat{\theta}_n,\hat{\sigma}^2_n)$, model $1$, supersmooth case, $\alpha=0.1$}
\begin{tabular}{cccccccc}
\hline
\multirow{2}{*}{$n$} & \multirow{2}{*}{$10c$} & \multicolumn{2}{c}{DGP(0)} & \multicolumn{2}{c}{DGP(1)} & \multicolumn{2}{c}{DGP(2)} \\
\cline{3-8}
& & KS & CvM & KS & CvM & KS & CvM \\
\hline
\multirow{10}{*}{$n=250$} & 1 & 0.111 & 0.109 & 0.772 & 0.771 & 0.965 & 0.976 \\
& 2 & 0.135 & 0.139 & 0.773 & 0.778 & 0.970 & 0.977 \\
& 3 & 0.102 & 0.104 & 0.746 & 0.748 & 0.958 & 0.977 \\
& 4 & 0.096 & 0.091 & 0.767 & 0.768 & 0.964 & 0.975 \\
& 5 & 0.103 & 0.104 & 0.775 & 0.781 & 0.961 & 0.966 \\
& 6 & 0.139 & 0.136 & 0.786 & 0.787 & 0.967 & 0.973 \\
& 7 & 0.099 & 0.100 & 0.794 & 0.790 & 0.954 & 0.968 \\
& 8 & 0.087 & 0.081 & 0.754 & 0.756 & 0.967 & 0.980 \\
& 9 & 0.118 & 0.118 & 0.766 & 0.765 & 0.967 & 0.975 \\
& 10 & 0.110 & 0.112 & 0.760 & 0.764 & 0.964 & 0.967 \\
\hline
\multirow{10}{*}{$n=500$} & 1 & 0.087 & 0.085 & 0.945 & 0.947 & 1.000 & 1.000 \\
& 2 & 0.112 & 0.100 & 0.973 & 0.972 & 1.000 & 1.000 \\
& 3 & 0.111 & 0.110 & 0.949 & 0.948 & 1.000 & 1.000 \\
& 4 & 0.080 & 0.083 & 0.964 & 0.968 & 0.999 & 1.000 \\
& 5 & 0.103 & 0.100 & 0.969 & 0.967 & 1.000 & 1.000 \\
& 6 & 0.113 & 0.113 & 0.958 & 0.957 & 0.999 & 1.000 \\
& 7 & 0.098 & 0.091 & 0.961 & 0.959 & 1.000 & 1.000 \\
& 8 & 0.113 & 0.117 & 0.949 & 0.950 & 1.000 & 1.000 \\
& 9 & 0.099 & 0.101 & 0.967 & 0.963 & 1.000 & 1.000 \\
& 10 & 0.110 & 0.111 & 0.958 & 0.958 & 0.998 & 0.999 \\
\hline
\multirow{10}{*}{$n=1000$} & 1 & 0.109 & 0.103 & 1.000 & 1.000 & 1.000 & 1.000 \\
& 2 & 0.096 & 0.099 & 1.000 & 1.000 & 1.000 & 1.000 \\
& 3 & 0.102 & 0.104 & 0.999 & 0.999 & 1.000 & 1.000 \\
& 4 & 0.102 & 0.105 & 0.999 & 0.999 & 1.000 & 1.000 \\
& 5 & 0.107 & 0.106 & 0.998 & 0.998 & 1.000 & 1.000 \\
& 6 & 0.108 & 0.111 & 0.998 & 0.998 & 1.000 & 1.000 \\
& 7 & 0.086 & 0.087 & 1.000 & 1.000 & 1.000 & 1.000 \\
& 8 & 0.091 & 0.096 & 0.999 & 0.999 & 1.000 & 1.000 \\
& 9 & 0.100 & 0.100 & 0.999 & 0.999 & 1.000 & 1.000 \\
& 10 & 0.096 & 0.092 & 0.999 & 0.999 & 1.000 & 1.000 \\
\hline
\end{tabular}
\end{table}

\begin{table}[htbp]
\centering
\caption{Results for $S_n(\xi,\theta_n,\sigma^2_n)$, model $2$, ordinary smooth case, $\alpha=0.01$}
\begin{tabular}{cccccccc}
\hline
\multirow{2}{*}{$n$} & \multirow{2}{*}{$10c$} & \multicolumn{2}{c}{DGP(0)} & \multicolumn{2}{c}{DGP(1)} & \multicolumn{2}{c}{DGP(2)} \\
\cline{3-8}
& & KS & CvM & KS & CvM & KS & CvM \\
\hline
\multirow{10}{*}{$n=250$} & 1 & 0.009 & 0.010 & 0.547 & 0.543 & 0.906 & 0.917 \\
& 2 & 0.010 & 0.011 & 0.566 & 0.575 & 0.903 & 0.910 \\
& 3 & 0.013 & 0.014 & 0.535 & 0.533 & 0.912 & 0.908 \\
& 4 & 0.018 & 0.021 & 0.550 & 0.555 & 0.870 & 0.885 \\
& 5 & 0.015 & 0.014 & 0.553 & 0.556 & 0.899 & 0.911 \\
& 6 & 0.011 & 0.012 & 0.571 & 0.570 & 0.928 & 0.930 \\
& 7 & 0.020 & 0.019 & 0.538 & 0.543 & 0.907 & 0.909 \\
& 8 & 0.017 & 0.016 & 0.557 & 0.553 & 0.915 & 0.910 \\
& 9 & 0.016 & 0.017 & 0.561 & 0.568 & 0.889 & 0.903 \\
& 10 & 0.016 & 0.014 & 0.553 & 0.548 & 0.915 & 0.925 \\
\hline
\multirow{10}{*}{$n=500$} & 1 & 0.011 & 0.012 & 0.921 & 0.924 & 1.000 & 1.000 \\
& 2 & 0.007 & 0.007 & 0.918 & 0.913 & 1.000 & 1.000 \\
& 3 & 0.019 & 0.021 & 0.918 & 0.920 & 0.999 & 0.999 \\
& 4 & 0.017 & 0.017 & 0.929 & 0.924 & 1.000 & 1.000 \\
& 5 & 0.008 & 0.008 & 0.920 & 0.921 & 0.999 & 0.999 \\
& 6 & 0.009 & 0.008 & 0.918 & 0.920 & 0.999 & 1.000 \\
& 7 & 0.009 & 0.010 & 0.922 & 0.919 & 0.999 & 1.000 \\
& 8 & 0.011 & 0.012 & 0.930 & 0.928 & 1.000 & 1.000 \\
& 9 & 0.010 & 0.010 & 0.921 & 0.917 & 1.000 & 0.999 \\
& 10 & 0.015 & 0.015 & 0.923 & 0.918 & 0.998 & 0.999 \\
\hline
\multirow{10}{*}{$n=1000$} & 1 & 0.015 & 0.016 & 0.999 & 0.999 & 1.000 & 1.000 \\
& 2 & 0.019 & 0.018 & 0.998 & 0.998 & 1.000 & 1.000 \\
& 3 & 0.012 & 0.011 & 1.000 & 1.000 & 1.000 & 1.000 \\
& 4 & 0.014 & 0.012 & 0.999 & 0.999 & 1.000 & 1.000 \\
& 5 & 0.017 & 0.016 & 1.000 & 0.998 & 1.000 & 1.000 \\
& 6 & 0.013 & 0.017 & 0.999 & 0.999 & 1.000 & 1.000 \\
& 7 & 0.011 & 0.012 & 0.999 & 0.999 & 1.000 & 1.000 \\
& 8 & 0.018 & 0.019 & 1.000 & 0.999 & 1.000 & 1.000 \\
& 9 & 0.020 & 0.020 & 1.000 & 1.000 & 1.000 & 1.000 \\
& 10 & 0.018 & 0.017 & 1.000 & 1.000 & 1.000 & 1.000 \\
\hline
\end{tabular}
\end{table}

\begin{table}[htbp]
\centering
\caption{Results for $S_n(\xi,\theta_n,\sigma^2_n)$, model $2$, ordinary smooth case, $\alpha=0.05$}
\begin{tabular}{cccccccc}
\hline
\multirow{2}{*}{$n$} & \multirow{2}{*}{$10c$} & \multicolumn{2}{c}{DGP(0)} & \multicolumn{2}{c}{DGP(1)} & \multicolumn{2}{c}{DGP(2)} \\
\cline{3-8}
& & KS & CvM & KS & CvM & KS & CvM \\
\hline
\multirow{10}{*}{$n=250$} & 1 & 0.044 & 0.044 & 0.789 & 0.795 & 0.977 & 0.984 \\
& 2 & 0.057 & 0.061 & 0.802 & 0.806 & 0.976 & 0.978 \\
& 3 & 0.052 & 0.054 & 0.776 & 0.789 & 0.984 & 0.985 \\
& 4 & 0.050 & 0.055 & 0.816 & 0.816 & 0.981 & 0.988 \\
& 5 & 0.058 & 0.055 & 0.796 & 0.801 & 0.983 & 0.985 \\
& 6 & 0.040 & 0.042 & 0.807 & 0.804 & 0.983 & 0.989 \\
& 7 & 0.054 & 0.055 & 0.804 & 0.799 & 0.983 & 0.982 \\
& 8 & 0.049 & 0.048 & 0.794 & 0.800 & 0.984 & 0.988 \\
& 9 & 0.062 & 0.059 & 0.812 & 0.804 & 0.981 & 0.985 \\
& 10 & 0.040 & 0.041 & 0.802 & 0.800 & 0.985 & 0.987 \\
\hline
\multirow{10}{*}{$n=500$} & 1 & 0.052 & 0.052 & 0.977 & 0.979 & 1.000 & 1.000 \\
& 2 & 0.046 & 0.044 & 0.978 & 0.976 & 1.000 & 1.000 \\
& 3 & 0.042 & 0.042 & 0.984 & 0.983 & 1.000 & 1.000 \\
& 4 & 0.046 & 0.047 & 0.972 & 0.973 & 1.000 & 1.000 \\
& 5 & 0.054 & 0.056 & 0.981 & 0.981 & 1.000 & 1.000 \\
& 6 & 0.057 & 0.054 & 0.985 & 0.983 & 1.000 & 1.000 \\
& 7 & 0.059 & 0.055 & 0.987 & 0.986 & 1.000 & 1.000 \\
& 8 & 0.068 & 0.062 & 0.979 & 0.980 & 1.000 & 1.000 \\
& 9 & 0.041 & 0.043 & 0.979 & 0.979 & 1.000 & 1.000 \\
& 10 & 0.047 & 0.051 & 0.984 & 0.985 & 1.000 & 1.000 \\
\hline
\multirow{10}{*}{$n=1000$} & 1 & 0.053 & 0.051 & 1.000 & 1.000 & 1.000 & 1.000 \\
& 2 & 0.062 & 0.060 & 1.000 & 1.000 & 1.000 & 1.000 \\
& 3 & 0.057 & 0.057 & 1.000 & 1.000 & 1.000 & 1.000 \\
& 4 & 0.061 & 0.061 & 1.000 & 1.000 & 1.000 & 1.000 \\
& 5 & 0.043 & 0.049 & 1.000 & 1.000 & 1.000 & 1.000 \\
& 6 & 0.065 & 0.067 & 1.000 & 1.000 & 1.000 & 1.000 \\
& 7 & 0.054 & 0.052 & 1.000 & 1.000 & 1.000 & 1.000 \\
& 8 & 0.046 & 0.047 & 1.000 & 1.000 & 1.000 & 1.000 \\
& 9 & 0.060 & 0.060 & 1.000 & 1.000 & 1.000 & 1.000 \\
& 10 & 0.057 & 0.055 & 1.000 & 1.000 & 1.000 & 1.000 \\
\hline
\end{tabular}
\end{table}

\begin{table}[htbp]
\centering
\caption{Results for $S_n(\xi,\theta_n,\sigma^2_n)$, model $2$, ordinary smooth case, $\alpha=0.1$}
\begin{tabular}{cccccccc}
\hline
\multirow{2}{*}{$n$} & \multirow{2}{*}{$10c$} & \multicolumn{2}{c}{DGP(0)} & \multicolumn{2}{c}{DGP(1)} & \multicolumn{2}{c}{DGP(2)} \\
\cline{3-8}
& & KS & CvM & KS & CvM & KS & CvM \\
\hline
\multirow{10}{*}{$n=250$} & 1 & 0.090 & 0.089 & 0.890 & 0.893 & 0.992 & 0.991 \\
& 2 & 0.106 & 0.108 & 0.871 & 0.874 & 0.994 & 0.994 \\
& 3 & 0.122 & 0.117 & 0.872 & 0.870 & 0.997 & 0.997 \\
& 4 & 0.102 & 0.104 & 0.889 & 0.887 & 0.992 & 0.995 \\
& 5 & 0.109 & 0.109 & 0.887 & 0.886 & 0.996 & 0.996 \\
& 6 & 0.096 & 0.099 & 0.870 & 0.870 & 0.991 & 0.990 \\
& 7 & 0.106 & 0.101 & 0.879 & 0.880 & 0.994 & 0.994 \\
& 8 & 0.100 & 0.098 & 0.879 & 0.878 & 0.995 & 0.996 \\
& 9 & 0.119 & 0.115 & 0.882 & 0.878 & 0.994 & 0.991 \\
& 10 & 0.112 & 0.111 & 0.885 & 0.882 & 0.997 & 0.998 \\
\hline
\multirow{10}{*}{$n=500$} & 1 & 0.107 & 0.111 & 0.993 & 0.993 & 1.000 & 1.000 \\
& 2 & 0.101 & 0.109 & 0.993 & 0.993 & 1.000 & 1.000 \\
& 3 & 0.084 & 0.086 & 0.995 & 0.995 & 1.000 & 1.000 \\
& 4 & 0.107 & 0.106 & 0.990 & 0.989 & 1.000 & 1.000 \\
& 5 & 0.106 & 0.106 & 0.994 & 0.992 & 1.000 & 1.000 \\
& 6 & 0.085 & 0.075 & 0.987 & 0.986 & 1.000 & 1.000 \\
& 7 & 0.106 & 0.102 & 0.987 & 0.986 & 1.000 & 1.000 \\
& 8 & 0.100 & 0.097 & 0.996 & 0.995 & 1.000 & 1.000 \\
& 9 & 0.091 & 0.091 & 0.992 & 0.991 & 1.000 & 1.000 \\
& 10 & 0.092 & 0.092 & 0.993 & 0.989 & 1.000 & 1.000 \\
\hline
\multirow{10}{*}{$n=1000$} & 1 & 0.103 & 0.106 & 1.000 & 1.000 & 1.000 & 1.000 \\
& 2 & 0.110 & 0.110 & 1.000 & 1.000 & 1.000 & 1.000 \\
& 3 & 0.104 & 0.103 & 1.000 & 1.000 & 1.000 & 1.000 \\
& 4 & 0.095 & 0.095 & 1.000 & 1.000 & 1.000 & 1.000 \\
& 5 & 0.093 & 0.095 & 0.999 & 0.999 & 1.000 & 1.000 \\
& 6 & 0.118 & 0.113 & 1.000 & 1.000 & 1.000 & 1.000 \\
& 7 & 0.103 & 0.103 & 1.000 & 1.000 & 1.000 & 1.000 \\
& 8 & 0.092 & 0.093 & 1.000 & 1.000 & 1.000 & 1.000 \\
& 9 & 0.101 & 0.105 & 1.000 & 1.000 & 1.000 & 1.000 \\
& 10 & 0.099 & 0.099 & 1.000 & 1.000 & 1.000 & 1.000 \\
\hline
\end{tabular}
\end{table}

\begin{table}[htbp]
\centering
\caption{Results for $S_n(\xi,\theta_n,\sigma^2_n)$, model $2$, supersmooth case, $\alpha=0.01$}
\begin{tabular}{cccccccc}
\hline
\multirow{2}{*}{$n$} & \multirow{2}{*}{$10c$} & \multicolumn{2}{c}{DGP(0)} & \multicolumn{2}{c}{DGP(1)} & \multicolumn{2}{c}{DGP(2)} \\
\cline{3-8}
& & KS & CvM & KS & CvM & KS & CvM \\
\hline
\multirow{10}{*}{$n=250$} & 1 & 0.012 & 0.014 & 0.539 & 0.544 & 0.904 & 0.902 \\
& 2 & 0.015 & 0.016 & 0.569 & 0.561 & 0.887 & 0.894 \\
& 3 & 0.017 & 0.016 & 0.565 & 0.563 & 0.914 & 0.922 \\
& 4 & 0.010 & 0.009 & 0.557 & 0.550 & 0.886 & 0.896 \\
& 5 & 0.015 & 0.015 & 0.560 & 0.562 & 0.903 & 0.911 \\
& 6 & 0.014 & 0.013 & 0.533 & 0.535 & 0.907 & 0.928 \\
& 7 & 0.007 & 0.007 & 0.554 & 0.552 & 0.892 & 0.902 \\
& 8 & 0.016 & 0.016 & 0.562 & 0.560 & 0.907 & 0.924 \\
& 9 & 0.014 & 0.014 & 0.575 & 0.580 & 0.901 & 0.917 \\
& 10 & 0.015 & 0.014 & 0.534 & 0.523 & 0.907 & 0.916 \\
\hline
\multirow{10}{*}{$n=500$} & 1 & 0.017 & 0.016 & 0.912 & 0.910 & 1.000 & 0.999 \\
& 2 & 0.010 & 0.008 & 0.929 & 0.930 & 0.998 & 0.999 \\
& 3 & 0.021 & 0.021 & 0.903 & 0.904 & 1.000 & 1.000 \\
& 4 & 0.012 & 0.014 & 0.918 & 0.922 & 0.999 & 1.000 \\
& 5 & 0.012 & 0.011 & 0.917 & 0.920 & 0.998 & 0.998 \\
& 6 & 0.015 & 0.015 & 0.920 & 0.916 & 1.000 & 1.000 \\
& 7 & 0.012 & 0.012 & 0.922 & 0.924 & 1.000 & 1.000 \\
& 8 & 0.006 & 0.008 & 0.917 & 0.921 & 1.000 & 1.000 \\
& 9 & 0.017 & 0.016 & 0.915 & 0.912 & 0.999 & 0.999 \\
& 10 & 0.013 & 0.014 & 0.940 & 0.935 & 0.999 & 0.999 \\
\hline
\multirow{10}{*}{$n=1000$} & 1 & 0.012 & 0.012 & 0.999 & 0.999 & 1.000 & 1.000 \\
& 2 & 0.016 & 0.015 & 1.000 & 1.000 & 1.000 & 1.000 \\
& 3 & 0.016 & 0.016 & 0.997 & 0.997 & 1.000 & 1.000 \\
& 4 & 0.015 & 0.015 & 0.999 & 0.999 & 1.000 & 1.000 \\
& 5 & 0.019 & 0.021 & 0.997 & 0.997 & 1.000 & 1.000 \\
& 6 & 0.011 & 0.011 & 0.999 & 0.999 & 1.000 & 1.000 \\
& 7 & 0.012 & 0.012 & 0.998 & 0.998 & 1.000 & 1.000 \\
& 8 & 0.018 & 0.016 & 0.999 & 0.999 & 1.000 & 1.000 \\
& 9 & 0.012 & 0.012 & 0.999 & 0.999 & 1.000 & 1.000 \\
& 10 & 0.007 & 0.006 & 0.998 & 0.998 & 1.000 & 1.000 \\
\hline
\end{tabular}
\end{table}

\begin{table}[htbp]
\centering
\caption{Results for $S_n(\xi,\theta_n,\sigma^2_n)$, model $2$, supersmooth case, $\alpha=0.05$}
\begin{tabular}{cccccccc}
\hline
\multirow{2}{*}{$n$} & \multirow{2}{*}{$10c$} & \multicolumn{2}{c}{DGP(0)} & \multicolumn{2}{c}{DGP(1)} & \multicolumn{2}{c}{DGP(2)} \\
\cline{3-8}
& & KS & CvM & KS & CvM & KS & CvM \\
\hline
\multirow{10}{*}{$n=250$} & 1 & 0.043 & 0.044 & 0.802 & 0.802 & 0.978 & 0.984 \\
& 2 & 0.035 & 0.035 & 0.769 & 0.772 & 0.988 & 0.991 \\
& 3 & 0.049 & 0.049 & 0.771 & 0.774 & 0.992 & 0.994 \\
& 4 & 0.045 & 0.040 & 0.772 & 0.769 & 0.983 & 0.982 \\
& 5 & 0.047 & 0.048 & 0.770 & 0.769 & 0.986 & 0.990 \\
& 6 & 0.046 & 0.044 & 0.815 & 0.811 & 0.986 & 0.992 \\
& 7 & 0.046 & 0.047 & 0.798 & 0.793 & 0.990 & 0.989 \\
& 8 & 0.056 & 0.056 & 0.803 & 0.808 & 0.983 & 0.985 \\
& 9 & 0.064 & 0.066 & 0.782 & 0.788 & 0.981 & 0.987 \\
& 10 & 0.049 & 0.051 & 0.785 & 0.794 & 0.974 & 0.976 \\
\hline
\multirow{10}{*}{$n=500$} & 1 & 0.042 & 0.040 & 0.976 & 0.973 & 1.000 & 1.000 \\
& 2 & 0.048 & 0.053 & 0.982 & 0.982 & 1.000 & 1.000 \\
& 3 & 0.048 & 0.048 & 0.979 & 0.979 & 1.000 & 1.000 \\
& 4 & 0.050 & 0.050 & 0.979 & 0.981 & 1.000 & 1.000 \\
& 5 & 0.060 & 0.062 & 0.977 & 0.976 & 1.000 & 1.000 \\
& 6 & 0.059 & 0.057 & 0.981 & 0.978 & 1.000 & 1.000 \\
& 7 & 0.042 & 0.039 & 0.982 & 0.983 & 1.000 & 1.000 \\
& 8 & 0.046 & 0.047 & 0.971 & 0.971 & 1.000 & 1.000 \\
& 9 & 0.056 & 0.058 & 0.984 & 0.982 & 1.000 & 1.000 \\
& 10 & 0.043 & 0.047 & 0.980 & 0.979 & 1.000 & 1.000 \\
\hline
\multirow{10}{*}{$n=1000$} & 1 & 0.067 & 0.065 & 1.000 & 1.000 & 1.000 & 1.000 \\
& 2 & 0.062 & 0.061 & 1.000 & 1.000 & 1.000 & 1.000 \\
& 3 & 0.056 & 0.056 & 1.000 & 1.000 & 1.000 & 1.000 \\
& 4 & 0.044 & 0.046 & 1.000 & 1.000 & 1.000 & 1.000 \\
& 5 & 0.054 & 0.054 & 1.000 & 1.000 & 1.000 & 1.000 \\
& 6 & 0.064 & 0.063 & 1.000 & 1.000 & 1.000 & 1.000 \\
& 7 & 0.056 & 0.056 & 1.000 & 1.000 & 1.000 & 1.000 \\
& 8 & 0.055 & 0.055 & 1.000 & 1.000 & 1.000 & 1.000 \\
& 9 & 0.051 & 0.049 & 1.000 & 1.000 & 1.000 & 1.000 \\
& 10 & 0.049 & 0.049 & 1.000 & 1.000 & 1.000 & 1.000 \\
\hline
\end{tabular}
\end{table}

\begin{table}[htbp]
\centering
\caption{Results for $S_n(\xi,\theta_n,\sigma^2_n)$, model $2$, supersmooth case, $\alpha=0.1$}
\begin{tabular}{cccccccc}
\hline
\multirow{2}{*}{$n$} & \multirow{2}{*}{$10c$} & \multicolumn{2}{c}{DGP(0)} & \multicolumn{2}{c}{DGP(1)} & \multicolumn{2}{c}{DGP(2)} \\
\cline{3-8}
& & KS & CvM & KS & CvM & KS & CvM \\
\hline
\multirow{10}{*}{$n=250$} & 1 & 0.091 & 0.089 & 0.871 & 0.868 & 0.995 & 0.997 \\
& 2 & 0.115 & 0.116 & 0.869 & 0.868 & 0.991 & 0.996 \\
& 3 & 0.082 & 0.087 & 0.879 & 0.880 & 0.995 & 0.994 \\
& 4 & 0.126 & 0.128 & 0.850 & 0.856 & 0.997 & 0.995 \\
& 5 & 0.111 & 0.108 & 0.891 & 0.892 & 0.991 & 0.989 \\
& 6 & 0.098 & 0.103 & 0.872 & 0.871 & 0.991 & 0.992 \\
& 7 & 0.101 & 0.101 & 0.867 & 0.865 & 0.997 & 0.996 \\
& 8 & 0.115 & 0.118 & 0.884 & 0.890 & 0.989 & 0.991 \\
& 9 & 0.092 & 0.088 & 0.866 & 0.868 & 0.993 & 0.996 \\
& 10 & 0.115 & 0.114 & 0.880 & 0.880 & 0.990 & 0.992 \\
\hline
\multirow{10}{*}{$n=500$} & 1 & 0.121 & 0.120 & 0.992 & 0.991 & 1.000 & 1.000 \\
& 2 & 0.102 & 0.100 & 0.987 & 0.986 & 1.000 & 1.000 \\
& 3 & 0.098 & 0.102 & 0.991 & 0.991 & 1.000 & 1.000 \\
& 4 & 0.098 & 0.095 & 0.993 & 0.993 & 1.000 & 1.000 \\
& 5 & 0.106 & 0.103 & 0.992 & 0.992 & 1.000 & 1.000 \\
& 6 & 0.096 & 0.096 & 0.990 & 0.992 & 1.000 & 1.000 \\
& 7 & 0.093 & 0.088 & 0.995 & 0.995 & 1.000 & 1.000 \\
& 8 & 0.099 & 0.100 & 0.988 & 0.988 & 1.000 & 1.000 \\
& 9 & 0.110 & 0.111 & 0.991 & 0.993 & 1.000 & 1.000 \\
& 10 & 0.107 & 0.108 & 0.989 & 0.988 & 1.000 & 1.000 \\
\hline
\multirow{10}{*}{$n=1000$} & 1 & 0.097 & 0.098 & 1.000 & 1.000 & 1.000 & 1.000 \\
& 2 & 0.100 & 0.097 & 1.000 & 1.000 & 1.000 & 1.000 \\
& 3 & 0.110 & 0.119 & 1.000 & 1.000 & 1.000 & 1.000 \\
& 4 & 0.100 & 0.098 & 1.000 & 1.000 & 1.000 & 1.000 \\
& 5 & 0.108 & 0.109 & 1.000 & 1.000 & 1.000 & 1.000 \\
& 6 & 0.110 & 0.107 & 1.000 & 1.000 & 1.000 & 1.000 \\
& 7 & 0.100 & 0.104 & 1.000 & 1.000 & 1.000 & 1.000 \\
& 8 & 0.098 & 0.096 & 1.000 & 1.000 & 1.000 & 1.000 \\
& 9 & 0.097 & 0.097 & 1.000 & 1.000 & 1.000 & 1.000 \\
& 10 & 0.095 & 0.090 & 1.000 & 1.000 & 1.000 & 1.000 \\
\hline
\end{tabular}
\end{table}

\begin{table}[htbp]
\centering
\caption{Results for $\hat{S}_n(\xi,\hat{\theta}_n,\hat{\sigma}^2_n)$, model $2$, ordinary smooth case, $\alpha=0.01$}
\begin{tabular}{cccccccc}
\hline
\multirow{2}{*}{$n$} & \multirow{2}{*}{$10c$} & \multicolumn{2}{c}{DGP(0)} & \multicolumn{2}{c}{DGP(1)} & \multicolumn{2}{c}{DGP(2)} \\
\cline{3-8}
& & KS & CvM & KS & CvM & KS & CvM \\
\hline
\multirow{10}{*}{$n=250$} & 1 & 0.009 & 0.009 & 0.550 & 0.549 & 0.907 & 0.921 \\
& 2 & 0.017 & 0.018 & 0.538 & 0.549 & 0.893 & 0.903 \\
& 3 & 0.014 & 0.014 & 0.539 & 0.538 & 0.897 & 0.887 \\
& 4 & 0.014 & 0.013 & 0.556 & 0.561 & 0.909 & 0.907 \\
& 5 & 0.016 & 0.014 & 0.536 & 0.539 & 0.910 & 0.922 \\
& 6 & 0.009 & 0.009 & 0.530 & 0.524 & 0.902 & 0.917 \\
& 7 & 0.008 & 0.009 & 0.574 & 0.567 & 0.916 & 0.916 \\
& 8 & 0.014 & 0.014 & 0.569 & 0.564 & 0.924 & 0.933 \\
& 9 & 0.005 & 0.007 & 0.534 & 0.534 & 0.903 & 0.915 \\
& 10 & 0.013 & 0.012 & 0.558 & 0.562 & 0.891 & 0.900 \\
\hline
\multirow{10}{*}{$n=500$} & 1 & 0.010 & 0.011 & 0.924 & 0.922 & 0.999 & 0.999 \\
& 2 & 0.011 & 0.010 & 0.915 & 0.919 & 1.000 & 1.000 \\
& 3 & 0.017 & 0.016 & 0.923 & 0.923 & 0.999 & 0.999 \\
& 4 & 0.014 & 0.014 & 0.916 & 0.919 & 1.000 & 1.000 \\
& 5 & 0.017 & 0.019 & 0.917 & 0.918 & 0.999 & 1.000 \\
& 6 & 0.011 & 0.012 & 0.916 & 0.916 & 1.000 & 1.000 \\
& 7 & 0.018 & 0.019 & 0.916 & 0.918 & 0.999 & 0.999 \\
& 8 & 0.018 & 0.016 & 0.911 & 0.903 & 1.000 & 1.000 \\
& 9 & 0.015 & 0.015 & 0.920 & 0.917 & 0.999 & 0.998 \\
& 10 & 0.016 & 0.017 & 0.935 & 0.932 & 1.000 & 1.000 \\
\hline
\multirow{10}{*}{$n=1000$} & 1 & 0.016 & 0.016 & 0.998 & 0.998 & 1.000 & 1.000 \\
& 2 & 0.016 & 0.015 & 0.999 & 0.999 & 1.000 & 1.000 \\
& 3 & 0.011 & 0.012 & 1.000 & 1.000 & 1.000 & 1.000 \\
& 4 & 0.015 & 0.015 & 0.999 & 0.999 & 1.000 & 1.000 \\
& 5 & 0.016 & 0.015 & 1.000 & 1.000 & 1.000 & 1.000 \\
& 6 & 0.012 & 0.013 & 0.999 & 0.999 & 1.000 & 1.000 \\
& 7 & 0.009 & 0.010 & 1.000 & 1.000 & 1.000 & 1.000 \\
& 8 & 0.021 & 0.019 & 0.999 & 0.999 & 1.000 & 1.000 \\
& 9 & 0.018 & 0.018 & 0.998 & 0.998 & 1.000 & 1.000 \\
& 10 & 0.012 & 0.011 & 0.995 & 0.996 & 1.000 & 1.000 \\
\hline
\end{tabular}
\end{table}

\begin{table}[htbp]
\centering
\caption{Results for $\hat{S}_n(\xi,\hat{\theta}_n,\hat{\sigma}^2_n)$, model $2$, ordinary smooth case, $\alpha=0.05$}
\begin{tabular}{cccccccc}
\hline
\multirow{2}{*}{$n$} & \multirow{2}{*}{$10c$} & \multicolumn{2}{c}{DGP(0)} & \multicolumn{2}{c}{DGP(1)} & \multicolumn{2}{c}{DGP(2)} \\
\cline{3-8}
& & KS & CvM & KS & CvM & KS & CvM \\
\hline
\multirow{10}{*}{$n=250$} & 1 & 0.055 & 0.052 & 0.794 & 0.795 & 0.991 & 0.991 \\
& 2 & 0.040 & 0.040 & 0.791 & 0.799 & 0.986 & 0.988 \\
& 3 & 0.058 & 0.052 & 0.795 & 0.795 & 0.985 & 0.985 \\
& 4 & 0.046 & 0.050 & 0.805 & 0.796 & 0.986 & 0.981 \\
& 5 & 0.057 & 0.055 & 0.795 & 0.795 & 0.977 & 0.981 \\
& 6 & 0.058 & 0.057 & 0.799 & 0.801 & 0.980 & 0.982 \\
& 7 & 0.053 & 0.050 & 0.793 & 0.797 & 0.981 & 0.984 \\
& 8 & 0.049 & 0.051 & 0.818 & 0.818 & 0.977 & 0.981 \\
& 9 & 0.050 & 0.054 & 0.817 & 0.812 & 0.977 & 0.984 \\
& 10 & 0.045 & 0.047 & 0.812 & 0.816 & 0.980 & 0.987 \\
\hline
\multirow{10}{*}{$n=500$} & 1 & 0.048 & 0.049 & 0.972 & 0.973 & 1.000 & 1.000 \\
& 2 & 0.045 & 0.046 & 0.975 & 0.976 & 1.000 & 1.000 \\
& 3 & 0.058 & 0.058 & 0.982 & 0.980 & 1.000 & 1.000 \\
& 4 & 0.054 & 0.055 & 0.979 & 0.980 & 1.000 & 1.000 \\
& 5 & 0.056 & 0.058 & 0.984 & 0.984 & 1.000 & 1.000 \\
& 6 & 0.063 & 0.066 & 0.977 & 0.978 & 1.000 & 1.000 \\
& 7 & 0.053 & 0.057 & 0.980 & 0.979 & 1.000 & 1.000 \\
& 8 & 0.062 & 0.064 & 0.989 & 0.988 & 1.000 & 1.000 \\
& 9 & 0.052 & 0.054 & 0.984 & 0.984 & 1.000 & 1.000 \\
& 10 & 0.049 & 0.045 & 0.982 & 0.982 & 1.000 & 1.000 \\
\hline
\multirow{10}{*}{$n=1000$} & 1 & 0.055 & 0.054 & 1.000 & 1.000 & 1.000 & 1.000 \\
& 2 & 0.050 & 0.050 & 1.000 & 1.000 & 1.000 & 1.000 \\
& 3 & 0.057 & 0.056 & 1.000 & 1.000 & 1.000 & 1.000 \\
& 4 & 0.046 & 0.052 & 1.000 & 1.000 & 1.000 & 1.000 \\
& 5 & 0.048 & 0.048 & 1.000 & 1.000 & 1.000 & 1.000 \\
& 6 & 0.068 & 0.065 & 1.000 & 1.000 & 1.000 & 1.000 \\
& 7 & 0.057 & 0.054 & 1.000 & 1.000 & 1.000 & 1.000 \\
& 8 & 0.055 & 0.058 & 1.000 & 1.000 & 1.000 & 1.000 \\
& 9 & 0.062 & 0.058 & 1.000 & 1.000 & 1.000 & 1.000 \\
& 10 & 0.056 & 0.059 & 1.000 & 1.000 & 1.000 & 1.000 \\
\hline
\end{tabular}
\end{table}

\begin{table}[htbp]
\centering
\caption{Results for $\hat{S}_n(\xi,\hat{\theta}_n,\hat{\sigma}^2_n)$, model $2$, ordinary smooth case, $\alpha=0.1$}
\begin{tabular}{cccccccc}
\hline
\multirow{2}{*}{$n$} & \multirow{2}{*}{$10c$} & \multicolumn{2}{c}{DGP(0)} & \multicolumn{2}{c}{DGP(1)} & \multicolumn{2}{c}{DGP(2)} \\
\cline{3-8}
& & KS & CvM & KS & CvM & KS & CvM \\
\hline
\multirow{10}{*}{$n=250$} & 1 & 0.103 & 0.100 & 0.894 & 0.893 & 0.992 & 0.993 \\
& 2 & 0.092 & 0.091 & 0.864 & 0.863 & 0.992 & 0.992 \\
& 3 & 0.098 & 0.097 & 0.896 & 0.895 & 0.992 & 0.993 \\
& 4 & 0.097 & 0.098 & 0.879 & 0.880 & 0.996 & 0.996 \\
& 5 & 0.091 & 0.091 & 0.892 & 0.893 & 0.997 & 0.994 \\
& 6 & 0.084 & 0.088 & 0.872 & 0.870 & 0.993 & 0.993 \\
& 7 & 0.114 & 0.106 & 0.883 & 0.883 & 0.997 & 0.998 \\
& 8 & 0.142 & 0.140 & 0.883 & 0.884 & 0.997 & 0.997 \\
& 9 & 0.084 & 0.084 & 0.893 & 0.895 & 0.994 & 0.996 \\
& 10 & 0.102 & 0.104 & 0.895 & 0.895 & 0.994 & 0.992 \\
\hline
\multirow{10}{*}{$n=500$} & 1 & 0.110 & 0.110 & 0.992 & 0.991 & 1.000 & 1.000 \\
& 2 & 0.098 & 0.093 & 0.992 & 0.990 & 1.000 & 1.000 \\
& 3 & 0.090 & 0.090 & 0.995 & 0.995 & 1.000 & 1.000 \\
& 4 & 0.082 & 0.080 & 0.996 & 0.996 & 1.000 & 1.000 \\
& 5 & 0.104 & 0.105 & 0.991 & 0.991 & 1.000 & 1.000 \\
& 6 & 0.108 & 0.100 & 0.992 & 0.992 & 1.000 & 1.000 \\
& 7 & 0.108 & 0.112 & 0.988 & 0.987 & 1.000 & 1.000 \\
& 8 & 0.118 & 0.120 & 0.993 & 0.993 & 1.000 & 1.000 \\
& 9 & 0.109 & 0.102 & 0.988 & 0.989 & 1.000 & 1.000 \\
& 10 & 0.103 & 0.100 & 0.994 & 0.993 & 1.000 & 1.000 \\
\hline
\multirow{10}{*}{$n=1000$} & 1 & 0.112 & 0.114 & 1.000 & 1.000 & 1.000 & 1.000 \\
& 2 & 0.103 & 0.104 & 1.000 & 1.000 & 1.000 & 1.000 \\
& 3 & 0.092 & 0.094 & 1.000 & 1.000 & 1.000 & 1.000 \\
& 4 & 0.098 & 0.101 & 1.000 & 1.000 & 1.000 & 1.000 \\
& 5 & 0.104 & 0.100 & 1.000 & 1.000 & 1.000 & 1.000 \\
& 6 & 0.098 & 0.099 & 1.000 & 1.000 & 1.000 & 1.000 \\
& 7 & 0.107 & 0.107 & 1.000 & 1.000 & 1.000 & 1.000 \\
& 8 & 0.096 & 0.095 & 1.000 & 1.000 & 1.000 & 1.000 \\
& 9 & 0.110 & 0.114 & 1.000 & 1.000 & 1.000 & 1.000 \\
& 10 & 0.115 & 0.114 & 0.999 & 0.999 & 1.000 & 1.000 \\
\hline
\end{tabular}
\end{table}

\begin{table}[htbp]
\centering
\caption{Results for $\hat{S}_n(\xi,\hat{\theta}_n,\hat{\sigma}^2_n)$, model $2$, ordinary smooth case, $\alpha=0.01$}
\begin{tabular}{cccccccc}
\hline
\multirow{2}{*}{$n$} & \multirow{2}{*}{$10c$} & \multicolumn{2}{c}{DGP(0)} & \multicolumn{2}{c}{DGP(1)} & \multicolumn{2}{c}{DGP(2)} \\
\cline{3-8}
& & KS & CvM & KS & CvM & KS & CvM \\
\hline
\multirow{10}{*}{$n=250$} & 1 & 0.012 & 0.013 & 0.539 & 0.537 & 0.895 & 0.910 \\
& 2 & 0.009 & 0.010 & 0.568 & 0.570 & 0.891 & 0.911 \\
& 3 & 0.011 & 0.010 & 0.566 & 0.568 & 0.905 & 0.912 \\
& 4 & 0.008 & 0.008 & 0.538 & 0.537 & 0.899 & 0.903 \\
& 5 & 0.007 & 0.008 & 0.545 & 0.552 & 0.903 & 0.905 \\
& 6 & 0.014 & 0.013 & 0.528 & 0.530 & 0.891 & 0.911 \\
& 7 & 0.014 & 0.014 & 0.562 & 0.558 & 0.904 & 0.907 \\
& 8 & 0.019 & 0.018 & 0.543 & 0.542 & 0.876 & 0.887 \\
& 9 & 0.014 & 0.012 & 0.554 & 0.560 & 0.905 & 0.919 \\
& 10 & 0.013 & 0.012 & 0.558 & 0.558 & 0.915 & 0.913 \\
\hline
\multirow{10}{*}{$n=500$} & 1 & 0.007 & 0.009 & 0.933 & 0.935 & 0.999 & 1.000 \\
& 2 & 0.016 & 0.018 & 0.911 & 0.909 & 0.998 & 0.997 \\
& 3 & 0.010 & 0.011 & 0.918 & 0.917 & 0.998 & 0.998 \\
& 4 & 0.009 & 0.012 & 0.900 & 0.901 & 1.000 & 1.000 \\
& 5 & 0.010 & 0.011 & 0.922 & 0.921 & 1.000 & 1.000 \\
& 6 & 0.018 & 0.017 & 0.918 & 0.919 & 1.000 & 1.000 \\
& 7 & 0.017 & 0.016 & 0.921 & 0.922 & 1.000 & 1.000 \\
& 8 & 0.010 & 0.011 & 0.926 & 0.924 & 0.999 & 1.000 \\
& 9 & 0.020 & 0.021 & 0.917 & 0.915 & 0.999 & 1.000 \\
& 10 & 0.018 & 0.019 & 0.917 & 0.914 & 0.999 & 0.999 \\
\hline
\multirow{10}{*}{$n=1000$} & 1 & 0.020 & 0.020 & 0.999 & 0.999 & 1.000 & 1.000 \\
& 2 & 0.015 & 0.017 & 0.998 & 0.998 & 1.000 & 1.000 \\
& 3 & 0.012 & 0.013 & 1.000 & 1.000 & 1.000 & 1.000 \\
& 4 & 0.012 & 0.013 & 1.000 & 1.000 & 1.000 & 1.000 \\
& 5 & 0.011 & 0.013 & 0.999 & 0.999 & 1.000 & 1.000 \\
& 6 & 0.014 & 0.012 & 0.999 & 0.999 & 1.000 & 1.000 \\
& 7 & 0.015 & 0.014 & 0.999 & 0.999 & 1.000 & 1.000 \\
& 8 & 0.009 & 0.011 & 1.000 & 1.000 & 1.000 & 1.000 \\
& 9 & 0.017 & 0.019 & 0.995 & 0.995 & 1.000 & 1.000 \\
& 10 & 0.014 & 0.014 & 0.999 & 0.999 & 1.000 & 1.000 \\
\hline
\end{tabular}
\end{table}

\begin{table}[htbp]
\centering
\caption{Results for $\hat{S}_n(\xi,\hat{\theta}_n,\hat{\sigma}^2_n)$, model $2$, ordinary smooth case, $\alpha=0.05$}
\begin{tabular}{cccccccc}
\hline
\multirow{2}{*}{$n$} & \multirow{2}{*}{$10c$} & \multicolumn{2}{c}{DGP(0)} & \multicolumn{2}{c}{DGP(1)} & \multicolumn{2}{c}{DGP(2)} \\
\cline{3-8}
& & KS & CvM & KS & CvM & KS & CvM \\
\hline
\multirow{10}{*}{$n=250$} & 1 & 0.053 & 0.059 & 0.778 & 0.779 & 0.977 & 0.978 \\
& 2 & 0.057 & 0.057 & 0.806 & 0.806 & 0.979 & 0.983 \\
& 3 & 0.065 & 0.068 & 0.784 & 0.772 & 0.976 & 0.983 \\
& 4 & 0.048 & 0.052 & 0.777 & 0.774 & 0.983 & 0.985 \\
& 5 & 0.065 & 0.066 & 0.798 & 0.796 & 0.976 & 0.984 \\
& 6 & 0.049 & 0.049 & 0.793 & 0.799 & 0.971 & 0.982 \\
& 7 & 0.037 & 0.037 & 0.769 & 0.770 & 0.981 & 0.977 \\
& 8 & 0.050 & 0.052 & 0.803 & 0.803 & 0.983 & 0.987 \\
& 9 & 0.061 & 0.062 & 0.788 & 0.782 & 0.985 & 0.988 \\
& 10 & 0.059 & 0.061 & 0.770 & 0.764 & 0.975 & 0.980 \\
\hline
\multirow{10}{*}{$n=500$} & 1 & 0.042 & 0.039 & 0.981 & 0.979 & 1.000 & 1.000 \\
& 2 & 0.054 & 0.051 & 0.986 & 0.986 & 1.000 & 1.000 \\
& 3 & 0.062 & 0.060 & 0.984 & 0.985 & 1.000 & 1.000 \\
& 4 & 0.060 & 0.057 & 0.976 & 0.975 & 1.000 & 1.000 \\
& 5 & 0.049 & 0.051 & 0.974 & 0.976 & 0.999 & 1.000 \\
& 6 & 0.053 & 0.054 & 0.982 & 0.980 & 1.000 & 1.000 \\
& 7 & 0.054 & 0.055 & 0.978 & 0.978 & 1.000 & 1.000 \\
& 8 & 0.047 & 0.045 & 0.970 & 0.970 & 1.000 & 1.000 \\
& 9 & 0.048 & 0.049 & 0.980 & 0.980 & 1.000 & 1.000 \\
& 10 & 0.051 & 0.051 & 0.979 & 0.976 & 1.000 & 1.000 \\
\hline
\multirow{10}{*}{$n=1000$} & 1 & 0.062 & 0.062 & 1.000 & 1.000 & 1.000 & 1.000 \\
& 2 & 0.056 & 0.053 & 1.000 & 1.000 & 1.000 & 1.000 \\
& 3 & 0.052 & 0.052 & 1.000 & 1.000 & 1.000 & 1.000 \\
& 4 & 0.058 & 0.055 & 1.000 & 1.000 & 1.000 & 1.000 \\
& 5 & 0.065 & 0.063 & 1.000 & 1.000 & 1.000 & 1.000 \\
& 6 & 0.052 & 0.050 & 1.000 & 1.000 & 1.000 & 1.000 \\
& 7 & 0.051 & 0.051 & 1.000 & 1.000 & 1.000 & 1.000 \\
& 8 & 0.057 & 0.061 & 1.000 & 1.000 & 1.000 & 1.000 \\
& 9 & 0.051 & 0.050 & 1.000 & 1.000 & 1.000 & 1.000 \\
& 10 & 0.069 & 0.067 & 1.000 & 1.000 & 1.000 & 1.000 \\
\hline
\end{tabular}
\end{table}

\begin{table}[htbp]
\centering
\caption{Results for $\hat{S}_n(\xi,\hat{\theta}_n,\hat{\sigma}^2_n)$, model $2$, ordinary smooth case, $\alpha=0.1$}
\begin{tabular}{cccccccc}
\hline
\multirow{2}{*}{$n$} & \multirow{2}{*}{$10c$} & \multicolumn{2}{c}{DGP(0)} & \multicolumn{2}{c}{DGP(1)} & \multicolumn{2}{c}{DGP(2)} \\
\cline{3-8}
& & KS & CvM & KS & CvM & KS & CvM \\
\hline
\multirow{10}{*}{$n=250$} & 1 & 0.114 & 0.114 & 0.887 & 0.893 & 0.994 & 0.994 \\
& 2 & 0.110 & 0.111 & 0.873 & 0.876 & 0.995 & 0.996 \\
& 3 & 0.095 & 0.100 & 0.875 & 0.870 & 0.995 & 0.996 \\
& 4 & 0.121 & 0.122 & 0.880 & 0.882 & 0.995 & 0.996 \\
& 5 & 0.108 & 0.109 & 0.871 & 0.875 & 0.994 & 0.995 \\
& 6 & 0.119 & 0.117 & 0.873 & 0.878 & 0.991 & 0.993 \\
& 7 & 0.106 & 0.106 & 0.879 & 0.882 & 0.993 & 0.993 \\
& 8 & 0.098 & 0.097 & 0.874 & 0.875 & 0.988 & 0.988 \\
& 9 & 0.100 & 0.104 & 0.885 & 0.884 & 0.993 & 0.993 \\
& 10 & 0.098 & 0.102 & 0.883 & 0.884 & 0.995 & 0.993 \\
\hline
\multirow{10}{*}{$n=500$} & 1 & 0.106 & 0.107 & 0.994 & 0.994 & 1.000 & 1.000 \\
& 2 & 0.109 & 0.115 & 0.993 & 0.993 & 1.000 & 1.000 \\
& 3 & 0.107 & 0.108 & 0.992 & 0.991 & 1.000 & 1.000 \\
& 4 & 0.100 & 0.100 & 0.995 & 0.995 & 1.000 & 1.000 \\
& 5 & 0.104 & 0.101 & 0.992 & 0.992 & 1.000 & 1.000 \\
& 6 & 0.088 & 0.084 & 0.994 & 0.994 & 1.000 & 1.000 \\
& 7 & 0.117 & 0.119 & 0.988 & 0.988 & 1.000 & 1.000 \\
& 8 & 0.114 & 0.111 & 0.995 & 0.994 & 1.000 & 1.000 \\
& 9 & 0.102 & 0.108 & 0.990 & 0.990 & 1.000 & 1.000 \\
& 10 & 0.099 & 0.101 & 0.990 & 0.991 & 1.000 & 1.000 \\
\hline
\multirow{10}{*}{$n=1000$} & 1 & 0.109 & 0.114 & 1.000 & 1.000 & 1.000 & 1.000 \\
& 2 & 0.109 & 0.106 & 1.000 & 1.000 & 1.000 & 1.000 \\
& 3 & 0.095 & 0.094 & 1.000 & 1.000 & 1.000 & 1.000 \\
& 4 & 0.115 & 0.114 & 1.000 & 1.000 & 1.000 & 1.000 \\
& 5 & 0.113 & 0.114 & 1.000 & 1.000 & 1.000 & 1.000 \\
& 6 & 0.125 & 0.129 & 1.000 & 1.000 & 1.000 & 1.000 \\
& 7 & 0.112 & 0.114 & 1.000 & 1.000 & 1.000 & 1.000 \\
& 8 & 0.100 & 0.096 & 1.000 & 1.000 & 1.000 & 1.000 \\
& 9 & 0.118 & 0.115 & 1.000 & 1.000 & 1.000 & 1.000 \\
& 10 & 0.097 & 0.096 & 1.000 & 1.000 & 1.000 & 1.000 \\
\hline
\end{tabular}
\end{table}

\section{Some Definitions}\label{sec.AppendixB}
Define $d_i=\{Y_i,W_i\}$ and $D_i=\{Y_i,W_i,W_i^r\}$ where $\{Y_i,W_i\}_{i=1}^n$ is a random sample of observables mentioned in Section \ref{sec.Test} and $\{W_i^r\}_{i=1}^n$ is a sample of repeated measurements mentioned in Section \ref{sec.Unknown}. 

In this subsection, we first give the notations to be used, and then begin our proofs of theorems and lemmas. We define the deconvolution kernel as mentioned in section \ref{sec.Test}, where $d$ represent the dimension of the variable $X$. 
\begin{align*}
    &\mathcal{K}_b(a)=\frac{1}{2\pi b}\int \ee^{-\ii ta}\frac{K^{\text{ft}}(t)}{f_\epsilon^{\text{ft}}(t/b)}\,dt, 
\end{align*}
And for convenience, we give the following notations about the kernel function,
\begin{align*}
    &\mathcal{K}_\epsilon(x) = b\mathcal{K}_b(x),\,\,\,
    \hat{\mathcal{K}}_\epsilon(x) = b\hat{\mathcal{K}}_b(x)\,\,\,
    \mathcal{K}_{\epsilon,1}(x) = b\mathcal{K}_{b,1}(x),\,\,\,
    \mathcal{K}_{\epsilon,2}(x) = b\mathcal{K}_{b,2}(x) .
\end{align*}
As explained in \cite{fan1992deconvolution}, the idea behind the construction of the kernel function is to use Fourier transform and its inverse, possessing the following properties, $\mathcal{F}(f\ast g) = \mathcal{F}(f)\mathcal{F}(g)$ and $\mathcal{F}^{-1}(fg) = \mathcal{F}^{-1}(f)\ast\mathcal{F}^{-1}(g)$, where
\begin{align*}
    &\mathcal{F}(f) = \frac{1}{2\pi}\int\ee^{\ii xt}f(x)\,dx, \,\,\, \mathcal{F}^{-1}(g) = \int\ee^{\ii xt}g(t)\,dt.
\end{align*}
Then for the unknown measurement error case mentioned in section \ref{sec.Unknown}, we denote,
\begin{align*}
    &\hat{f}_\epsilon^{\text{ft}}(t) = \left\vert\frac{1}{n}\sum_{j=1}^n\zeta_j(t)\right\vert^{1/2},\,\,\,\zeta_j(t) = \cos\left[t(W_j-W_j^r)\right].
\end{align*}
Next, we decompose deconvolution kernel with the Fr\'{e}chet derivative as mentioned in \cite{dong2022nonparametric},
\begin{align}
    &\mathcal{K}_{\epsilon,1}(a)=\frac{1}{2\pi}\int \ee^{-\ii t^\top a}\frac{K^{\text{ft}}(t)}{f_\epsilon^{\text{ft}}(t/b)}\psi_1(t/b)\,dt,\notag \\ 
    &\psi_1(t) = \frac{\left[f_\epsilon^{\text{ft}}(t)\right]^2 - \left[\hat{f}_\epsilon^{\text{ft}}(t)\right]^2}{2\left[f_\epsilon^{\text{ft}}(t)\right]^2},\label{Decompose of est of decon part1}\\
    &\mathcal{K}_{\epsilon,2}(a)=\frac{1}{2\pi}\int \ee^{-\ii t^\top a}\frac{K^{\text{ft}}(t)}{f_\epsilon^{\text{ft}}(t/b)}\psi_2(t/b)\,dt,\notag \\ 
    &\psi_2(t) = \frac{\left[\hat{f}_\epsilon^{\text{ft}}(t)+2f_\epsilon^{\text{ft}}(t)\right]\left\{\left[\hat{f}_\epsilon^{\text{ft}}(t)\right]^2 - \left[f_\epsilon^{\text{ft}}(t)\right]^2\right\}^2}{2\hat{f}_\epsilon^{\text{ft}}(t)\left[f_\epsilon^{\text{ft}}(t)\right]^2\left[\hat{f}_\epsilon^{\text{ft}}(t)+f_\epsilon^{\text{ft}}(t)\right]^2}.\label{Decompose of est of decon part2}
\end{align}
For convenience, we define the following transformation of kernels,
\begin{align*} 
    &K_{\epsilon,j}(x) = \frac{1}{2\pi}\int \ee^{-\ii tx}K^{\mathrm{ft}}(t)\psi_j(t/b)\,dt, \,\,\, j=1,2.
\end{align*}
Thus we can decompose the estimation of deconvolution kernel by the following equation,
\begin{align*} 
    &\hat{\mathcal{K}}_\epsilon(x) = \mathcal{K}_\epsilon(x) + \mathcal{K}_{\epsilon,1}(x) + \mathcal{K}_{\epsilon,2}(x).
\end{align*}
We also define
\begin{align*}
    &S_n(\xi,\theta_0,\sigma^2_0)=\frac{1}{n}\sum_{i=1}^n\int\left[\left(Y_i-g(x;\theta_0)\right)^2-\sigma_0^2\right]\mathcal{K}_b\left(\frac{x-W_i}{b}\right)\ee^{\ii x\xi}\,dx,\\
    &\hat{S}_n(\xi,\theta_0,\sigma^2_0)=\frac{1}{n}\sum_{i=1}^n\int\left[\left(Y_i-g(x;\theta_0)\right)^2-\sigma_0^2\right]\hat{\mathcal{K}}_b\left(\frac{x-W_i}{b}\right)\ee^{\ii x\xi}\,dx,\\
    &S^\ast_n(\xi,\theta_0,\sigma^2_0)=\frac{1}{n}\sum_{i=1}^nV_i\int\left[\left(Y_i-g(x;\theta_0)\right)^2-\sigma_0^2\right]\mathcal{K}_b\left(\frac{x-W_i}{b}\right)\ee^{\ii x\xi}\,dx,\\
    &\hat{S}_n^\ast(\xi,\theta_0,\sigma^2_0)=\frac{1}{n}\sum_{i=1}^n\int\left[\left(Y_i-g(x;\theta_0)\right)^2-\sigma_0^2\right]\hat{\mathcal{K}}_b^\ast\left(\frac{x-W_i}{b}\right)\ee^{\ii x\xi}\,dx,
\end{align*}
and
\begin{align*}
    &G_n(\xi) = \frac{1}{n}\sum_{i=1}^n\int\mathcal{K}_b\left(\frac{x-W_i}{b}\right)\ee^{\ii x\xi}\,dx,\\
    &\hat{G}_n(\xi) = \frac{1}{n}\sum_{i=1}^n\int\hat{\mathcal{K}}_b\left(\frac{x-W_i}{b}\right)\ee^{\ii x\xi}\,dx,\\
    &G^\ast_n(\xi) = \frac{1}{n}\sum_{i=1}^nV_i\int\mathcal{K}_b\left(\frac{x-W_i}{b}\right)\ee^{\ii x\xi}\,dx,\\
    &\hat{G}^\ast_n(\xi) = \frac{1}{n}\sum_{i=1}^n\int\hat{\mathcal{K}}_b^\ast\left(\frac{x-W_i}{b}\right)\ee^{\ii x\xi}\,dx
\end{align*}
for notational simplicity.

\section{Proofs of Theorems}\label{sec.AppendixC}
\begin{proof}[Proof of Theorem \ref{theorem.known ordinary smooth under H0}]
    We start be decomposing
    \begin{align}\label{theoremproof.known ordinary smooth under H0 statdecomp}
        S_n(\xi,\theta_n,\sigma^2_n)  = S_n(\xi,\theta_0,\sigma^2_0)+S_{n1}(\xi,\theta_n)-S_{n2}(\xi,\sigma^2_n),
    \end{align}
    where
    \begin{align*}
        &S_{n1}(\xi,\theta_n) = \frac{1}{n}\sum_{i=1}^n\int\left[\left(Y_i-g(x;\theta_n)\right)^2-\left(Y_i-g(x;\theta_0)\right)^2\right]\mathcal{K}_b\left(\frac{x-W_i}{b}\right)\ee^{\ii x\xi}\,dx,\\
        &S_{n2}(\xi,\sigma^2_n) = \left(\sigma_n^2-\sigma_0^2\right)\left[\frac{1}{n}\sum_{i=1}^n\int\mathcal{K}_b\left(\frac{x-W_i}{b}\right)\ee^{\ii x\xi}\,dx\right] = \left(\sigma_n^2-\sigma_0^2\right)G_n(\xi).
    \end{align*}
    For the main term, noting that the null hypothesis\eqref{hyp.H0} implies $\mathbb{E}[(U^2-\sigma_0^2)\ee^{\ii X\xi}]=0$, along with the proof of Lemma \ref{lemma.main term known}, we obtain
    \begin{align}\label{theoremproof.known ordinary smooth under H0 main term}
        \sup_{\xi\in\Pi}\left\vert \sqrt{n}S_n(\xi,\theta_0,\sigma^2_0) -\frac{1}{\sqrt{n}}\sum_{i=1}^n  \left\{r_{1,\infty}(d_i;\xi) - \mathbb{E}\left[r_{1,\infty}(d_i;\xi)\right]\right\}\right\vert = o_p\left(1\right).
    \end{align}
    For the term $S_{n2}(\xi,\sigma^2_n)$, which represents the \textquotedblleft parametric estimation effect\textquotedblright, we note that $\sigma_n^2$ can be decomposed as
    \begin{align}\label{theoremproof.known ordinary smooth under H0 sigmadecomp}
        &\sigma_n^2-\sigma_0^2 = S_n(0,\theta_0,\sigma^2_0)+S_{n1}(0,\theta_n).
    \end{align}
    Consequently, 
    \begin{align*}
        \left\vert \sqrt{n}\left(\sigma_n^2-\sigma_0^2\right) -\frac{1}{\sqrt{n}}\sum_{i=1}^n  \left\{r_{1,\infty}(d_i;0) - \mathbb{E}\left[r_{1,\infty}(d_i;0)\right]\right\}\right\vert = o_p\left(1\right)
    \end{align*}
    is implied by null hypothesis\eqref{hyp.H0}, \eqref{theoremproof.known ordinary smooth under H0 main term} and Lemma \ref{lemma.negligible term known}. Together with conclusions in Lemma \ref{lemma.negligible term known} and \ref{lemma.Gn known}, 
    \begin{align}\label{theoremproof.known ordinary smooth under H0 negligible term}
        \sup_{\xi\in\Pi}\left\vert \sqrt{n}S_{n1}(\xi,\theta_n)\right\vert = o_p\left(1\right)
    \end{align}
    and
    \begin{align}\label{theoremproof.known ordinary smooth under H0 parameffect}
        \sup_{\xi\in\Pi}\left\vert \sqrt{n}S_{n2}(\xi,\sigma^2_n)-\frac{f^{\mathrm{ft}}_X(\xi)}{\sqrt{n}}\sum_{i=1}^n  \left\{r_{1,\infty}(d_i;0) - \mathbb{E}\left[r_{1,\infty}(d_i;0)\right]\right\}\right\vert = o_p\left(1\right)
    \end{align}
    hold. Combining \eqref{theoremproof.known ordinary smooth under H0 statdecomp}, \eqref{theoremproof.known ordinary smooth under H0 main term}, \eqref{theoremproof.known ordinary smooth under H0 negligible term} and \eqref{theoremproof.known ordinary smooth under H0 parameffect}, conclusions in Theorem \ref{theorem.known ordinary smooth under H0} hold.
\end{proof}

\begin{proof}[Proof of Theorem \ref{theorem.known supersmooth under H0}]
    The proof is identical to that of Theorem \ref{theorem.known ordinary smooth under H0} except that the limiting process of the test statistic is $S^{ss}_{\infty}(\xi,\theta_0,\sigma^2_0)$ which depends on $r_{2,\infty}(d_i;\xi)$, rather than $S^{os}_{\infty}(\xi,\theta_0,\sigma^2_0)$ which depends on $r_{1,\infty}(d_i;\xi)$. We first decompose $S_n(\xi,\theta_n,\sigma^2_n)$ as \eqref{theoremproof.known ordinary smooth under H0 statdecomp}, then invoke the null hypothesis to establish the negligibility of the bias term, and finally results from Lemma \ref{lemma.main term known}, \ref{lemma.negligible term known} and \ref{lemma.Gn known} are applied, as their conclusions are stated to be valid in the supersmooth case as well.
\end{proof}

\begin{proof}[Proof of Theorem \ref{theorem.known under H1n}]
    It is important to note that, in the proof of Theorem \ref{theorem.known ordinary smooth under H0} and \ref{theorem.known supersmooth under H0} concerning the limiting process, the convergence of $S_{n1}(\xi,\theta_n)$ does not depend on the null hypothesis. Therefore, conclusion about $S_{n1}(\xi,\theta_n)$ in \eqref{theoremproof.known ordinary smooth under H0 negligible term} remains valid. For the analysis of main term, we observe that under the local alternative hypothesis\eqref{hyp.H1}, 
    \begin{align*}
        \mathbb{E}\left[\left(U^2-\sigma_0^2\right)\ee^{\ii X\xi}\right] = \frac{1}{\sqrt{n}}\mathbb{E}\left[\Delta(X)\ee^{\ii X\xi}\right], \quad K^{\mathrm{ft}}(b\xi)\to 1,
    \end{align*}
    which permits
    \begin{align}\label{theoremproof.known under H1n main term ordi}
        \sup_{\xi\in\Pi}\left\vert\sqrt{n}S_n(\xi,\theta_0,\sigma^2_0)-\frac{1}{\sqrt{n}}\sum_{i=1}^n  \left\{r_{1,\infty}(d_i;\xi) - \mathbb{E}\left[r_{1,\infty}(d_i;\xi)\right]\right\}-\mathbb{E}\left[\Delta(X)\ee^{\ii X\xi}\right]\right\vert = o_p\left(1\right)
    \end{align}
    for the ordinary smooth case and 
    \begin{align}\label{theoremproof.known under H1n main term super}
        \sup_{\xi\in\Pi}\left\vert\sqrt{n}S_n(\xi,\theta_0,\sigma^2_0)-\frac{1}{\sqrt{n}}\sum_{i=1}^n  \left\{r_{2,\infty}(d_i;\xi) - \mathbb{E}\left[r_{2,\infty}(d_i;\xi)\right]\right\}-\mathbb{E}\left[\Delta(X)\ee^{\ii X\xi}\right]\right\vert = o_p\left(1\right)
    \end{align}
    for the supersmooth case. Furthermore, since \eqref{theoremproof.known under H1n main term ordi} and \eqref{theoremproof.known under H1n main term super} both imply $\vert S_n(0,\theta_0,\sigma^2_0)\vert=O_p(n^{-1/2})$, along with decomposition \eqref{theoremproof.known ordinary smooth under H0 sigmadecomp} and $\vert S_{n1}(0,\theta_n)\vert=o_p(n^{-1/2})$ implied by \eqref{theoremproof.known ordinary smooth under H0 negligible term}, 
    \begin{align}\label{theoremproof.known under H1n sigma}
        \sup_{\xi\in\Pi}\left\vert \sqrt{n}S_{n2}(\xi,\sigma^2_n)-\sqrt{n}S_n(0,\theta_0,\sigma^2_0)f_X^{\mathrm{ft}}(\xi)\right\vert = o_p\left(1\right)
    \end{align}
    follows as a consequence. Thus, Theorem \ref{theorem.known under H1n} follows by combining \eqref{theoremproof.known ordinary smooth under H0 statdecomp}, \eqref{theoremproof.known ordinary smooth under H0 parameffect}, either \eqref{theoremproof.known under H1n main term ordi} for ordinary smooth case or \eqref{theoremproof.known under H1n main term super} for supersmooth case and \eqref{theoremproof.known under H1n sigma}.
\end{proof}

\begin{proof}[Proof of Theorem \ref{theorem.known under H1}]
    Notice that the convergence of $S_{n1}(\xi,\theta_n)$ mentioned in \eqref{theoremproof.known ordinary smooth under H0 negligible term} still holds, along with decomposition \eqref{theoremproof.known ordinary smooth under H0 sigmadecomp} and Lemma \ref{lemma.Gn known}, we obtain
    \begin{align*}
        \sup_{\xi\in\Pi}\left\vert \sqrt{n}S_{n2}(\xi,\sigma^2_n)-\sqrt{n}S_n(0,\theta_0,\sigma^2_0)G_n(\xi)\right\vert = o_p\left(1\right).
    \end{align*}
    Additionally, note that \eqref{theoremproof.known ordinary smooth under H0 statdecomp} provides a decomposition of $S_n(\xi,\theta_n,\sigma^2_n)$,
    \begin{align}\label{theorem.known under H1 decomp}
        \left\vert S_n(\xi,\theta_n,\sigma^2_n) - \left[S_n(\xi,\theta_0,\sigma^2_0)-S_n(0,\theta_0,\sigma^2_0)G_n(\xi)\right]\right\vert = o_p\left(n^{-\frac{1}{2}}\right).
    \end{align}
    For the main term, 
    \begin{align}\label{theorem.known under H1 mainterm}
        \left\vert S_n(\xi,\theta_0,\sigma^2_0)-\mathbb{E}\left[\left(U^2-\sigma_0^2\right)\ee^{\ii X\xi}\right]\right\vert =o_p\left(1\right)
    \end{align}
    is implied by $K^{\mathrm{ft}}(b\xi)\to1$ and Lemma \ref{lemma.main term known}. Consequently, Theorem \ref{theorem.known under H1} follows by combining \eqref{theorem.known under H1 decomp}, \eqref{theorem.known under H1 mainterm} and Lemma \ref{lemma.Gn known}.
\end{proof}

\begin{proof}[Proof of Theorem \ref{theorem.unknown ordinary smooth under H0}]
    The proof follows along the same line as Theorem \ref{theorem.known ordinary smooth under H0}. Specifically, the test statistic $\hat{S}_n(\xi,\hat{\theta}_n,\hat{\sigma}^2_n)$ is rewritten through a decomposition analogous to \eqref{theoremproof.known ordinary smooth under H0 statdecomp}, as shown
    \begin{align}\label{theoremproof.unknown ordinary smooth under H0 statdecomp}
        \hat{S}_n(\xi,\hat{\theta}_n,\hat{\sigma}^2_n)  = \hat{S}_n(\xi,\theta_0,\sigma^2_0)+\hat{S}_{n1}(\xi,\hat{\theta}_n)-\hat{S}_{n2}(\xi,\hat{\sigma}^2_n),
    \end{align}
    where
    \begin{align*}
        &\hat{S}_{n1}(\xi,\hat{\theta}_n) = \frac{1}{n}\sum_{i=1}^n\int\left[\left(Y_i-g(x;\hat{\theta}_n)\right)^2-\left(Y_i-g(x;\theta_0)\right)^2\right]\hat{\mathcal{K}}_b\left(\frac{x-W_i}{b}\right)\ee^{\ii x\xi}\,dx,\\
        &\hat{S}_{n2}(\xi,\hat{\sigma}^2_n) = \left(\hat{\sigma}_n^2-\sigma_0^2\right)\left[\frac{1}{n}\sum_{i=1}^n\int\hat{\mathcal{K}}_b\left(\frac{x-W_i}{b}\right)\ee^{\ii x\xi}\,dx\right] = \left(\hat{\sigma}_n^2-\sigma_0^2\right)\hat{G}_n(\xi).
    \end{align*}
    The bias of main term $\hat{S}_n(\xi,\theta_0,\sigma^2_0)$ vanishes relying on the implication of the null hypothesis, and the convergence result is then derived by following arguments in the proof of Lemma \ref{lemma.main term unknown}, leading to the following conclusion,
    \begin{align}\label{theoremproof.unknown ordinary smooth under H0 main term}
        \sup_{\xi\in\Pi}\left\vert \sqrt{n}\hat{S}_n(\xi,\theta_0,\sigma^2_0) -\frac{1}{\sqrt{n}}\sum_{i=1}^n  \left\{\hat{r}_{1,\infty}(D_i;\xi) - \mathbb{E}\left[\hat{r}_{1,\infty}(D_i;\xi)\right]\right\}\right\vert = o_p\left(1\right),
    \end{align}
    where $\hat{r}_{1,\infty}(D_i;\xi)$ represents $r^{\epsilon,os}_\infty(Y,W,W^r;\xi,\theta_0,\sigma_0^2)$ mentioned in Theorem \ref{theorem.unknown ordinary smooth under H0}. For the term representing the \textquotedblleft parametric estimation effect\textquotedblright, $\hat{\sigma}_n^2-\sigma_0^2 = \hat{S}_n(0,\theta_0,\sigma^2_0)+\hat{S}_{n1}(0,\hat{\theta}_n)$ still holds. Given that 
    \begin{align}\label{theoremproof.unknown ordinary smooth under H0 negligible term}
        \sup_{\xi\in\Pi}\left\vert \sqrt{n}\hat{S}_{n1}(\xi,\hat{\theta}_n)\right\vert = o_p\left(1\right),
    \end{align}
    follows from Lemma \ref{lemma.negligible term unknown}, 
    \begin{align}\label{theoremproof.unknown ordinary smooth under H0 sigma}
        \sup_{\xi\in\Pi}\left\vert \sqrt{n}\left(\hat{\sigma}_n^2-\sigma_0^2\right) -\frac{1}{\sqrt{n}}\sum_{i=1}^n  \left\{\hat{r}_{1,\infty}(D_i;0) - \mathbb{E}\left[\hat{r}_{1,\infty}(D_i;0)\right]\right\}\right\vert = o_p\left(1\right)
    \end{align}
    can be derived by \eqref{theoremproof.unknown ordinary smooth under H0 main term} and \eqref{theoremproof.unknown ordinary smooth under H0 negligible term}, thereby establishing
    \begin{align}\label{theoremproof.unknown ordinary smooth under H0 parameffect}
        \sup_{\xi\in\Pi}\left\vert \sqrt{n}\hat{S}_{n2}(\xi,\sigma^2_n)-\frac{f^{\mathrm{ft}}_X(\xi)}{\sqrt{n}}\sum_{i=1}^n  \left\{\hat{r}_{1,\infty}(D_i;0) - \mathbb{E}\left[\hat{r}_{1,\infty}(D_i;0)\right]\right\}\right\vert = o_p\left(1\right)
    \end{align}
    by \eqref{theoremproof.unknown ordinary smooth under H0 sigma} and Lemma \ref{lemma.Gn unknown}. Conclusions in Theorem \ref{theorem.unknown ordinary smooth under H0} follow by combining \eqref{theoremproof.unknown ordinary smooth under H0 statdecomp}, \eqref{theoremproof.unknown ordinary smooth under H0 main term}, \eqref{theoremproof.unknown ordinary smooth under H0 negligible term} and \eqref{theoremproof.unknown ordinary smooth under H0 parameffect}.
\end{proof}

\begin{proof}[Proof of Theorem \ref{theorem.unknown supersmooth under H0}]
    As stated in Theorem \ref{theorem.known supersmooth under H0}, the proof of this theorem is identical to that of Theorem \ref{theorem.unknown supersmooth under H0}, except that $\hat{S}^{os}_{\infty}(\xi,\theta_0,\sigma^2_0)$ which depends on $\hat{r}_{1,\infty}(D_i;\xi)$ is replaced by $\hat{S}^{ss}_{\infty}(\xi,\theta_0,\sigma^2_0)$ which depends on $\hat{r}_{2,\infty}(D_i;\xi)$. Therefore, we omit the detailed proof.
\end{proof}

\begin{proof}[Proof of Theorem \ref{theorem.unknown under H1n}]
    Lemma \ref{lemma.main term unknown}, \ref{lemma.negligible term unknown} and \ref{lemma.Gn unknown} show that the estimation of deconvolution kernel does not influence the bias of the proposed statistics. In addition, Theorem \ref{theorem.unknown ordinary smooth under H0} and \ref{theorem.unknown supersmooth under H0} demonstrate that the convergence rate of the main term is unaffected by the presence of an estimated measurement error distribution. Therefore, the proof of this theorem is analogous to that of Theorem \ref{theorem.known under H1n}, except for the specific form of the main term, where $r_{1,\infty}(d_i;\xi)$ is replaced by $\hat{r}_{1,\infty}(D_i;\xi)$ for the ordinary smooth case and $r_{2,\infty}(d_i;\xi)$ by $\hat{r}_{2,\infty}(D_i;\xi)$ for supersmooth case, respectively.
\end{proof}

\begin{proof}[Proof of Theorem \ref{theorem.unknown under H1}]
    Since the results follow analogously to the proof of Theorem \ref{theorem.known under H1} as stated in Theorem \ref{theorem.unknown under H1n}, the detailed arguments are omitted.
2\end{proof}

\begin{proof}[Proof of Theorem \ref{theorem.known boot}]
    We begin by deriving the following decomposition for the bootstrap version of the test statistic,
    \begin{align}\label{theoremproof.boot statdecomp1}
        S_n^{pro,\ast}(\xi,\theta_n,\sigma^2_n) = S_n^\ast(\xi,\theta_n,\sigma^2_n)-S_n^\ast(0,\theta_n,\sigma^2_n)G_n(\xi)
    \end{align}
    Notice that the decomposition employed in the proof of the previous theorem can still be applied to the bootstrap version of the test statistics,
    \begin{align}\label{theoremproof.boot statdecomp2}
        S_n^\ast(\xi,\theta_n,\sigma^2_n)  =& S_n^\ast(\xi,\theta_0,\sigma^2_0)+S^\ast_{n1}(\xi,\theta_n)-S^\ast_{n2}(\xi,\sigma^2_n)\notag\\
        =&S_n^\ast(\xi,\theta_0,\sigma^2_0)+S^\ast_{n1}(\xi,\theta_n)-\left(S_n(0,\theta_0,\sigma^2_0)+S_{n1}(0,\theta_n)\right)G^\ast_n(\xi).
    \end{align}
    where the last equation follows by \eqref{theoremproof.known ordinary smooth under H0 sigmadecomp} and 
    \begin{align*}
        &S^\ast_{n1}(\xi,\theta_n) = \frac{1}{n}\sum_{i=1}^nV_i\int\left[\left(Y_i-g(x;\theta_n)\right)^2-\left(Y_i-g(x;\theta_0)\right)^2\right]\mathcal{K}_b\left(\frac{x-W_i}{b}\right)\ee^{\ii x\xi}\,dx.
    \end{align*}
    \eqref{theoremproof.boot statdecomp1} and \eqref{theoremproof.boot statdecomp2} provide the decomposition of the bootstrap version of the test statistic after introducing the projection structure,
    \begin{align}\label{theoremproof.known boot statdecomp}
        S_n^{pro,\ast}(\xi,\theta_n,\sigma^2_n) =& S_n^\ast(\xi,\theta_0,\sigma^2_0)-S_n^\ast(0,\theta_0,\sigma^2_0)G_n(\xi)+S^\ast_{n1}(\xi,\theta_n)-S^\ast_{n1}(0,\theta_n)G_n(\xi)\notag\\
        &-\left(S_n(0,\theta_0,\sigma^2_0)+S_{n1}(0,\theta_n)\right)G^\ast_n(\xi)\notag\\
        &+\left(S_n(0,\theta_0,\sigma^2_0)+S_{n1}(0,\theta_n)\right)G_n(\xi)G_n^\ast(0).
    \end{align}
    For the main term, Lemma \ref{lemma.main term known} shows that the convergence order of $S_n^\ast(\xi,\theta_0,\sigma^2_0)$ is the same as that of $S_n(\xi,\theta_0,\sigma^2_0)$. Moreover, as stated in Lemma \ref{lemma.main term known}, the unit variance of the multipliers ensures that the form of the main term is preserved, that is,
    \begin{align}\label{theoremproof.known boot main term ordi}
        \sup_{\xi\in\Pi}\left\vert \sqrt{n}S^\ast_n(\xi,\theta_0,\sigma^2_0) -\frac{1}{\sqrt{n}}\sum_{i=1}^n V_i\left\{r_{1,\infty}(d_i;\xi) - \mathbb{E}\left[r_{1,\infty}(d_i;\xi)\right]\right\}\right\vert = o_p\left(1\right)
    \end{align}
    for the ordinary smooth case and 
    \begin{align}\label{theoremproof.known boot main term super}
        \sup_{\xi\in\Pi}\left\vert \sqrt{n}S^\ast_n(\xi,\theta_0,\sigma^2_0) -\frac{1}{\sqrt{n}}\sum_{i=1}^n V_i\left\{r_{2,\infty}(d_i;\xi) - \mathbb{E}\left[r_{2,\infty}(d_i;\xi)\right]\right\}\right\vert = o_p\left(1\right)
    \end{align}
    for supersmooth case. Note that $\sup_{\xi\in\Pi}|S_{n1}(\xi,\theta_n)|=o_p(n^{-1/2})$ and $\sup_{\xi\in\Pi}|S^\ast_{n1}(\xi,\theta_n)|=o_p(n^{-1/2})$ implied by Lemma \ref{lemma.negligible term known}, along with $\sup_{\xi\in\Pi}|G_n(\xi)-f_X^{\mathrm{ft}}(\xi)|=O_p(n^{-1/2})$ and $\sup_{\xi\in\Pi}|G^\ast_n(\xi)|=O_p(n^{-1/2})$ implied by Lemma \ref{lemma.Gn known}, it further follows that,
    \begin{align}\label{theoremproof.known boot convergence}
        \sup_{\xi\in\Pi}\left\vert S_n^{pro,\ast}(\xi,\theta_n,\sigma^2_n) - S_n^\ast(\xi,\theta_0,\sigma^2_0)+S_n^\ast(0,\theta_0,\sigma^2_0)G_n(\xi) \right\vert=o_p\left(n^{-\frac{1}{2}}\right).
    \end{align}
    Combining \eqref{theoremproof.known boot main term ordi}, \eqref{theoremproof.known boot main term super}, \eqref{theoremproof.known boot convergence} and the unit variance property of multipliers, \eqref{theorem.known boot ordi} and \eqref{theorem.known boot super} hold under the null hypothesis. 

    Subsequently, conclusions about $S_{n1}(\xi,\theta_n)$, $S^\ast_{n1}(\xi,\theta_n)$, $G_n(\xi)$, $G^\ast_n(\xi)$ and \eqref{theoremproof.known boot statdecomp} still hold under the local alternative hypothesis. The only difference from the proof of Theorem \ref{theorem.known under H1n} is that the bias term arising from the local alternative hypothesis is canceled out by the zero-mean multipliers. Consequently, \eqref{theoremproof.known boot main term ordi} and \eqref{theoremproof.known boot main term super} remain valid. Thus, we claim that \eqref{theorem.known boot ordi} and \eqref{theorem.known boot super} still hold under the local alternative hypothesis.
\end{proof}

\begin{proof}[Proof of Theorem \ref{theorem.unknown boot}]
    We first consider the null hypothesis and focus on the ordinary smooth case. By decomposiong
    \begin{align*}
        \hat{S}^\ast_n(\xi,\hat{\theta}_n,\hat{\sigma}_n^2) = \hat{S}^\ast_n(\xi,\theta_0,\sigma_0^2)+\hat{S}^\ast_{n1}(\xi,\hat{\theta}_n)-\hat{S}^\ast_{n2}(\xi,\hat{\sigma}_n^2)
    \end{align*}
    and subsequently
    \begin{align*}
        \hat{S}^\ast_n(\xi,\theta_0,\sigma_0^2) = S_n(\xi,\theta_0,\sigma_0^2)+S_{n,1}^\ast(\xi,\theta_0,\sigma_0^2)+S_{n,2}^\ast(\xi,\theta_0,\sigma_0^2),
    \end{align*}
    where $\hat{S}^\ast_{n1}(\xi,\hat{\theta}_n)$, $\hat{S}^\ast_{n2}(\xi,\hat{\sigma}_n^2)$, $S_{n,1}^\ast(\xi,\theta_0,\sigma_0^2)$ and $S_{n,2}^\ast(\xi,\theta_0,\sigma_0^2)$ denote the counterparts of $\hat{S}_{n1}(\xi,\hat{\theta}_n)$ and $\hat{S}_{n2}(\xi,\hat{\sigma}_n^2)$ as mentioned in proof of Theorem \ref{theorem.unknown ordinary smooth under H0}, $S_{n,1}(\xi,\theta_0,\sigma_0^2)$ and $S_{n,2}(\xi,\theta_0,\sigma_0^2)$ in Lemma \ref{lemma.main term unknown}, respectively, with $\hat{f}^{\mathrm{ft}}_\epsilon(\cdot)$ replaced by the multiplier-perturbed estimator $\hat{f}^{\mathrm{ft},\ast}_\epsilon(\cdot)$. The proof follows the same line as that of Theorem \ref{theorem.unknown ordinary smooth under H0}, with the only difference being that $\hat{r}_{1,\infty}(D_i;\xi)$ is replaced by 
    \begin{align*}
        \hat{r}^\ast_{1,\infty}(D;\xi) =&\ee^{\ii W\xi}\sum_{l=0}^\alpha c_l^{os}(\xi)\left[\left(Y-g(W;\theta_0)\right)^2-\sigma_0^2\right]^{(l)} \\
        &+2V^\ast\left[g^2 f_X\right]^{\mathrm{ft}}(\xi)\Pi_{\epsilon}(\xi)+\frac{V^\ast}{\pi}\int f_X^{\mathrm{ft}}(t)(g^2)^{\mathrm{ft}}(\xi-t)\Pi_{\epsilon}(t)\,dt\\
        &-\frac{V^\ast}{\pi}\int (g f_X)^{\mathrm{ft}}(t)g^{\mathrm{ft}}(\xi-t)\Pi_{\epsilon}(t)\,dt.
    \end{align*}
    Owing to the unit mean, unit variance, and sample independence of $V^\ast$, the resulting limiting process remains $\hat{S}^{os}_\infty(\cdot,\theta_0,\sigma_0)$. An analogous argument applies in the supersmooth case, yielding the same limiting process $\hat{S}^{ss}_\infty(\cdot,\theta_0,\sigma_0)$ as that of the empirical process constructed under the null hypothesis. Under local alternatives, arguments analogous to those in Lemma \ref{lemma.main term unknown} establish the negligibility of $S_{n,2}^\ast(\xi,\theta_0,\sigma_0^2)$, while the convergence of $S_{n,1}^\ast(\xi,\theta_0,\sigma_0^2)$ is obtained by an argument similar to that leading to \eqref{lemmaproof.main term unknown Sn1},
    \begin{align*}
        &\sup_{\xi\in\Pi}\left\vert 
        \begin{aligned}
            &S_{n,1}^\ast(\xi,\theta_0,\sigma^2_0)-\mathbb{E}\left[\left(U^2-\sigma_0^2\right)\ee^{\ii X\xi}\right]K^{\mathrm{ft}}(b\xi)\left[\frac{1}{n}\sum_{i=1}^nV_i^\ast\Pi_{\epsilon,i}(\xi)\right]\\
            &-\left[g^2 f_X\right]^{\mathrm{ft}}(\xi)\left[\frac{1}{n}\sum_{i=1}^nV_i^\ast\Pi_{\epsilon,i}(\xi)\right]\\
            &-\frac{1}{2\pi n}\sum_{i=1}^n\int f_X^{\mathrm{ft}}(t)(g^2)^{\mathrm{ft}}(\xi-t)V_i^\ast\Pi_{\epsilon,i}(t)\,dt\\
            &+\frac{1}{\pi n}\sum_{i=1}^n\int (g f_X)^{\mathrm{ft}}(t)g^{\mathrm{ft}}(\xi-t)V_i^\ast\Pi_{\epsilon,i}(t)\,dt
        \end{aligned}
         \right\vert=o_p(n^{-\frac{1}{2}}).
    \end{align*}
    Note that the local alternative implies $\mathbb{E}[(U^2-\sigma_0^2)\ee^{\ii X\xi}] = n^{-1/2}\mathbb{E}[\Delta(X)\ee^{\ii X\xi}]$. Combining this with $\mathbb{E}[\Pi_{\epsilon}(\xi)]=0$ and the desired properties of the multipliers, we obtain
    \begin{align*}
        \sup_{\xi\in\Pi}\left\vert \mathbb{E}\left[\left(U^2-\sigma_0^2\right)\ee^{\ii X\xi}\right]K^{\mathrm{ft}}(b\xi)\left[\frac{1}{n}\sum_{i=1}^nV_i^\ast\Pi_{\epsilon,i}(\xi)\right]\right\vert = o_p\left(n^{-\frac{1}{2}}\right).
    \end{align*}
    Consequently, $S_{n,1}^\ast(\xi,\theta_0,\sigma_0^2)$ does not affect the deterministic shift function described in Section \ref{sec.Asy}. Combining this with the analysis of $S_n(\xi,\theta_0,\sigma_0^2)$ under local alternatives in Theorem \ref{theorem.known under H1}, we conclude that $\mu(\cdot)$ and $\mathbb{E}[\sqrt{n}\hat{S}^\ast_n(\xi,\hat{\theta}_n,\hat{\sigma}_n^2)]$ are asymptotically equivalent, thereby establishing the stated result of the theorem.
\end{proof}

\section{Lemmas and Proofs}\label{sec.AppendixD}
\begin{lemma}\label{lemma.power of tsf kernel}
Under Assumptions \ref{ass.D} and \ref{ass.D'}, with addition of Assumption \ref{ass.O} or Assumption \ref{ass.S}, we have,
\begin{align} \label{Power of new kernel with ME}
    &\int b^lx^lK_{\epsilon,2}(x)\,dx=0,\,\,\, l < p.
\end{align}
Meanwhile, for the ordinary smooth case, with the addition of Assumption \ref{ass.O'},
\begin{align} \label{Power of abs new kernel with ME}
    &\int \vert b^lx^lK_{\epsilon,2}(x)\vert\,dx=o_p\left(n^{-\frac{1}{2}}\right),\,\,\, l = p, p+1.
\end{align}
\end{lemma}

\begin{proof}[Proof of Lemma \ref{lemma.power of tsf kernel}]
    The proof of this lemma can be found in \cite{song2025specification}.
\end{proof}

\begin{lemma}\label{lemma.power of decon}
Under Assumption \ref{ass.O},
\begin{align} \label{Power of decon without ME Assm O}
    \int \mathcal{K}_\epsilon(x)x^l\ee^{\ii bx\xi}\,dx=I_{\{l\leq\alpha\}}c_l^{os}(\xi)l!b^{-l}+\ii b\sum_{h=0}^\alpha\frac{c_h^{os}(-1)^h\xi^{h+1}}{(h+1)!}\int K^{(h)}(x)x^l\tilde{x}_h^{h+1}\,dx,
\end{align}
holds for $l=0,1,\cdots,p$, where $\bar{x}_h\in(0,x)$ for $h=1,\cdots,\alpha$.

Under Assumption \ref{ass.S},
\begin{align} \label{Power of decon without ME Assm S}
    \int \mathcal{K}_\epsilon(x)x^l\ee^{\ii bx\xi}\,dx=c_l^{ss}(\xi)l!b^{-l}, \qquad \text{for $l\geq0$}.
\end{align}

Suppose that Assumptions \ref{ass.D} and \ref{ass.D'} hold, together with either Assumption \ref{ass.O},
\begin{align} \label{Power of decon with ME}
    \sup\limits_{\xi\in\Pi}\left\vert\int (bx)^l\mathcal{K}_{\epsilon,2}(x)\ee^{\ii bx\xi}\,dx\right\vert=o_p\left(n^{-\frac{1}{2}}\right)\quad \text{for }l<p.
\end{align}
For the ordinary smooth case, with the addition of Assumption \ref{ass.O'},
\begin{align} \label{Power of abs decon with ME}
    \int \left\vert (bx)^l\mathcal{K}_{\epsilon,2}(x)\right\vert \,dx=o_p\left(n^{-\frac{1}{2}}\right), \quad \text{for }l=p,p+1.
\end{align}
For the supersmooth case, under Assumption \ref{ass.S} and \ref{ass.S'},
\begin{align}\label{Power of super kernel}
    \sum_{l=0}^\infty\sup\limits_{\xi\in\Pi}\left\vert\int (bx)^l\mathcal{K}_{\epsilon,2}(x)\ee^{\ii bx\xi}\,dx\right\vert=o_p\left(1\right)\quad.
\end{align}
\end{lemma}

\begin{proof}[Proof of Lemma \ref{lemma.power of decon}]
    The proof of this lemma can be found in \cite{song2025specification}.
\end{proof}

\begin{lemma}\label{lemma.main term known}
    Suppose Assumption \ref{ass.D} holds, together with either Assumption \ref{ass.O} for the ordinary smooth case or Assumption \ref{ass.S} for the supersmooth case,
    \begin{align}\label{lemma.main term known eq1}
        &\sup_{\xi\in\Pi}\left\vert
        \begin{aligned}
            &\frac{1}{n}\sum_{i=1}^n\int\left[\left(Y_i-g(x;\theta_0)\right)^2-\sigma_0^2\right]\mathcal{K}_b\left(\frac{x-W_i}{b}\right)\ee^{\ii x\xi}\,dx\\
            &-\mathbb{E}\left[\left(U^2-\sigma_0^2\right)\ee^{\ii X\xi}\right]K^{\mathrm{ft}}(b\xi)
        \end{aligned}
        \right\vert=O_p\left(n^{-\frac{1}{2}}\right).
    \end{align}
    For the bootstrap version,
    \begin{align}\label{lemma.main term known eq2}
        &\sup_{\xi\in\Pi}\left\vert\frac{1}{n}\sum_{i=1}^nV_i\int\left[\left(Y_i-g(x;\theta_0)\right)^2-\sigma_0^2\right]\mathcal{K}_b\left(\frac{x-W_i}{b}\right)\ee^{\ii x\xi}\,dx\right\vert=O_p\left(n^{-\frac{1}{2}}\right).
    \end{align}
\end{lemma}

\begin{proof}[Proof of Lemma \ref{lemma.main term known}]
    Denote $H(Y,x) = (Y-g(x;\theta_0))^2-\sigma_0^2$. For the ordinary smooth, Assumption \ref{ass.O} implies that function $g$ is $p$-times continuously differentiable, which allows us to expand $H(Y_i,b\tilde{x}+W_i)$ around $W_i$ to the $p$-th order, where $x=b\tilde{x}+W_i$,
    \begin{align}\label{lemmaproof.main term known maintermdecomp}
        &\int\left[\left(Y_i-g(x;\theta_0)\right)^2-\sigma_0^2\right]\mathcal{K}_b\left(\frac{x-W_i}{b}\right)\ee^{\ii  x\xi}\,dx = r_{1,n}(d_i;\xi)+t_{1,n}(d_i;\xi),
    \end{align}
    where
    \begin{align*}
        &r_{1,n}(d_i;\xi) = \ee^{\ii  W_i\xi}\sum_{l=0}^{p-1}H^{(l)}(Y_i,W_i)\frac{b^l}{l!}\int \tilde{x}^l\mathcal{K}_\epsilon(\tilde{x})\ee^{\ii  b\tilde{x}\xi}\,d\tilde{x},\\
        &t_{1,n}(d_i;\xi) = \ee^{\ii  W_i\xi}\frac{b^p}{p!}\int H^{(p)}(Y_i,\tilde{W}_i)\tilde{x}^p\mathcal{K}_b\left(\tilde{x}\right)\ee^{\ii  b\tilde{x}\xi}\,d\tilde{x},
    \end{align*}
    and $\tilde{W}_i$ values between $W_i$ and $bx+W_i$. We claim that the variance of residual term $t_{1,n}$ is negligible, i.e., 
    \begin{align}\label{lemmaproof.main term known negligible residual}
        &\sup_{\xi \in \Pi}\left\vert \frac{1}{n}\sum_{i=1}^n \left\{ t_{1,n}(d_i;\xi) - \mathbb{E}\left[t_{1,n}(d_i;\xi)\right]\right\}\right\vert = o_p\left(n^{-\frac{1}{2}}\right).
    \end{align}
    To show \eqref{lemmaproof.main term known negligible residual}, 
    using Lipschitz continuity of $g$ and $g^2$ mentioned in Assumption \ref{ass.O}, 
    \begin{align*}
        &Var\left(\sup_{\xi \in \Pi}\left\vert \frac{1}{\sqrt{n}}\sum_{i=1}^n  t_{1,n}(d_i;\xi)\right\vert\right) \leq \mathbb{E}\left[\sup_{\xi \in \Pi}t_{1,n}(d_i;\xi)^2\right] \\
        = & O\left(b^{2p} \mathbb{E}\left\{\int \left\vert H^{(p)}(Y,\tilde{W})\right\vert\times\left\vert \tilde{x}^p\mathcal{K}_\epsilon(\tilde{x}) \right\vert\,d\tilde{x}\right\}^2\right)\\
        \leq & O\left(b^{2p} \mathbb{E}\left\{\int \left\vert H^{(p)}(Y,W)\right\vert\times\left\vert \tilde{x}^p\mathcal{K}_\epsilon(\tilde{x}) \right\vert\,d\tilde{x}\right\}^2\right)\\
        & + O\left(b^{2p+2} \mathbb{E}\left\{\int m(Y,W) \times\left\vert \tilde{x}^{p+1}\mathcal{K}_\epsilon(\tilde{x}) \right\vert\,d\tilde{x}\right\}^2\right)\\
        & + O\left(b^{2p+1} \mathbb{E}\left\{\int \left\vert H^{(p)}(Y,W)\right\vert \times m(Y,W) \times\left\vert \tilde{x}^{p+1}\mathcal{K}_\epsilon(\tilde{x}) \right\vert\,d\tilde{x}\right\}\right)\\
        = &O\left(b^{2(p-\alpha)}\right) = o\left(1\right).
    \end{align*}
    where $m(Y,W) = L_{[g^2]^{(p)}}(W)+L_{g^{(p)}}(W)\vert Y\vert$. The existence of the second moment of $m(Y,W)$ is guaranteed by Assumption \ref{ass.D} and \ref{ass.O}. Equation in the last line follows by $\sup_{0\leq l \leq p+1}\int \vert \tilde{x}^l\mathcal{K}(\tilde{x})\vert \,dx = O(b^{-\alpha})$ mentioned in Lemma 4 of \cite{dong2022nonparametric}. For the main term $r_{1,n}$, given \eqref{Power of decon without ME Assm O} in Lemma \ref{lemma.power of decon}, we obtain
    \begin{align*}
        &Var\left(\sup_{\xi \in \Pi}\vert \frac{1}{\sqrt{n}}\sum_{i=1}^n  \left[r_{1,n}(d_i;\xi) - r_{1,\infty}(d_i;\xi)\right]\vert\right) \\
        \leq & \mathbb{E}\left[\sup_{\xi \in \Pi}\left[r_{1,n}(d_i;\xi)-r_{1,\infty}(d_i;\xi)\right]^2\right] = O\left(b^2\right) = o\left(1\right),
    \end{align*}
    where $r_{1,\infty}(d_i;\xi)= \sum_{l=0}^{\alpha}c_l^{os}(\xi)\left[(Y_i-g(W_i;\theta_0))^2-\sigma_0^2\right]^{(l)}\ee^{\ii  W_i\xi}$ and $H^{(l)}(Y_i,x)$ represents for the derivative of $H(Y_i,x)$ with respect to $x$. Therefore,
    \begin{align}\label{lemmaproof.main term known negligible main and residual}
        &\sup_{\xi \in \Pi}\left\vert \frac{1}{n}\sum_{i=1}^n \left\{ r_{1,n}(d_i;\xi) - r_{1,\infty}(d_i;\xi) - \mathbb{E}\left[r_{1,n}(d_i;\xi)- r_{1,\infty}(d_i;\xi)\right]\right\}\right\vert = o_p\left(n^{-\frac{1}{2}}\right),
    \end{align}
    Noting that Assumption \ref{ass.O}(v) implies $\mathbb{E}\left[\int_\Pi r_{1,\infty}(d_i;\xi)^2d\xi\right]<\infty$, we thereby claim that
    \begin{align}\label{lemmaproof.main term known main}
        &\sup_{\xi\in\Pi}\left\vert\frac{1}{n}\sum_{i=1}^n  \left\{r_{1,\infty}(d_i;\xi) - \mathbb{E}\left[r_{1,\infty}(d_i;\xi)\right]\right\}\right\vert = O_p\left(n^{-\frac{1}{2}}\right).
    \end{align}
    This is because for each $\xi$, $c_l^{os}(\xi)$ is a polynomial function of $\xi$. Since the parameter space $\Pi$ is compact, both $\{c_l^{os}(\xi):\xi\in\Pi\}$ and their first derivatives are uniformly bounded over $\Pi$. Moreover, noting that $\vert \ee^{\ii W_i\xi}\vert=1$, differentiation with respect to $\xi$ yields
    \begin{align*}
        \sup_{\xi\in\Pi}\left\vert \partial_{\xi} r_{1,\infty}(d_i;\xi)\right\vert \leq C \left(1 + \vert W_i\vert\right)\sum_{l=0}^{\alpha}\left\vert\left[(Y_i - g(W_i;\theta_0))^2- \sigma_0^2\right]^{(l)}\right\vert
    \end{align*}
    for some finite constant $C>0$. Under the condition given in Assumption \ref{ass.O}(v), the right-hand side is square integrable, which implies that the class $\mathcal{F}=\{r_{1,\infty}(\cdot;\xi):\xi\in\Pi\}$ satisfies a uniform $L_2(\mathbb{P})$-Lipschitz condition in the parameter $\xi$. Since $\Pi$ is a compact set, this implies that $\mathcal{F}$ is a Euclidean class with an integrable envelope. Consequently, by Theorem 2.7.11 of \cite{van1996weak}, the class $\mathcal{F}$ is $\mathbb{P}$-Donsker. Combining \eqref{lemmaproof.main term known maintermdecomp}, \eqref{lemmaproof.main term known negligible main and residual} and \eqref{lemmaproof.main term known main}, 
    \begin{align}\label{lemmaproof.main term known eq1}
        &\sup_{\xi\in\Pi}\left\vert
        \begin{aligned}
            &\frac{1}{n}\sum_{i=1}^n\int\left[\left(Y_i-g(x;\theta_0)\right)^2-\sigma_0^2\right]\mathcal{K}_b\left(\frac{x-W_i}{b}\right)\ee^{\ii x\xi}\,dx\\
            &-\mathbb{E}\left\{\int\left[\left(Y-g(x;\theta_0)\right)^2-\sigma_0^2\right]\mathcal{K}_b\left(\frac{x-W}{b}\right)\ee^{\ii x\xi}\,dx\right\}
        \end{aligned}
        \right\vert=O_p\left(n^{-\frac{1}{2}}\right)
    \end{align}
    holds.

    For the bias term,
    \begin{align}\label{lemmaproof.main term known biastermdecomp}
        &\mathbb{E}\left\{\int\left[\left(Y-g(x;\theta_0)\right)^2-\sigma_0^2\right]\mathcal{K}_b\left(\frac{x-W}{b}\right)\ee^{\ii x\xi}\,dx\right\}\notag\\
        =&\mathbb{E}\left\{\int\left[\left(Y-g(X;\theta_0)\right)^2-\sigma_0^2\right]\mathcal{K}_b\left(\frac{x-W}{b}\right)\ee^{\ii x\xi}\,dx\right\}\notag\\
        &+2\mathbb{E}\left\{\int\left[\left(Y-g(X;\theta_0)\right)\left(g(X;\theta_0)-g(x;\theta_0)\right)\right]\mathcal{K}_b\left(\frac{x-W}{b}\right)\ee^{\ii x\xi}\,dx\right\}\notag\\
        &+\mathbb{E}\left\{\int\left(g(X;\theta_0)-g(x;\theta_0)\right)^2\mathcal{K}_b\left(\frac{x-W}{b}\right)\ee^{\ii x\xi}\,dx\right\}.
    \end{align}
    The first component of the decomposition \eqref{lemmaproof.main term known biastermdecomp} is characterized by equation, 
    \begin{align}\label{lemmaproof.main term known biastermdecomp main}
        &\mathbb{E}\left\{\int\left[\left(Y-g(X;\theta_0)\right)^2-\sigma_0^2\right]\mathcal{K}_b\left(\frac{x-W}{b}\right)\ee^{\ii x\xi}\,dx\right\}\notag\\
        =&\int\mathbb{E}\left[(U^2-\sigma_0^2)K_b\left(\frac{x-X}{b}\right)\right]\ee^{\ii x\xi}\,dx
    \end{align}
    which provides the form of the bias term. For the second component of \eqref{lemmaproof.main term known biastermdecomp}, noting $\mathbb{E}[Y-g(X;\theta_0)\mid X]=\mathbb{E}[U\mid X]=0$ and independence between $\epsilon$ and $X$, it follows that
    \begin{align}\label{lemmaproof.main term known biastermdecomp cross}
        \mathbb{E}\left\{\int\left[\left(Y-g(X;\theta_0)\right)\left(g(X;\theta_0)-g(x;\theta_0)\right)\right]\mathcal{K}_b\left(\frac{x-W}{b}\right)\ee^{\ii x\xi}\,dx\right\}=0.
    \end{align}
    As for the third component of \eqref{lemmaproof.main term known biastermdecomp},
    \begin{align*}
        &\mathbb{E}\left[g^2(X;\theta_0)\mathcal{K}_b\left(\frac{x-W}{b}\right)\right] = \mathbb{E}\left[g^2(X;\theta_0)K_b\left(\frac{x-X}{b}\right)\right] \\
        = &\int g^2(y;\theta_0)K_b\left(\frac{x-y}{b}\right)f_X(y)\,dy=\int g^2(x-by;\theta_0)f_X(x-by)K_b\left(y\right)\,dy\\
        = &g^2(x;\theta_0)f_X(x)+b^p\Delta_1(x).
    \end{align*}
    Denote $\mu_K^p = \int k(u)\left\vert u\right\vert^p du$, along with the Lipschitz continuity mentioned in Assumption \ref{ass.O}, 
    \begin{align*}
        \left\vert\Delta_1(x)\right\vert = \left\vert\frac{1}{p!}\int \left[g^2 f_X\right]^{(p)}(\tilde{y})K(y)y^p\,dy\right\vert\leq \frac{1}{p!}\left\{ \left[g^2 f_X\right]^{(p)}(x)\mu_K^p+ bL_{[g^2f_X]^{(p)}}(x)\mu_K^{p+1}\right\},
    \end{align*}
    where $\tilde{y}$ lies between $x$ and $x-by$. Thus,
    \begin{align*}
        \mathbb{E}\left[g^2(X;\theta_0)\mathcal{K}_b\left(\frac{x-W}{b}\right)\right] = g^2(x;\theta_0)f_X(x)+o\left(n^{-\frac{1}{2}}\right)
    \end{align*}
    is implied by integrability of $L_{[g^2]^{(p)}}(W)$ and the undersmoothing property mentioned in Assumption \ref{ass.O}. By similar arguments,
    \begin{align*}
        \mathbb{E}\left[g(X;\theta_0)\mathcal{K}_b\left(\frac{x-W}{b}\right)\right] = g(x;\theta_0)f_X(x)+o\left(n^{-\frac{1}{2}}\right),\quad \mathbb{E}\left[\mathcal{K}_b\left(\frac{x-W}{b}\right)\right] = f_X(x)+o\left(n^{-\frac{1}{2}}\right).
    \end{align*}
    Consequently, for the third component of \eqref{lemmaproof.main term known biastermdecomp},
    \begin{align}\label{lemmaproof.main term known biastermdecomp qua}
        \left\vert\mathbb{E}\left\{\int\left(g(X;\theta_0)-g(x;\theta_0)\right)^2\mathcal{K}_b\left(\frac{x-W}{b}\right)\ee^{\ii x\xi}\,dx\right\}\right\vert = o\left(n^{-\frac{1}{2}}\right).
    \end{align}
    Combining \eqref{lemmaproof.main term known eq1}, \eqref{lemmaproof.main term known biastermdecomp}, \eqref{lemmaproof.main term known biastermdecomp main},  \eqref{lemmaproof.main term known biastermdecomp cross} and \eqref{lemmaproof.main term known biastermdecomp qua}, \eqref{lemma.main term known eq1} follows for the ordinary smooth case. 

    For the supersmooth case, main terms, $r_{2,\infty}(d_i;\xi) = \sum_{l=0}^{\infty}c_l^{ss}(\xi)\left[(Y_i-g(W_i;\theta_0))^2-\sigma_0^2\right]^{(l)}\ee^{\ii W_i\xi}$, Lemma \ref{lemma.power of decon},
    \begin{align}\label{lemmaproof.main term known biastermdecomp super}
        &\int \left[(Y_i-g(x;\theta_0)^2-\sigma_0^2\right]\mathcal{K}_b\left(\frac{x-W_i}{b}\right)\ee^{\ii x\xi}\,dx =r_{2,n}(d_i; \xi)
    \end{align}
    Note that Assumption \ref{ass.S}(v) implies $\mathbb{E}\left[\int_\Pi r_{2,\infty}(d_i;\xi)^2d\xi\right]<\infty$, by similar arguments to ordinary smooth case,
    \begin{align}\label{lemmaproof.main term known main super}
        &\sup_{\xi\in\Pi}\left\vert\frac{1}{n}\sum_{i=1}^n  \left\{r_{2,\infty}(d_i;\xi) - \mathbb{E}\left[r_{2,\infty}(d_i;\xi)\right]\right\}\right\vert = O_p\left(n^{-\frac{1}{2}}\right).
    \end{align}
    Thus, \eqref{lemmaproof.main term known eq1} is established by combining \eqref{lemmaproof.main term known biastermdecomp super} and \eqref{lemmaproof.main term known main super}.

    We note that
    \begin{align*}
        &\mathbb{E}\left[g^2(X;\theta_0)\mathcal{K}_b\left(\frac{x-W}{b}\right)\right] = \int g^2(y;\theta_0)K_b\left(\frac{x-y}{b}\right)f_X(y)\,dy\\
        =&\int g^2(x-by;\theta_0)f_X(x-by)K_b\left(y\right)\,dy= g^2(x;\theta_0)f_X(x).
    \end{align*}
    By similar arguments, 
    \begin{align*}
        \mathbb{E}\left[g(X;\theta_0)\mathcal{K}_b\left(\frac{x-W}{b}\right)\right] = g(x;\theta_0)f_X(x),\quad \mathbb{E}\left[\mathcal{K}_b\left(\frac{x-W}{b}\right)\right] = f_X(x).
    \end{align*}
    Thus, 
    \begin{align}\label{lemmaproof.main term known biastermdecomp qua super}
        \mathbb{E}\left\{\int\left(g(X;\theta_0)-g(x;\theta_0)\right)^2\mathcal{K}_b\left(\frac{x-W}{b}\right)\ee^{\ii x\xi}\,dx\right\}= 0.
    \end{align}
    Combining \eqref{lemmaproof.main term known biastermdecomp cross}, \eqref{lemmaproof.main term known biastermdecomp main}, \eqref{lemmaproof.main term known biastermdecomp super}, \eqref{lemmaproof.main term known main super} and \eqref{lemmaproof.main term known biastermdecomp qua super}, \eqref{lemma.main term known eq1} follows for the supersmooth case. 
    
    For bootstrap version, notice the independence between $(Y,X,W)$ and $V$, together with unit variance of $V$, we claim our proofs above still hold except
    \begin{align*}
        \mathbb{E}\left\{V\int\left[\left(Y-g(X;\theta_0)\right)^2-\sigma_0^2\right]\mathcal{K}_b\left(\frac{x-W}{b}\right)\ee^{\ii x\xi}\,dx\right\}=0,
    \end{align*}
    which implies
    \begin{align}\label{Cov of Mean under Ass O boot}
        &\sup_{\xi\in\Pi}\left\vert\mathbb{E}\left\{V\int\left[\left(Y-g(x;\theta_0)\right)^2-\sigma_0^2\right]\mathcal{K}_b\left(\frac{x-W}{b}\right)\ee^{\ii x\xi}\,dx\right\}\right\vert=o_p\left(n^{-\frac{1}{2}}\right).
    \end{align}
    Thus, \eqref{lemma.main term known eq2} holds. Given the similarity between the bootstrap version \eqref{lemma.main term known eq2} and equation \eqref{lemma.main term known eq1}, the bootstrap proofs are omitted hereafter unless stated otherwise.
\end{proof}

\begin{lemma}\label{lemma.main term unknown}
    Suppose that Assumption \ref{ass.D} and \ref{ass.D'} hold, along with either Assumption \ref{ass.O} and \ref{ass.O'} for the ordinary smooth case or Assumption \ref{ass.S} and \ref{ass.S'} for the supersmooth case,
    \begin{align}\label{lemma.main term unknown eq1}
        &\sup_{\xi\in\Pi}\left\vert
        \begin{aligned}
            &\frac{1}{n}\sum_{i=1}^n\int\left[\left(Y_i-g(x;\theta_0)\right)^2-\sigma_0^2\right]\hat{\mathcal{K}}_b\left(\frac{x-W_i}{b}\right)\ee^{\ii x\xi}\,dx\\
            &-\mathbb{E}\left[\left(U^2-\sigma_0^2\right)\ee^{\ii X\xi}\right]K^{\mathrm{ft}}(b\xi)
        \end{aligned}
        \right\vert=O_p\left(n^{-\frac{1}{2}}\right).
    \end{align}
    For the bootstrap version with an unknown distribution of measurement errors,
    \begin{align}\label{lemma.main term unknown eq2}
        &\sup_{\xi\in\Pi}\left\vert
        \begin{aligned}
            &\frac{1}{n}\sum_{i=1}^n\int\left[\left(Y_i-g(x;\theta_0)\right)^2-\sigma_0^2\right]\hat{\mathcal{K}}_b^\ast\left(\frac{x-W_i}{b}\right)\ee^{\ii x\xi}\,dx\\
            &-\mathbb{E}\left[\left(U^2-\sigma_0^2\right)\ee^{\ii X\xi}\right]K^{\mathrm{ft}}(b\xi)
        \end{aligned}
        \right\vert=O_p\left(n^{-\frac{1}{2}}\right).
    \end{align}
\end{lemma}

\begin{proof}[Proof of Lemma \ref{lemma.main term unknown}]
    In the absence of information about distribution of measurement error, we start be decomposing $\hat{S}_n(\xi,\theta_0,\sigma^2_0) - S_n(\xi,\theta_0,\sigma^2_0)= S_{n,1}(\xi,\theta_0,\sigma^2_0)+S_{n,2}(\xi,\theta_0,\sigma^2_0)$, where
    \begin{align*}
        &S_{n,1}(\xi,\theta_0,\sigma^2_0) = \frac{1}{n}\sum_{i=1}^n \int \left[\left(Y_i-g(x;\theta_0)\right)^2-\sigma_0^2\right]\mathcal{K}_{b,1}\left(\frac{x-W_i}{b}\right)\ee^{\ii x\xi}\,dx,\notag\\
        &S_{n,2}(\xi,\theta_0,\sigma^2_0) = \frac{1}{n}\sum_{i=1}^n \int \left[\left(Y_i-g(x;\theta_0)\right)^2-\sigma_0^2\right]\mathcal{K}_{b,2}\left(\frac{x-W_i}{b}\right)\ee^{\ii x\xi}\,dx.
    \end{align*}
    For the first term $S_{n,1}$, we define $p(D_i,D_j;\xi)$ as
    \begin{align*}
        p(D_i,D_j;\xi) = &\frac{1}{2\pi b}\iint \left[\left(Y_i-g(x;\theta_0)\right)^2-\sigma_0^2\right]\ee^{-\ii \frac{x-W_i}{b}t}\frac{K^{\mathrm{ft}}(t)}{f^{\mathrm{ft}}_\epsilon(\frac{t}{b})}\Pi_{\epsilon,j}(\frac{t}{b})\ee^{\ii x\xi}\,dx\,dt\\
        =&\frac{1}{2\pi}\iint \left[\left(Y_i-g(x;\theta_0)\right)^2-\sigma_0^2\right]\ee^{-\ii (x-W_i)t}\frac{K^{\mathrm{ft}}(bt)}{f^{\mathrm{ft}}_\epsilon(t)}\Pi_{\epsilon,j}(t)\ee^{\ii x\xi}\,dx\,dt,
    \end{align*}
    where $\Pi_{\epsilon,j}(\cdot)$ is denoted by
    \begin{align}\label{part of psi}
        &\Pi_{\epsilon,j}\left(t\right) = \frac{\left[f_\epsilon^{\mathrm{ft}}(t)\right]^2-\zeta_j(t)}{2\left[f_\epsilon^{\mathrm{ft}}(t)\right]^2}.
    \end{align}
    Then, we can write $S_{n,1}(\xi,\theta_0,\sigma^2_0) = S_{n,11}(\xi,\theta_0,\sigma^2_0)+nS_{n,12}(\xi,\theta_0,\sigma^2_0)/(n-1)$,
    \begin{align*}
        &S_{n,11}(\xi,\theta_0,\sigma^2_0) = \frac{1}{n^2}\sum_{i=1}^np(D_i,D_i;\xi), \quad S_{n,12}(\xi,\theta_0,\sigma^2_0) = \frac{1}{n(n-1)}\sum_{i\neq j}^n p(D_i,D_j;\xi).
    \end{align*}
    By similar arguments to the proof of Lemma \ref{lemma.power of tsf kernel}, we obtain
    \begin{align*}
        \mathbb{E}\left[\frac{1}{2\pi}\int\left\vert \int x^{l}\ee^{-\ii xt}\frac{K^{\mathrm{ft}}(bt)}{f^{\mathrm{ft}}_\epsilon(t)}\Pi_{\epsilon,j}(t)\,dt\right\vert dx\right]=o_p\left(n^{\frac{1}{2}}\right)\quad \text{for }l\leq p+1
    \end{align*}
    when $b^{-3\alpha-2}=o(n^{1/2})$ which holds under Assumption \ref{ass.O'}. Still note $H(Y,x)=(Y-g(x;\theta_0))^2-\sigma_0^2$, along with Assumption \ref{ass.O} and \ref{ass.O'}, we obtain
    \begin{align*}
        &\mathbb{E}\left[\sup_{\xi\in\Pi}\left\vert\frac{1}{2\pi}\iint H(Y_i,x+W_i)\ee^{-\ii xt}\frac{K^{\mathrm{ft}}(bt)}{f^{\mathrm{ft}}_\epsilon(t)}\Pi_{\epsilon,j}(t)e^{\ii x\xi}\,dx\,dt\right\vert\right]\\
        \leq&\mathbb{E}\sup_{\xi\in\Pi}\left\{
        \begin{aligned}
            &\frac{1}{2\pi}\sum_{l=0}^{p-1}\frac{\left\vert H^{(l)}(Y_i,W_i)\right\vert}{l!}\left\vert\iint x^l\ee^{-\ii xt}\frac{K^{\mathrm{ft}}(bt)}{f^{\mathrm{ft}}_\epsilon(t)}\Pi_{\epsilon,j}(t)e^{\ii x\xi}\,dx\,dt\right\vert\\
            &+\frac{1}{2\pi p!}\left\vert\iint H^{(p)}(Y_i,\bar{W})x^l\ee^{-\ii xt}\frac{K^{\mathrm{ft}}(bt)}{f^{\mathrm{ft}}_\epsilon(t)}\Pi_{\epsilon,j}(t)e^{\ii x\xi}\,dx\,dt\right\vert\\
        \end{aligned}
        \right\}\\
        \leq&\mathbb{E}\left\{
        \begin{aligned}
            &\frac{1}{2\pi}\sum_{l=0}^p\frac{\left\vert H^{(l)}(Y_i,W_i)\right\vert}{l!}\int\left\vert \int x^l\ee^{-\ii xt}\frac{K^{\mathrm{ft}}(bt)}{f^{\mathrm{ft}}_\epsilon(t)}\Pi_{\epsilon,j}(t)\,dt\right\vert\,dx \\
            &+ \frac{1}{2\pi p!} \left\vert m(Y_i,W_i)\right\vert\int\left\vert \int x^{p+1}\ee^{-\ii xt}\frac{K^{\mathrm{ft}}(bt)}{f^{\mathrm{ft}}_\epsilon(t)}\Pi_{\epsilon,j}(t)\,dt\right\vert\,dx
        \end{aligned}
        \right\}=o_p\left(n^{\frac{1}{2}}\right),
    \end{align*}
    where $\bar{W}$ lies between $W_i$ and $W_i+x$ and $m(Y,W)$ as mentioned in proof of Lemma \ref{lemma.main term known}. 
    Thus, $S_{n,11}(\xi,\theta_0)=o_p(n^{-1/2})$ holds. By similar arguments, $\mathbb{E}[\sup_{\xi\in\Pi} p^2(D_i,D_i;\xi)]=o_p(n)$ follows by $b^{-6\alpha-4}=o(n)$ which holds under Assumption \ref{ass.O'}.  Notice that $S_{n,12}$ is a second-order U-statistic with symmetric kernel $q(D_i,D_j,\xi) = [p(D_i,D_j,\xi)+ p(D_j,D_i,\xi)]/2$ and H\'{a}jek projection
    \begin{align}\label{lemmaproof.main term unknown hajek proj}
        \hat{S}_{n,12}(\xi,\theta_0,\sigma^2_0) = \mathbb{E}\left[p_1(D_i,\xi)\right]+\frac{2}{n}\sum_{i=1}^n\left\{p_1(D_i,\xi)-\mathbb{E}\left[p_1(D_i,\xi)\right]\right\}.
    \end{align}
    Provided $\mathbb{E}(U\mid X)=0$, we note that
    \begin{align}\label{lemmaproof.main term unknown first term}
        &2p_1(D_i,\xi) = \mathbb{E}\left[2q(D_i,D_j,\xi)\vert D_i\right]\notag\\
        =&\frac{1}{2\pi}\iint \mathbb{E}\left\{\left[\left(Y-g(X;\theta_0)\right)^2-\sigma_0^2\right]\ee^{\ii Wt}\right\}\frac{K^{\mathrm{ft}}(bt)}{f^{\mathrm{ft}}_\epsilon(t)}\Pi_{\epsilon,i}(t)\ee^{\ii x(\xi-t)}\,dx\,dt\notag\\
        &+\frac{1}{2\pi}\iint \mathbb{E}\left[\left(g(X;\theta_0)-g(x;\theta_0)\right)^2\ee^{\ii Xt}\right]K^{\mathrm{ft}}(bt)\Pi_{\epsilon,i}(t)\ee^{\ii x(\xi-t)}\,dx\,dt,\notag\\
        =&\mathbb{E}\left\{\left[\left(Y-g(X;\theta_0)\right)^2-\sigma_0^2\right]\ee^{\ii W\xi}\right\}K^{\mathrm{ft}}(b\xi)\Pi_{\epsilon,i}(\xi)+\left[g^2 f_X\right]^{\mathrm{ft}}(\xi)K^{\mathrm{ft}}(b\xi)\Pi_{\epsilon,i}(\xi)\notag\\
        &+\frac{1}{2\pi}\int \left[f_X^{\mathrm{ft}}(t)(g^2)^{\mathrm{ft}}(\xi-t)-
        2(g f_X)^{\mathrm{ft}}(t)g^{\mathrm{ft}}(\xi-t) 
        \right]K^{\mathrm{ft}}(bt)\Pi_{\epsilon,i}(t)dt,
    \end{align}
    where $(g f_X)^{\mathrm{ft}}(t)$ represents $\int g(x;\theta_0)f_X(x)\ee^{\ii xt}\,dx$. Consequently,
    \begin{align}\label{lemmaproof.main term unknown hajek proj and U}
        \sup_{\xi\in\Pi}\left\vert S_{n,12}(\xi,\theta_0,\sigma^2_0) - \hat{S}_{n,12}(\xi,\theta_0,\sigma^2_0)\right\vert = o_p\left(n^{-\frac{1}{2}}\right).
    \end{align}
    Furthermore, Assumption \ref{ass.O} provides that $K^{\mathrm{ft}}(bt)$ goes to $1$, together with Assumption \ref{ass.O'}, we obtain
    \begin{align}\label{lemmaproof.main term unknown hajek and limit}
        \sup_{\xi\in\Pi}\left\vert 
        \begin{aligned}
            &2p_1(D_i,\xi)\\
            &-\frac{1}{2\pi}\int \left[f_X^{\mathrm{ft}}(t)(g^2)^{\mathrm{ft}}(\xi-t)-
        2(g f_X)^{\mathrm{ft}}(t)g^{\mathrm{ft}}(\xi-t)\right]\Pi_{\epsilon,i}(t)\,dt\\
            &-\left[g^2f_X\right]^{\mathrm{ft}}(\xi)\Pi_{\epsilon,i}(\xi)-\mathbb{E}\left[\left(U^2-\sigma_0^2\right)\ee^{\ii X\xi}\right]\Pi_{\epsilon,i}(\xi)
        \end{aligned}
        \right\vert=o_p\left(1\right).
    \end{align}
    by using the dominant convergence theorem. Combining \eqref{lemmaproof.main term unknown hajek and limit} and Assumption \ref{ass.O'}, we obtain
    \begin{align}\label{lemmaproof.main term unknown hajek convergence}
        \sup_{\xi\in\Pi}\left\vert\frac{1}{n}\sum_{i=1}^n\left\{p_1(D_i,\xi)-\mathbb{E}\left[p_1(D_i,\xi)\right]\right\}\right\vert=O_p(n^{-\frac{1}{2}}).
    \end{align}
    Hence, we claim that $S_{n,12}(\xi,\theta_0,\sigma^2_0)=O_p(n^{-1/2})$ by combining \eqref{lemmaproof.main term unknown hajek proj}, \eqref{lemmaproof.main term unknown hajek proj and U}, \eqref{lemmaproof.main term unknown hajek convergence} and $\mathbb{E}[p_1(D,\xi)]=0$. Along with $S_{n,11}(\xi,\theta_0,\sigma^2_0)=o_p(n^{-1/2})$, we obtain
    \begin{align}\label{lemmaproof.main term unknown Sn1}
        &\sup_{\xi\in\Pi}\left\vert 
        \begin{aligned}
            &S_{n,1}(\xi,\theta_0,\sigma^2_0)-\mathbb{E}\left[\left(U^2-\sigma_0^2\right)\ee^{\ii X\xi}\right]K^{\mathrm{ft}}(b\xi)\left[\frac{1}{n}\sum_{i=1}^n\Pi_{\epsilon,i}(\xi)\right]\\
            &-\left[g^2 f_X\right]^{\mathrm{ft}}(\xi)\left[\frac{1}{n}\sum_{i=1}^n\Pi_{\epsilon,i}(\xi)\right]\\
            &-\frac{1}{2\pi n}\sum_{i=1}^n\int f_X^{\mathrm{ft}}(t)(g^2)^{\mathrm{ft}}(\xi-t)\Pi_{\epsilon,i}(t)\,dt\\
            &+\frac{1}{\pi n}\sum_{i=1}^n\int (g f_X)^{\mathrm{ft}}(t)g^{\mathrm{ft}}(\xi-t)\Pi_{\epsilon,i}(t)\,dt
        \end{aligned}
         \right\vert=o_p(n^{-\frac{1}{2}}).
    \end{align}

    For the second term, 
    \begin{align}\label{lemmaproof.main term unknown Rn123}
        S_{n,2}(\xi,\theta_0,\sigma^2_0) =& \frac{1}{n}\sum_{i=1}^n \ee^{\ii W_i\xi} \int H(Y_i, bx+W_i)\mathcal{K}_{\epsilon,2}(x)\ee^{\ii bx\xi}\,dx\notag\\
        =&R^{os}_{n,1}(\xi,\theta_0,\sigma^2_0) + R^{os}_{n,2}(\xi,\theta_0,\sigma^2_0) + R^{os}_{n,3}(\xi,\theta_0,\sigma^2_0),
    \end{align}
    by Taylor expansion, where
    \begin{align*}
        &R^{os}_{n,1}(\xi,\theta_0,\sigma^2_0) = \sum_{l=0}^{p-1}\frac{b^l}{l!}\left\{
        \begin{aligned}
             &\frac{1}{n}\sum_{i=1}^n\ee^{\ii W_i\xi} H^{(l)}(Y_i,W_i)\\
             &-\mathbb{E}\left[\ee^{\ii W\xi} H^{(l)}(Y,W)\right]
        \end{aligned}
        \right\}\int x^l\mathcal{K}_{\epsilon,2}(x)\ee^{\ii bx\xi}\,dx,\\
        &R^{os}_{n,2}(\xi,\theta_0,\sigma^2_0) = \frac{b^p}{p!}\frac{1}{n}\sum_{i=1}^n\int\left\{
        \begin{aligned}
            &\ee^{\ii W_i\xi}H^{(p)}(Y_i,\bar{W}_i)\\
            -&\mathbb{E}\left[H^{(p)}(Y,\bar{W})\ee^{\ii W\xi}\right]
        \end{aligned}
        \right\}x^p\mathcal{K}_{\epsilon,2}(x)\ee^{\ii bx\xi}\,dx,\\
        &R^{os}_{n,3}(\xi,\theta_0,\sigma^2_0) = \int \mathbb{E}\left[H(Y,x)\mathcal{K}_{b,2}\left(\frac{x-W}{b}\right)\right] \ee^{\ii x\xi}\,dx.
    \end{align*}
    Notice that Assumption \ref{ass.O} implies $\mathbb{E}\vert H^{(l)}(Y,W)\vert^2<\infty$,
    \begin{align}\label{lemmaproof.convergence of main term unknown Rn1 main term}
        &\sup_{\xi\in\Pi}\left\vert\frac{1}{n}\sum_{i=1}^n\ee^{\ii W_i\xi} H^{(l)}(Y_i,W_i) - \mathbb{E}\left[\ee^{\ii W\xi} H^{(l)}(Y,W)\right]\right\vert = O_p\left(n^{-\frac{1}{2}}\right)\quad \text{for }l<p
    \end{align}
    follows by standard entropy-based criteria for Donsker classes (see, e.g., Example 2.10.10 and Theorem 2.5.2 in \cite{van1996weak} which implies that the exponential weight class is $\mathbb{P}$-Donsker. Together with $\int \mathcal{K}_{\epsilon,2}(x)(bx)^l\ee^{\ii bx\xi}\,dx=o_p(1)$ mentioned in Lemma \ref{lemma.power of decon}, we claim that
    \begin{align}\label{lemmaproof.convergence of main term unknown Rn1}
        \sup_{\xi\in\Pi}\left\vert R^{os}_{n,1}(\xi,\theta_0,\sigma^2_0)\right\vert = o_p\left(n^{-\frac{1}{2}}\right).
    \end{align}
    For $R^{os}_{n,2}(\xi,\theta_0,\sigma^2_0)$ in decomposition \eqref{lemmaproof.main term unknown Rn123}, using the Lipschitz continuity mentioned in Assumption \ref{ass.O},
    \begin{align*}
        \sup_{\xi\in\Pi}\vert R^{os}_{n,1}(\xi,\theta_0,\sigma^2_0) \vert \leq&  \frac{1}{p!}\left[\frac{1}{n}\sum_{i=1}^n \left\vert H^{(p)}(Y_i,W_i)\right\vert+\mathbb{E}\left\vert H^{(p)}(Y,W)\right\vert\right]  \int \left\vert (bx)^p\mathcal{K}_{\epsilon,2}(x)\right\vert\,dx \\
        &+ \frac{1}{p!}\left[\frac{1}{n}\sum_{i=1}^n \left\vert m(Y_i,W_i) \right\vert+\mathbb{E}\left\vert m(Y,W) \right\vert\right] \int \left\vert (bx)^{p+1}\mathcal{K}_{\epsilon,2}(x)\right\vert\,dx,
    \end{align*}
    where $m(Y,W) = \vert YL_{g^{(p)}}(W)\vert+\vert L_{[g^2]^{(p)}}(W)\vert$. Notice that Assumption \ref{ass.O} implies $\mathbb{E}\vert H^{(p)}(Y,W)\vert<\infty$ and $\mathbb{E}\vert m(Y,W) \vert<\infty$, along with $\int \vert (bx)^{p+1}\mathcal{K}_{\epsilon,2}(x)\vert\,dx = o_p(n^{-1/2})$ mentioned in Lemma \ref{lemma.power of decon}, we obtain 
    \begin{align}\label{lemmaproof.convergence of main term unknown Rn2}
        \sup_{\xi\in\Pi}\left\vert R^{os}_{n,2}(\xi,\theta_0,\sigma^2_0)\right\vert = o_p\left(n^{-\frac{1}{2}}\right).
    \end{align}
    Subsequently, $\mathbb{E}[U\mid X]=0$ implies 
    \begin{align*}
        \int\mathbb{E}\left\{\left[\left(Y-g(X;\theta_0)\right)\left(g(X;\theta_0)-g(X+bx;\theta_0)\right)\right]\ee^{\ii (X+bx)\xi}\right\} \mathbb{E}\left[K_{\epsilon,2}(x)\right]\,dx=0,
    \end{align*}
    rewriting $R^{os}_{n,3}(\xi,\theta_0,\sigma^2_0)$ as
    \begin{align*}
        &R^{os}_{n,3}(\xi,\theta_0,\sigma^2_0) =  \int \mathbb{E}\left[H(Y,X+bx)\ee^{\ii (X+bx)\xi}\right] \mathbb{E}\left[K_{\epsilon,2}(x)\right]\,dx\\
        = &\int\mathbb{E}\left\{\left[\left(Y-g(X;\theta_0)\right)^2-\sigma_0^2\right]\ee^{\ii (X+bx)\xi}\right\} \mathbb{E}\left[K_{\epsilon,2}(x)\right]\,dx\\
        &+\int \mathbb{E}\left[\left(g(X;\theta_0)-g(X+bx;\theta_0)\right)^2\ee^{\ii (X+bx)\xi}\right] \mathbb{E}\left[K_{\epsilon,2}(x)\right]\,dx\\
        = &\int \mathbb{E}\left\{\left[\left(Y-g(X;\theta_0)\right)^2-\sigma_0^2\right]K_{b,2}\left(\frac{x-X}{b}\right)\right\}\ee^{\ii x\xi}\,dx\\
        &+\iint \left[\left(g(y;\theta_0)-g(x;\theta_0)\right)^2\right]\ee^{\ii x\xi} \mathbb{E}\left[K_{b,2}\left(\frac{x-y}{b}\right)\right]f_X(y)\,dx\,dy.
    \end{align*}
    Let
    \begin{align*}
        &R^{os}_{n,31}(\xi,\theta_0,\sigma^2_0) = \frac{1}{2\pi b}\int \ee^{\ii x\xi} \left[\int g^2(y;\theta_0)K_{\epsilon,2}\left(\frac{x-y}{b}\right)f_X(y)\,dy\right] \,dx,\\
        &R^{os}_{n,32}(\xi,\theta_0,\sigma^2_0) = \frac{1}{2\pi b}\int g^2(x;\theta_0)\ee^{\ii x\xi} \left[\int K_{\epsilon,2}\left(\frac{x-y}{b}\right)f_X(y)\,dy\right] \,dx,\\
        &R^{os}_{n,33}(\xi,\theta_0,\sigma^2_0) = \frac{1}{2\pi b}\int g(x;\theta_0)\ee^{\ii x\xi} \left[\int g(y;\theta_0)K_{\epsilon,2}\left(\frac{x-y}{b}\right)f_X(y)\,dy\right] \,dx.
    \end{align*}
    By similar arguments to the proof of Lemma \ref{lemma.main term known}, 
    \begin{align*}
        &\int g^2(y;\theta_0)K_{\epsilon,2}\left(\frac{x-y}{b}\right)f_X(y)\,dy\\
        = & \sum_{l=0}^{p-1}\frac{b^l}{l!}(g^2 f_X)^{(l)}(x;\theta_0) \int y^lK_{\epsilon,2}(y)\,dy + \frac{b^p}{p!}\int (g^2 f_X)^{(p)}(\tilde{x};\theta_0) y^pK_{\epsilon,2}(y)\,dy\\
        = &\frac{b^p}{p!}\int (g^2 f_X)^{(p)}(\tilde{x};\theta_0) y^pK_{\epsilon,2}(y)\,dy
    \end{align*}
    holds using Taylor expansion, where $\tilde{x}$ is between $x-by$ and $x$. We note that the second equation follows by $\int y^lK_{\epsilon,2}(y)\,dy = 0$ for $l<p$, which is mentioned in Lemma \ref{lemma.power of tsf kernel}. Under Lipschitz continuity mentioned in Assumption \ref{ass.O}, 
    \begin{align*}
        &\left\vert \int g^2(y;\theta_0)K_{\epsilon,2}\left(\frac{x-y}{b}\right)f_X(y)\,dy \right\vert \\
        \leq & \frac{b^p}{p!}\left\vert (g^2 f_X)^{(p)}(x;\theta_0) \right\vert \int \left\vert y^pK_{\epsilon,2}(y) \right\vert \,dy + \frac{b^{p+1}}{p!}\left\vert L_{[g^2 f_X]^{(p)}}(x) \right\vert \int \left\vert y^{p+1}K_{\epsilon,2}(y) \right\vert \,dy.
    \end{align*}
    Subsequently, we claim $R^{os}_{n,31}(\xi,\theta_0,\sigma^2_0) = o_p(n^{-1/2})$ which follows by conclusions in Lemma \ref{lemma.power of tsf kernel}. By a similar arguments, $R^{os}_{n,32}(\xi,\theta_0,\sigma^2_0) = o_p(n^{-1/2})$ and $R^{os}_{n,33}(\xi,\theta_0,\sigma^2_0) = o_p(n^{-1/2})$. Consequently, 
    \begin{align}\label{lemmaproof.convergence of main term unknown Rn3}
        \sup_{\xi\in\Pi}\left\vert R^{os}_{n,3}(\xi,\theta_0,\sigma^2_0)-\int \mathbb{E}\left[\left(U^2-\sigma_0^2\right)K_{b,2}\left(\frac{x-X}{b}\right)\right]\ee^{\ii x\xi}\,dx\right\vert = o_p\left(n^{-\frac{1}{2}}\right).
    \end{align}
    Thus, \eqref{lemma.main term unknown eq1} holds for the ordinary smooth case by combining the definition of $\psi_2$, \eqref{lemma.main term known eq1}, \eqref{lemmaproof.main term unknown Sn1}, \eqref{lemmaproof.main term unknown Rn123}, \eqref{lemmaproof.convergence of main term unknown Rn1}, \eqref{lemmaproof.convergence of main term unknown Rn2} and \eqref{lemmaproof.convergence of main term unknown Rn3}.

    For the supersmooth case, \eqref{lemma.main term known eq1} and decomposition of $\hat{S}_n(\xi,\theta_0,\sigma^2_0) - S_n(\xi,\theta_0,\sigma^2_0)$ still hold. For the first term of decomposition, under Assumption \ref{ass.S'}, following the proof of Lemma \ref{lemma.power of tsf kernel},
    \begin{align*}
        &\mathbb{E}\left[\sup_{\xi\in\Pi}\left\vert\frac{1}{2\pi}\iint H(Y_i,x+W_i)\ee^{-\ii xt}\frac{K^{\mathrm{ft}}(bt)}{f^{\mathrm{ft}}_\epsilon(t)}\Pi_{\epsilon,j}(t)e^{\ii x\xi}\,dx\,dt\right\vert\right]\\
        \leq&
        \frac{1}{2\pi}\sum_{l=0}^{\infty}\frac{1}{l!}\mathbb{E}\left\vert H^{(l)}(Y_i,W_i)\right\vert\mathbb{E}\sup_{\xi\in\Pi}\left\vert\iint x^l\ee^{-\ii xt}\frac{K^{\mathrm{ft}}(bt)}{f^{\mathrm{ft}}_\epsilon(t)}\Pi_{\epsilon,j}(t)e^{\ii x\xi}\,dx\,dt\right\vert\\
        =&O_p\left(\ee^{-3\mu(1+b^{-1})^2}\right)=o_p\left(n^{\frac{1}{2}}\right),
    \end{align*}
     where the last equation follows by Assumption \ref{ass.S'}. Thus, by similar arguments to ordinary smooth case, $S_{n,11}(\xi,\theta_0,\sigma^2_0)=o_p(n^{-1/2})$ and $\mathbb{E}[\sup_{\xi\in\Pi} p^2(D_i,D_i;\xi)]=o_p(n)$ hold. Subsequently, \eqref{lemmaproof.main term unknown Sn1} follows by \eqref{lemmaproof.main term unknown hajek proj}--\eqref{lemmaproof.main term unknown hajek convergence}. For the second term of decomposition, we decompose
     \begin{align}\label{lemmaproof.main term unknown Rn12 super}
        S_{n,2}(\xi,\theta_0,\sigma^2_0) =&R^{ss}_{n,1}(\xi,\theta_0,\sigma^2_0) + R^{ss}_{n,2}(\xi,\theta_0,\sigma^2_0),
    \end{align}
    where
    \begin{align*}
        &R^{ss}_{n,1}(\xi,\theta_0,\sigma^2_0) = \sum_{l=0}^{\infty}\frac{b^l}{l!}\left\{
        \begin{aligned}
             &\frac{1}{n}\sum_{i=1}^n\ee^{\ii W_i\xi} H^{(l)}(Y_i,W_i)\\
             &-\mathbb{E}\left[\ee^{\ii W\xi} H^{(l)}(Y,W)\right]
        \end{aligned}
        \right\}\int x^l\mathcal{K}_{\epsilon,2}(x)\ee^{\ii bx\xi}\,dx,\\
        &R^{ss}_{n,2}(\xi,\theta_0,\sigma^2_0) = \int \mathbb{E}\left[H(Y,x)\mathcal{K}_{b,2}\left(\frac{x-W}{b}\right)\right] \ee^{\ii x\xi}\,dx.
    \end{align*}
    Under Assumption \ref{ass.S'}, \eqref{lemmaproof.convergence of main term unknown Rn1 main term} still holds and therefore,
    \begin{align}\label{lemmaproof.convergence of main term unknown Rn1 super}
        \sup_{\xi\in\Pi}\left\vert R^{ss}_{n,1}(\xi,\theta_0,\sigma^2_0)\right\vert = o_p\left(n^{-\frac{1}{2}}\right)
    \end{align}
    follows by conclusion \eqref{Power of super kernel} mentioned in Lemma \ref{lemma.power of decon}. We notice that $R^{ss}_{n,2}(\xi,\theta_0,\sigma^2_0) = R^{os}_{n,31}(\xi,\theta_0,\sigma^2_0)+R^{os}_{n,32}(\xi,\theta_0,\sigma^2_0)$ still holds, together with
    \begin{align*}
        &\int g^2(y;\theta_0)K_{\epsilon,2}\left(\frac{x-y}{b}\right)f_X(y)\,dy= \sum_{l=0}^{\infty}\frac{b^l}{l!}(g^2 f_X)^{(l)}(x;\theta_0) \int y^lK_{\epsilon,2}(y)\,dy=0
    \end{align*}
    and by similar arguments,
    \begin{align*}
        &\int K_{\epsilon,2}\left(\frac{x-y}{b}\right)f_X(y)\,dy=0,\quad \int g(y;\theta_0)K_{\epsilon,2}\left(\frac{x-y}{b}\right)f_X(y)\,dy=0,
    \end{align*}
    we claim 
    \begin{align}\label{lemmaproof.main term unknown Sn2}
         \sup_{\xi\in\Pi}\left\vert S_{n,2}(\xi,\theta_0,\sigma^2_0)-\int \mathbb{E}\left[\left(U^2-\sigma_0^2\right)K_{b,2}\left(\frac{x-X}{b}\right)\right]\ee^{\ii x\xi}\,dx\right\vert = o_p\left(n^{-\frac{1}{2}}\right).
    \end{align}
    Consequently, \eqref{lemma.main term unknown eq1} follows by \eqref{lemmaproof.main term unknown Sn1} and \eqref{lemmaproof.main term unknown Sn2}. Owing to the fact that the multipliers are defined to have unit mean and unit variance and to be independent of all random variables involved, replacing $f^{\mathrm{ft}}_\epsilon(\cdot)$ by $\hat{f}^{\mathrm{ft}}_\epsilon(\cdot)$ in all quantities defined in the above proof (including $\mathcal{K}_{\epsilon,1}(\cdot)$ and $\mathcal{K}_{\epsilon,2}(\cdot)$) does not affect their means or variances. Consequently, \eqref{lemma.main term unknown eq2} can be proved in a completely analogous manner, with $\Pi_{\epsilon,j}(\cdot)$ in \eqref{part of psi} replaced by its multiplier-based counterpart
    \begin{align*}
        &\Pi^\ast_{\epsilon,j}\left(t\right) = \frac{\left[f_\epsilon^{\mathrm{ft}}(t)\right]^2-\zeta^\ast_j(t)}{2\left[f_\epsilon^{\mathrm{ft}}(t)\right]^2}, \quad \zeta^\ast_j(t) = V_i^\ast\cos\left[t(W_j-W_j^r)\right].
    \end{align*}
\end{proof}

\begin{lemma}\label{lemma.negligible term known}
    Suppose Assumption \ref{ass.D} holds, together with either Assumption \ref{ass.O} for the ordinary smooth case or Assumption \ref{ass.S} for the supersmooth case,
    \begin{align}\label{lemma.negligible term known eq1}
        &\sup_{\xi\in\Pi}\left\vert
        \frac{1}{n}\sum_{i=1}^n\int\left[\left(Y_i-g(x;\theta_n)\right)^2-\left(Y_i-g(x;\theta_0)\right)^2\right]\mathcal{K}_b\left(\frac{x-W_i}{b}\right)\ee^{\ii x\xi}\,dx
        \right\vert=o_p\left(n^{-\frac{1}{2}}\right).
    \end{align}
    For the bootstrap version,
    \begin{align}\label{lemma.negligible term known eq2}
        &\sup_{\xi\in\Pi}\left\vert
        \frac{1}{n}\sum_{i=1}^nV_i\int\left[\left(Y_i-g(x;\theta_n)\right)^2-\left(Y_i-g(x;\theta_0)\right)^2\right]\mathcal{K}_b\left(\frac{x-W_i}{b}\right)\ee^{\ii x\xi}\,dx
        \right\vert=o_p\left(n^{-\frac{1}{2}}\right).
    \end{align}
\end{lemma}

\begin{proof}[Proof of Lemma \ref{lemma.negligible term known}]
    This lemma shows that the estimation of the parameter $\theta_0$ does not introduce \textquotedblleft parameter estimation uncertainty\textquotedblright. To show \eqref{lemma.negligible term known eq1}, we begin by denoting
    \begin{align*}
        &R_{n,1}(\xi,\theta_0) = \frac{1}{n}\sum_{i=1}^n\int\left(g(x;\theta_n)-g(x;\theta_0)\right)^2\mathcal{K}_b\left(\frac{x-W_i}{b}\right)\ee^{\ii x\xi}\,dx,\\
        &R_{n,2}(\xi,\theta_0) = \frac{1}{n}\sum_{i=1}^n\int\left[\left(Y_i-g(x;\theta_0)\right)\left(g(x;\theta_n)-g(x;\theta_0)\right)\right]\mathcal{K}_b\left(\frac{x-W_i}{b}\right)\ee^{\ii x\xi}\,dx.
    \end{align*}
    It then suffices to show
    \begin{align}\label{lemma.negligible term known eq1eq}
        \sup_{\xi\in\Pi}\left\vert R_{n,1}(\xi,\theta_0)-2R_{n,2}(\xi,\theta_0)\right\vert=o_p\left(n^{-\frac{1}{2}}\right). 
    \end{align}
    To investigate the effect of estimating $\theta_0$, we expand the function $g(x;\theta_n)$ at $\theta_0$ as follows,
    \begin{align}\label{lemmaproof.negligible term known Taylor}
        &g(x;\theta_n)-g(x;\theta_0) = \left(\theta_n-\theta_0\right)' \frac{\partial g(x;\tilde{\theta})}{\partial\theta},
    \end{align}
    where $\tilde{\theta}$ values between $\theta_0$ and $\theta_n$. Thus, for the first term $R_{n,1}$,
    \begin{align}\label{lemmaproof.negligible term known decom}
        &R_{n,1}(\xi,\theta_0) = \left(\theta_n-\theta_0\right)'\left[\frac{1}{n}\sum_{i=1}^n \int\frac{\partial g(x;\tilde{\theta})}{\partial\theta}\frac{\partial g(x;\tilde{\theta})}{\partial\theta'}\mathcal{K}_b\left(\frac{x-W_i}{b}\right)\ee^{\ii x\xi}\,dx\right]\left(\theta_n-\theta_0\right).
    \end{align}
    Under Assumption \ref{ass.O}, by similar arguments to the proof in Lemma \ref{lemma.main term known},
    \begin{align}\label{lemmaproof.negligible term known Rn1 main}
        \sup_{\xi\in\Pi}\left\vert
        \begin{aligned}
            &\frac{1}{n}\sum_{i=1}^n \int\frac{\partial g(x;\tilde{\theta})}{\partial\theta}\frac{\partial g(x;\tilde{\theta})}{\partial\theta'}\mathcal{K}_b\left(\frac{x-W_i}{b}\right)\ee^{\ii x\xi}\,dx\\
            &-\mathbb{E}\left[\int\frac{\partial g(x;\tilde{\theta})}{\partial\theta}\frac{\partial g(x;\tilde{\theta})}{\partial\theta'}\mathcal{K}_b\left(\frac{x-W}{b}\right)\ee^{\ii x\xi}\,dx\right]
        \end{aligned}
        \right\vert=O_p\left(n^{-\frac{1}{2}}\right).
    \end{align}
    Notice that $\theta_n\xrightarrow{P}\theta_0$, along with properties of kernel $K$ and integrability of $\partial g(x;\theta)/\partial\theta$ mentioned in Assumption \ref{ass.O} implies
    \begin{align}\label{lemmaproof.negligible term known Rn1 exp}
        &\sup_{\xi\in\Pi}\left\vert\mathbb{E}\left[\int\frac{\partial g(x;\tilde{\theta})}{\partial\theta}\frac{\partial g(x;\tilde{\theta})}{\partial\theta'}\mathcal{K}_b\left(\frac{x-W}{b}\right)\ee^{\ii x\xi}\,dx\right]\right\vert=O_p\left(1\right).
    \end{align}
    Combining \eqref{lemmaproof.negligible term known decom}, \eqref{lemmaproof.negligible term known Rn1 main} and \eqref{lemmaproof.negligible term known Rn1 exp}, it follows that
    \begin{align}\label{lemmaproof.negligible term known decom1}
        \sup_{\xi\in\Pi}\left\vert R_{n,1}(\xi,\theta_0)\right\vert=O_p\left(\left\vert\theta_n-\theta_0\right\vert^2\right).
    \end{align}
    Subsequently, for the second term $R_{n,2}$, the bias term can also be decomposed as follows,
    \begin{align}\label{lemmaproof.negligible term known decom2}
        R_{n,2}(\xi,\theta_0) = R_{n,21}(\xi,\theta_0)+R_{n,22}(\xi,\theta_0)
    \end{align}
    where
    \begin{align*}
        &R_{n,21}(\xi,\theta_0) = \frac{1}{n}\sum_{i=1}^n\left(Y_i-g(X_i;\theta_0)\right)\int\left(g(x_i;\theta_n)-g(x;\theta_0)\right)\mathcal{K}_b\left(\frac{x-W_i}{b}\right)\ee^{\ii x\xi}\,dx,\\
        &R_{n,22}(\xi,\theta_0) = \frac{1}{n}\sum_{i=1}^n\int\left[\left(g(X_i;\theta_0)-g(x;\theta_0)\right)\left(g(x;\theta_n)-g(x;\theta_0)\right)\right]\mathcal{K}_b\left(\frac{x-W_i}{b}\right)\ee^{\ii x\xi}\,dx
    \end{align*}
    By reapplying equation \eqref{lemmaproof.negligible term known Taylor} and simplifying $R_{n,21}$ and $R_{n,22}$,
    \begin{align*}
        &R_{n,21}(\xi,\theta_0) = \left(\theta_n-\theta_0\right)'\left[\frac{1}{n}\sum_{i=1}^n U_i\int\frac{\partial g(x;\tilde{\theta})}{\partial\theta}\mathcal{K}_b\left(\frac{x-W_i}{b}\right)\ee^{\ii x\xi}\,dx\right],\\
        &R_{n,22}(\xi,\theta_0) = \left(\theta_n-\theta_0\right)'\left[\frac{1}{n}\sum_{i=1}^n\int\left(g(X_i;\theta_0)-g(x;\theta_0)\right)\frac{\partial g(x;\tilde{\theta})}{\partial\theta}\mathcal{K}_b\left(\frac{x-W_i}{b}\right)\ee^{\ii x\xi}\,dx\right].
    \end{align*}
    Still under Assumption \ref{ass.O}, by similar arguments to the proof of Lemma \ref{lemma.main term known},
    \begin{align}
        \sup_{\xi\in\Pi}\left\vert
        \begin{aligned}
            &\frac{1}{n}\sum_{i=1}^n U_i\int\frac{\partial g(x;\tilde{\theta})}{\partial\theta}\mathcal{K}_b\left(\frac{x-W_i}{b}\right)\ee^{\ii x\xi}\,dx\\
            &-\mathbb{E}\left[U\int\frac{\partial g(x;\tilde{\theta})}{\partial\theta}\mathcal{K}_b\left(\frac{x-W}{b}\right)\ee^{\ii x\xi}\,dx\right]
        \end{aligned}
        \right\vert=O_p\left(n^{-\frac{1}{2}}\right).
    \end{align}
    Notice that Assumption \ref{ass.D} implies the independence between $\epsilon$ and $X$, along with $\mathbb{E}(U\mid X)=0$, we obtain,
    \begin{align*}
        \mathbb{E}\left[U\int\frac{\partial g(x;\tilde{\theta})}{\partial\theta}\mathcal{K}_b\left(\frac{x-W}{b}\right)\ee^{\ii x\xi}\,dx\right]=0
    \end{align*}
    Thus, $\sup_{\xi\in\Pi}\vert R_{n,21}(\xi,\theta_0)\vert=o_p(n^{-1/2})$ holds. Meanwhile, $\sup_{\xi\in\Pi}\vert R_{n,22}(\xi,\theta_0)\vert=o_p(n^{-1/2})$ holds by similar arguments to the proof of Lemma \ref{lemma.main term known},
    \begin{align}
        \sup_{\xi\in\Pi}\left\vert
        \begin{aligned}
            &\frac{1}{n}\sum_{i=1}^n\int\left(g(X_i;\theta_0)-g(x;\theta_0)\right)\frac{\partial g(x;\tilde{\theta})}{\partial\theta}\mathcal{K}_b\left(\frac{x-W_i}{b}\right)\ee^{\ii x\xi}\,dx\\
            &-\mathbb{E}\left[\int\left(g(X;\theta_0)-g(x;\theta_0)\right)\frac{\partial g(x;\tilde{\theta})}{\partial\theta}\mathcal{K}_b\left(\frac{x-W}{b}\right)\ee^{\ii x\xi}\,dx\right]
        \end{aligned}
        \right\vert=O_p\left(n^{-\frac{1}{2}}\right)
    \end{align}
    and
    \begin{align*}
        \sup_{\xi\in\Pi}\left\vert\mathbb{E}\left[\int\left(g(X;\theta_0)-g(x;\theta_0)\right)\frac{\partial g(x;\tilde{\theta})}{\partial\theta}\mathcal{K}_b\left(\frac{x-W}{b}\right)\ee^{\ii x\xi}\,dx\right]\right\vert=o_p\left(n^{-\frac{1}{2}}\right).
    \end{align*}
    Subsequently, \eqref{lemmaproof.negligible term known decom2} follows.
\end{proof}

\begin{lemma}\label{lemma.negligible term unknown}
    Suppose that Assumption \ref{ass.D} and \ref{ass.D'} hold, along with either Assumption \ref{ass.O} and \ref{ass.O'} for the ordinary smooth case or Assumption \ref{ass.S} and \ref{ass.S'} for the supersmooth case,
    \begin{align}\label{lemma.negligible term unknown eq1}
        &\sup_{\xi\in\Pi}\left\vert
        \frac{1}{n}\sum_{i=1}^n\int\left[\left(Y_i-g(x;\hat{\theta}_n)\right)^2-\left(Y_i-g(x;\theta_0)\right)^2\right]\hat{\mathcal{K}}_b\left(\frac{x-W_i}{b}\right)\ee^{\ii x\xi}\,dx
        \right\vert=o_p\left(n^{-\frac{1}{2}}\right).
    \end{align}
    and 
    \begin{align}\label{lemma.negligible term unknown eq2}
        &\sup_{\xi\in\Pi}\left\vert
        \frac{1}{n}\sum_{i=1}^n\int\left[\left(Y_i-g(x;\hat{\theta}_n)\right)^2-\left(Y_i-g(x;\theta_0)\right)^2\right]\hat{\mathcal{K}}^\ast_b\left(\frac{x-W_i}{b}\right)\ee^{\ii x\xi}\,dx
        \right\vert=o_p\left(n^{-\frac{1}{2}}\right).
    \end{align}
\end{lemma}

\begin{proof}[Proof of Lemma \ref{lemma.negligible term unknown}]
    The proof of this lemma follows along the same line as Lemma \ref{lemma.negligible term known}, except that it mimics the proof of Lemma \ref{lemma.main term unknown} but not Lemma \ref{lemma.main term known}. We start be rewriting \eqref{lemma.negligible term unknown eq1} into a form that shows it is sufficient to establish $\sup_{\xi\in\Pi}\vert \hat{R}_{n,1}(\xi,\theta_0)-2\hat{R}_{n,2}(\xi,\theta_0)\vert=o_p(n^{-1/2})$
    where
    \begin{align*}
        &\hat{R}_{n,1}(\xi,\theta_0) = \frac{1}{n}\sum_{i=1}^n\int\left(g(x;\hat{\theta}_n)-g(x;\theta_0)\right)^2\hat{\mathcal{K}}_b\left(\frac{x-W_i}{b}\right)\ee^{\ii x\xi}\,dx,\\
        &\hat{R}_{n,2}(\xi,\theta_0) = \frac{1}{n}\sum_{i=1}^n\int\left[\left(Y_i-g(x;\theta_0)\right)\left(g(x;\hat{\theta}_n)-g(x;\theta_0)\right)\right]\hat{\mathcal{K}}_b\left(\frac{x-W_i}{b}\right)\ee^{\ii x\xi}\,dx.
    \end{align*}
    Then we use the Taylor expansion similar to \eqref{lemmaproof.negligible term known Taylor}, and gives an upper bound restriction on the two terms similar to Lemma \ref{lemma.main term unknown}. \eqref{lemma.negligible term unknown eq1} holds. \eqref{lemma.negligible term unknown eq2} can be proved by analogous arguments.
\end{proof}

\begin{lemma}\label{lemma.Gn known}
    Suppose Assumption \ref{ass.D} holds, together with either Assumption \ref{ass.O} for the ordinary smooth case or Assumption \ref{ass.S} for the supersmooth case,
    \begin{align}\label{lemma.Gn known eq1}
        &\sup_{\xi\in\Pi}\left\vert
        \frac{1}{n}\sum_{i=1}^n\int\mathcal{K}_b\left(\frac{x-W_i}{b}\right)\ee^{\ii x\xi}\,dx-f^{\mathrm{ft}}_X(\xi)
        \right\vert=O_p\left(n^{-\frac{1}{2}}\right).
    \end{align}
\end{lemma}

\begin{proof}[Proof of Lemma \ref{lemma.Gn known}]
    We first note that the left-hand side of \eqref{lemma.Gn known eq1}, which we aim to prove, can be rewritten as follows,
    \begin{align*}
        \frac{1}{n}\sum_{i=1}^n\int\mathcal{K}_b\left(\frac{x-W_i}{b}\right)\ee^{\ii x\xi}\,dx = \frac{K^{\mathrm{ft}}(b\xi)}{f^{\mathrm{ft}}_\epsilon(\xi)}\left(\frac{1}{n}\sum_{i=1}^n\ee^{\ii W_i\xi}\right),
    \end{align*}
    For the bias term, note that Assumption \ref{ass.O} implies the boundedness of $(K^{\mathrm{ft}})^{(p)}(t)$ when $t$ values between $0$ and $b\xi$, and $(K^{\mathrm{ft}})^{(l)}(0)=0$ for $1\leq l<p$,
    \begin{align}\label{lemmaproof.Gn known bias}
        &\frac{K^{\mathrm{ft}}(b\xi)}{f^{\mathrm{ft}}_\epsilon(\xi)}\mathbb{E}\left(\frac{1}{n}\sum_{i=1}^n\ee^{\ii W_i\xi}\right) = K^{\mathrm{ft}}(b\xi)f^{\mathrm{ft}}_X(\xi) = f_X^{\mathrm{ft}}(\xi)+O_p\left(b^p\right).
    \end{align}
    For the main term, noting that a compact set $\Pi$ implies $\mathbb{E}[\int_\Pi\vert\ee^{\ii W\xi}\vert^2\,d\xi]<\infty$, and invoking Theorem 3.9 of \cite{chen1998central} and Lemma \ref{lemma.power of tsf kernel},
    \begin{align*}
        &\sup_{\xi\in\Pi}\left\vert\frac{1}{n}\sum_{i=1}^n\ee^{\ii W_i\xi}-\mathbb{E}\left[\ee^{\ii W\xi}\right]\right\vert=O_p\left(n^{-\frac{1}{2}}\right), \quad \sup_{\xi\in\Pi}\left\vert \frac{K^{\mathrm{ft}}(b\xi)}{f^{\mathrm{ft}}_\epsilon(\xi)}-c_0^{os}(\xi) \right\vert=o_p\left(1\right).
    \end{align*}
    Thus,
    \begin{align}\label{lemmaproof.Gn known main}
        &\sup_{\xi\in\Pi}\left\vert
        \frac{1}{n}\sum_{i=1}^n\int\mathcal{K}_b\left(\frac{x-W_i}{b}\right)\ee^{\ii x\xi}\,dx-\frac{K^{\mathrm{ft}}(b\xi)}{f^{\mathrm{ft}}_\epsilon(\xi)}\mathbb{E}\left(\frac{1}{n}\sum_{i=1}^n\ee^{\ii W_i\xi}\right)
        \right\vert=O_p\left(n^{-\frac{1}{2}}\right).
    \end{align}
    Combining \eqref{lemmaproof.Gn known bias} and \eqref{lemmaproof.Gn known main}, we obtain \eqref{lemma.Gn known eq1}.
\end{proof}

\begin{lemma}\label{lemma.Gn unknown}
    Suppose that Assumption \ref{ass.D} and \ref{ass.D'} hold, along with either Assumption \ref{ass.O} and \ref{ass.O'} for the ordinary smooth case or Assumption \ref{ass.S} and \ref{ass.S'} for the supersmooth case,
    \begin{align}\label{lemma.Gn unknown eq1}
        &\sup_{\xi\in\Pi}\left\vert
        \frac{1}{n}\sum_{i=1}^n\int\hat{\mathcal{K}}_b\left(\frac{x-W_i}{b}\right)\ee^{\ii x\xi}\,dx-\mathbb{E}\left[\ee^{\ii X\xi}\left(1+\psi_2(\xi)\right)\right]
        \right\vert=O_p\left(n^{-\frac{1}{2}}\right).
    \end{align}
\end{lemma}

\begin{proof}[Proof of Lemma \ref{lemma.Gn unknown}]
    The proof is identical to that of Lemma \ref{lemma.main term unknown} except for using Lemma \ref{lemma.Gn known} instead of Lemma \ref{lemma.main term known} and is therefore omitted.
\end{proof}
\end{document}